\PassOptionsToPackage{unicode}{hyperref}
\PassOptionsToPackage{hyphens}{url}
\PassOptionsToPackage{dvipsnames,svgnames,x11names}{xcolor}
\documentclass[
  12pt]{article}

\usepackage{amsmath,amssymb}
\usepackage{iftex}
\ifPDFTeX
  \usepackage[T1]{fontenc}
  \usepackage[utf8]{inputenc}
  \usepackage{textcomp} 
\else 
  \usepackage{unicode-math}
  \defaultfontfeatures{Scale=MatchLowercase}
  \defaultfontfeatures[\rmfamily]{Ligatures=TeX,Scale=1}
\fi
\usepackage{lmodern}
\ifPDFTeX\else  
\fi
\IfFileExists{upquote.sty}{\usepackage{upquote}}{}
\IfFileExists{microtype.sty}{
  \usepackage[]{microtype}
  \UseMicrotypeSet[protrusion]{basicmath} 
}{}
\makeatletter
\@ifundefined{KOMAClassName}{
  \IfFileExists{parskip.sty}{%
    \usepackage{parskip}
  }{
    \setlength{\parindent}{0pt}
    \setlength{\parskip}{6pt plus 2pt minus 1pt}}
}{
  \KOMAoptions{parskip=half}}
\makeatother
\usepackage{xcolor}
\makeatletter
\ifx\paragraph\undefined\else
  \let\oldparagraph\paragraph
  \renewcommand{\paragraph}{
    \@ifstar
      \xxxParagraphStar
      \xxxParagraphNoStar
  }
  \newcommand{\xxxParagraphStar}[1]{\oldparagraph*{#1}\mbox{}}
  \newcommand{\xxxParagraphNoStar}[1]{\oldparagraph{#1}\mbox{}}
\fi
\ifx\subparagraph\undefined\else
  \let\oldsubparagraph\subparagraph
  \renewcommand{\subparagraph}{
    \@ifstar
      \xxxSubParagraphStar
      \xxxSubParagraphNoStar
  }
  \newcommand{\xxxSubParagraphStar}[1]{\oldsubparagraph*{#1}\mbox{}}
  \newcommand{\xxxSubParagraphNoStar}[1]{\oldsubparagraph{#1}\mbox{}}
\fi
\makeatother

  \usepackage{longtable,booktabs,array}
\usepackage{calc} 
\usepackage{etoolbox}
\makeatletter
\patchcmd\longtable{\par}{\if@noskipsec\mbox{}\fi\par}{}{}
\makeatother
\IfFileExists{footnotehyper.sty}{\usepackage{footnotehyper}}{\usepackage{footnote}}
\makesavenoteenv{longtable}
\usepackage{graphicx}
\makeatletter
\def\maxwidth{\ifdim\Gin@nat@width>\linewidth\linewidth\else\Gin@nat@width\fi}
\def\maxheight{\ifdim\Gin@nat@height>\textheight\textheight\else\Gin@nat@height\fi}
\makeatother
\setkeys{Gin}{width=\maxwidth,height=\maxheight,keepaspectratio}
\makeatletter
\def\fps@figure{htbp}
\makeatother

\makeatletter
\@ifpackageloaded{caption}{}{\usepackage{caption}}
\AtBeginDocument{%
\ifdefined\contentsname
  \renewcommand*\contentsname{Table of contents}
\else
  \newcommand\contentsname{Table of contents}
\fi
\ifdefined\listfigurename
  \renewcommand*\listfigurename{List of Figures}
\else
  \newcommand\listfigurename{List of Figures}
\fi
\ifdefined\listtablename
  \renewcommand*\listtablename{List of Tables}
\else
  \newcommand\listtablename{List of Tables}
\fi
\ifdefined\figurename
  \renewcommand*\figurename{Figure}
\else
  \newcommand\figurename{Figure}
\fi
\ifdefined\tablename
  \renewcommand*\tablename{Table}
\else
  \newcommand\tablename{Table}
\fi
}
\@ifpackageloaded{float}{}{\usepackage{float}}
\floatstyle{ruled}
\@ifundefined{c@chapter}{\newfloat{codelisting}{h}{lop}}{\newfloat{codelisting}{h}{lop}[chapter]}
\floatname{codelisting}{Listing}

\makeatother
\makeatletter
\@ifpackageloaded{caption}{}{\usepackage{caption}}
\@ifpackageloaded{subcaption}{}{\usepackage{subcaption}}
\makeatother

\ifLuaTeX
  \usepackage{selnolig}  
\fi
\usepackage[]{natbib}
\usepackage{bookmark}

\IfFileExists{xurl.sty}{\usepackage{xurl}}{} 
\hypersetup{
  pdftitle={Title},
  pdfauthor={Author 1; Author 2},
  pdfkeywords={3 to 6 keywords, that do not appear in the title},
  colorlinks=true,
  linkcolor={blue},
  filecolor={Maroon},
  citecolor={Blue},
  urlcolor={Blue},
  pdfcreator={LaTeX via pandoc}}

\newcommand{\anon}{1}
\usepackage{amsmath}
\usepackage{graphicx}
\usepackage{enumerate}
\usepackage{natbib}
\usepackage{url} 

\usepackage{authblk}
\usepackage{hyperref}       
\usepackage{amsfonts}       
\usepackage{nicefrac}       
\usepackage{microtype}      
\usepackage{graphicx}
\usepackage{mathtools}
\usepackage{doi}
\usepackage{authblk}
\usepackage{algorithmic}
\usepackage{dsfont}
\usepackage{xcolor}
\usepackage{enumitem}
\usepackage{float}
\usepackage{atbegshi}
\usepackage{amsthm}
\usepackage{thmtools} 
\usepackage{algorithm2e}
\usepackage{amssymb}
\usepackage{array,multicol}
\usepackage{caption}
\usepackage{titlesec}
\usepackage{siunitx}    
\usepackage{booktabs}   
\usepackage{subcaption}   
\usepackage{longtable}
\usepackage{pdflscape}
\usepackage{rotating}
\usepackage{bibunits}

\usepackage{soul} 

\usepackage{setspace}

\definecolor{darkgreen}{rgb}{0.0, 0.2, 0.13}
\definecolor{ao}{rgb}{0.0, 0.5, 0.0}
\definecolor{vz}{rgb}{0.0, 0.0, 0.0}
\def\R{\mathbb{R}}
\def\C{\mathbb{C}}
\def\Z{\mathbb{Z}}
\def\N{\mathbb{N}}
\def\E{\mathbb{E}}
\def\P{\mathbb{P}}
\def\T{\mathsf{T}}
\def\G{\mathcal{G}} 
\def\Cov{\mathrm{Cov}}
\def\Var{\mathrm{Var}}
\def\indep{\perp\!\!\!\perp}
\def\th{^{\text{th}}}
\def\tr{\mathrm{tr}}
\def\df{\mathrm{df}}
\def\dim{\mathrm{dim}}
\def\col{\mathrm{col}}
\def\row{\mathrm{row}}
\def\nul{\mathrm{null}}
\def\rank{\mathrm{rank}}
\def\nuli{\mathrm{nullity}}
\def\spa{\mathrm{span}}
\def\sign{\mathrm{sign}}
\def\supp{\mathrm{supp}}
\def\diag{\mathrm{diag}}
\def\aff{\mathrm{aff}}
\def\conv{\mathrm{conv}}
\def\dom{\mathrm{dom}}
\def\hy{\hat{y}}
\def\hf{\hat{f}}
\def\hmu{\hat{\mu}}
\def\halpha{\hat{\alpha}}
\def\hbeta{\hat{\beta}}
\def\htheta{\hat{\theta}}
\def\cA{\mathcal{A}}
\def\cB{\mathcal{B}}
\def\cC{\mathcal{C}}
\def\cD{\mathcal{D}}
\def\cE{\mathcal{E}}
\def\cF{\mathcal{F}}
\def\cG{\mathcal{G}}
\def\cK{\mathcal{K}}
\def\cH{\mathcal{H}}
\def\cI{\mathcal{I}}

\def\cL{\mathcal{L}}
\def\cM{\mathcal{M}}
\def\cN{\mathcal{N}}

\def\cP{\mathcal{P}}
\def\cQ{\mathcal{Q}}
\def\cS{\mathcal{S}}
\def\cT{\mathcal{T}}
\def\cW{\mathcal{W}}
\def\cX{\mathcal{X}}
\def\cY{\mathcal{Y}}
\def\cZ{\mathcal{Z}}
\def\cQ{\mathcal{Q}}
\def\cbQ{\mathcal{\bar{Q}}}
\def\smallopsqrtN{o_P(N^{-\frac{1}{2}})}
\def\bQ{\bar{Q}}
\def\bO{\bar{O}}
\def\cbQ{\bar{\mathcal{Q}}}

\def\gtilde{\tilde{g}}

\def\bg{\mathbf{g}}

\def\bcH{\overline{\mathcal{H}}}

\def\bC{\mathcal{\bar{C}}}

\def\ind{\perp\!\!\!\perp}
\def\bgkt{\mathbf{g}_{k,1:t}}
\def\bg{\mathbf{g}}
\def\b0kt{\mathbf{g}_{0,1:t}}
\def\g0t{\mathbf{g_0^t}}
\def\g1{\mathbf{g_1}}
\def\gK{\mathbf{g_K}}
\def\tW{\tilde{W}}
\def\tA{\tilde{A}}
\def\tY{\tilde{Y}}
\def\tO{\tilde{O}}
\def\bO{\bar{O}}
\def\bo{\bar{o}}
\def\bbO{\overline{\overline{O}}}
\def\bbo{\overline{\overline{o}}}

\def\bO{\bar{O}}
\def\bo{\bar{o}}

\def\tPsi{\tilde{\Psi}}

\def\hsigma{\hat{\sigma}}
\def\hsigma2{\hat{\sigma}^2}
\def\asto{\overset{\mathrm{as}}{\to}}
\def\pto{\overset{p}{\to}}
\def\dto{\overset{d}{\to}}

\def\Pnonopt{p^{\textit{non-opt}}}

\def\Tend{T_{\textit{end}}}
\def\nend{n_{\textit{end}}}

\newcommand{\supcovering}[1]{N_{\infty}({#1},\bar{\mathcal{Q}}_K)}
\newcommand{\supcoveringTwo}[2]{\mathcal{N}_{\infty}\left({#1},{#2}\right)}
\newcommand{\supintegral}[1]{J_{\infty}({#1},\bar{\mathcal{Q}}_K)}
\newcommand{\supintegralTwo}[2]{\mathcal{J}_{\infty}\left({#1},{#2}\right)}
\newcommand{\supnorm}[1]{\left\| #1 \right\|_{\infty}}
\newcommand{\bracketingnumber}[2]{\mathcal{N}_{[]}\left(N, \Theta, {#2}, {#1}\right)}
\newcommand{\bracketingintegral}[2]{\mathcal{J}_{[]}\left(N, \Theta, {#2}, {#1}\right)}
\newcommand{\argmin}{\mathop{\mathrm{argmin}}}
\newcommand{\argmax}{\mathop{\mathrm{argmax}}}

\newcommand{\logit}{\mathop{\mathrm{logit}}}

\newtheorem{assumption}{Assumption}

\newtheorem{theorem}{Theorem}

\newtheorem{lemma}{Lemma}

\newtheoremstyle{custom}
  {0pt}   
  {0pt}   
  {\itshape} 
  {}       
  {\bfseries} 
  {.}      
  { }      
  {}       

\usepackage{etoolbox}
\AtBeginEnvironment{figure}{\footnotesize}
\AtBeginEnvironment{table}{\footnotesize}

\usepackage{etoolbox}
\AtBeginEnvironment{figure}{\small}

\allowdisplaybreaks

\begin{document}

\def\spacingset#1{\renewcommand{\baselinestretch}%
{#1}\small\normalsize} \spacingset{1}


\if1\anon
{
  \title{\bf An Online Meta-Level Adaptive Design Framework with Targeted Learning Inference: Applications to Evaluating and Utilizing Surrogate Outcomes in Adaptive Designs}
  \author{Wenxin Zhang\textsuperscript{1}, Aaron Hudson\textsuperscript{2,3}, Maya Petersen\textsuperscript{1}, Mark van der Laan\textsuperscript{1} \thanks{
    \looseness -1
    This work was supported by the National Institutes of Health under Grant R01AI074345. The content is solely the responsibility of the authors and does not necessarily represent the official views of the NIH. 
    We also gratefully acknowledge the contributions of the ADAPT-R Study Team (ClinicalTrials.gov number NCT02338739) and Sponsor (NIH grants R01MH104123, K24AI134413, and R01AI074345). 
    } \\
    
    \textsuperscript{1}Division of Biostatistics, University of California, Berkeley\\
    \textsuperscript{2}Vaccine and Infectious Disease Division, Fred Hutchinson Cancer Center\\
    \textsuperscript{3}Department of Biostatistics, University of Washington}
    
    \maketitle
} \fi

\if0\anon
{
  \bigskip
  \bigskip
  \bigskip
  \begin{center}
    {\LARGE\bf An Online Meta-Level Adaptive Design Framework with Targeted Learning Inference: Applications to Evaluating and Utilizing Surrogate Outcomes in Adaptive Designs}
\end{center}
  \medskip
} \fi

\vspace{-1em}
\begin{abstract}
\looseness -1
Adaptive designs are increasingly used in clinical trials and online experiments to improve participant outcomes by dynamically updating treatment allocation as data accumulate. In practice, experimenters often consider multiple candidate designs, each with distinct trade-offs. However, typically only one design is implemented at a time, leaving benefits and costs of alternative designs unobserved and unquantified. To address this, we propose a novel meta-level adaptive design framework that enables real-time, data-driven evaluation and selection among candidate adaptive designs. Specifically, we define a new class of causal estimands to evaluate adaptive designs and propose Targeted Maximum Likelihood Estimators for these estimands. These estimators are asymptotically normal while accommodating dependence in adaptive-design data without parametric assumptions, enabling online selection among candidate designs. We further apply this framework to a motivating example where multiple surrogates of a long-term outcome are considered for updating randomization probabilities in adaptive experiments. Unlike existing surrogate evaluation methods, our approach comprehensively quantifies surrogates’ utility to accelerate detection of heterogeneous treatment effects, expedite updates to treatment randomization, and improve participant outcomes, facilitating dynamic selection among surrogate-guided designs. Overall, our framework provides a unified approach for evaluating opportunities and costs of various adaptive designs and guiding real-time decision-making in adaptive experiments.
\end{abstract}

\noindent%
\vfill

\newpage
\spacingset{1.6} 

\bibliographystyle{apalike}

\renewcommand{\bibfont}{\small}
\titlespacing*{\section}    {0pt}{1.4ex plus .2ex minus .2ex}{0.8ex plus .2ex}
\titlespacing*{\subsection} {0pt}{1.2ex plus .2ex minus .2ex}{0.6ex plus .2ex}
\titlespacing*{\subsubsection}{0pt}{1.0ex plus .2ex minus .2ex}{0.5ex plus .2ex}

\setlength{\abovedisplayskip}{4pt plus 2pt minus 2pt}
\setlength{\belowdisplayskip}{4pt plus 2pt minus 2pt}
\setlength{\abovedisplayshortskip}{4pt}
\setlength{\belowdisplayshortskip}{4pt}


\section{Introduction}
\subsection{Overview}
\label{section:intro_overview}
\begingroup
\color{vz}

Adaptive experimental designs are gaining popularity in clinical trials and online experiments \citep{robertson2023response,larsen2024statistical}. Unlike traditional static
experimental designs, adaptive designs allow continuous updates to treatment
randomization probabilities and/or other experimental features based on accumulated data during the experiment, supporting a range of experimental objectives. 
A key advantage of adaptive designs is their ability to increase the probability of providing participants with their optimal treatment, particularly when this optimal treatment varies between individuals, thereby improving their outcomes.

However, in practice, experimenters often face the challenge of choosing among multiple candidate adaptive design configurations, each with its own trade-offs. 
A common and compelling example is an adaptive experiment with a longitudinal data structure, in which new participants are continuously enrolled and existing participants generate multiple short-term outcomes that unfold progressively over different follow-up times, while the primary outcome requires long-term follow-up. 
In such settings, surrogate outcomes, either the early proxies of the primary outcome or other relevant metrics requiring shorter follow-ups, have the potential to quickly identify more effective personalized treatments, inform timely adaptations of treatment randomization probabilities with sufficient data, and thereby enhance participants' welfare.
However, when multiple candidate surrogates are available, selecting among them involves complex trade-offs.
Early surrogates provide more abundant and timely data for rapid adaptation, but may be less accurate in detecting treatment effects relevant to the primary outcome. Later surrogates may offer more accurate guidance for optimal treatment, but their delayed availability and limited sample size reduce their impact on early adaptation and participant benefit.
In addition, the sensitivity with which a candidate surrogate captures heterogeneous treatment effects is a key driver of its utility: a surrogate that provides stronger signals for distinguishing which individuals benefit from which treatments enables more timely and targeted adaptation toward personalized treatment.
Despite all these important factors, in practice, only one candidate adaptation can be implemented at each time point of the experiment, which inevitably leaves the gains and losses of unchosen designs unobserved and unquantified, and further complicates the challenge of navigating these trade-offs in a real-time experiment, underscoring the need for a rigorous statistical framework to overcome these difficulties.

To address these challenges, we propose an online meta-level adaptive design framework with targeted learning inference: a general and data-driven adaptive design system
that continuously evaluates and dynamically updates its choice among different candidate designs in real time during ongoing adaptive experiments. 
We refer to this unified framework as TMLE-OSLAD, a Targeted Maximum Likelihood Estimation-based Online Super Learner Adaptive Design, which itself constitutes a meta-level adaptive design.
This framework comprises three key components, each making a distinct methodological contribution.

First, the proposed framework introduces a novel counterfactual perspective for adaptive experiments by asking: How might participants' outcomes have differed under alternative adaptive designs? To formalize this question, we define a new class of target estimands: the mean counterfactual primary outcome that would have been observed had each participant’s treatment, beginning at their enrollment, been assigned according to treatment randomization functions specified by the candidate adaptive design.

Second, to enable valid inference from the non-i.i.d. data generated by adaptive experiments, we establish a Targeted Maximum Likelihood Estimation (TMLE) approach that yields a consistent and asymptotically normal estimator without parametric assumptions, building upon martingale central limit theorems and equicontinuity results via maximal inequalities for martingale processes. We further leverage recent advances in higher-order Highly Adaptive Lasso \citep{van2023higher} to provide equicontinuity conditions. These developments enable robust inference on the proposed causal estimand—the expected outcome under a candidate counterfactual adaptive design. At the end of the experiment, we also provide TMLEs for estimating other causal estimands of interest, such as the
average treatment effect (ATE) and the mean outcome under optimal dynamic
treatment regimes (rules) \citep{murphy2003optimal,robins2004optimal,van2015targeted}.

Third, we develop a TMLE-based Online Super Learner
Adaptive Design, which continuously updates estimates of the proposed estimand to evaluate and compare a library of candidate adaptive designs in real time, selects the data-driven best-performing design based on TMLE estimation and inference, dynamically adapts treatment randomization probabilities according to the chosen design, and continuously monitors and reevaluates the set of candidate designs throughout the experiment.

Returning to the motivating example of adaptive experiments involving surrogates, TMLE-OSLAD demonstrates an additional methodological advantage: it provides an adaptive design perspective for rigorously quantifying the value of information flow as it unfolds over time and is continually leveraged to guide sequential decision-making.
Specifically, TMLE-OSLAD introduces, to the best of our knowledge, for the first time, a comprehensive causal estimand to evaluate the evolving surrogate utility in such dynamic settings with statistical inference.
In this context, the target estimand is defined as the mean counterfactual primary outcome under a candidate surrogate-guided adaptive design, where each participant’s treatment randomization probability at enrollment is determined by the estimated conditional average treatment effect function of the surrogate outcome and adaptively tilted toward the estimated optimal personalized treatment with the objective of optimizing the surrogate outcome.
This metric reflects all critical factors of the utility of a surrogate in accelerating participant benefit through sequential adaptive design, including the timeliness of surrogate availability for enabling early adaptation, its sensitivity to signal heterogeneous treatment effects, and its predictive accuracy for identifying the optimal personalized treatments for the primary outcome. This is distinct from existing methods that focus on evaluating static statistical associations or causal relationships between surrogate and primary outcomes rather than adopting a dynamic, sequential decision-making perspective.

Furthermore, TMLE-OSLAD enables a principled, data-driven implementation of adaptive experiments that facilitates real-time evaluation and selection among different surrogate-guided adaptive designs to improve the primary outcome.

\looseness -1
In summary, TMLE-OSLAD provides a comprehensive meta-level adaptive design framework for monitoring and quantifying opportunities and costs across a library of candidate adaptive experimental strategies. 
These strategies may differ in the quantities that guide adaptation, the statistical methods used to estimate them, or the rules governing design updates. 
The framework dynamically selects among these strategies as data accumulate, empowering data-driven, statistical-inference-based decision-making in diverse adaptive experimental settings.

\par
\endgroup

\subsection{Related Work}
\label{subsection:related_work}

\textbf{Adaptive experimental designs.}
Adaptive design is an active research area in statistics. In clinical trials, adaptive design methods are useful due to their flexibility and efficiency \citep{chow2008adaptive}. These designs enable adaptive modifications to key aspects such as eligibility criteria, randomization probabilities, and other design features based on sequentially accrued data. Among these, Response-Adaptive Randomization (RAR) updates treatment probabilities using treatment and outcome data of previously enrolled subjects to achieve estimation efficiency or enhance participant welfare by reducing assignment to inferior treatments \citep{hu2006theory,robertson2023response}.
Covariate-Adjusted Response-Adaptive (CARA) designs extend RAR by incorporating baseline covariates into randomization, enabling more personalized allocation and improved efficiency \citep{rosenberger2001covariate,hu2006theory,zhang2007asymptotic,van2008construction,rosenberger2008handling,atkinson2011covariate,cheung2014covariate}. One can also guide CARA for multiple objectives, e.g., by integrating both efficiency and ethics considerations in adaptation procedures \citep{hu2015unified}, and by both balancing prognostic covariates and facilitating adaptation procedures using predictive covariates in CARA designs \citep{zhao2022incorporating}.
The origins of RAR and CARA designs also trace back to \citet{thompson1933likelihood} and \citet{robbins1952some} and connect to multi-armed and contextual bandit literature \citep{bubeck2012regret,lattimore2020bandit,simchi2023multi}.
Our work differs by considering settings where multiple candidate adaptive designs are available to choose from and providing a principled way to navigate their trade-offs and select from these designs by TMLE-OSLAD framework, with applications to CARA designs with surrogate outcomes.

\textbf{Surrogate outcomes.}
Surrogate outcomes are often considered as a substitute for the primary outcome when the latter requires a long follow-up time or is expensive to measure. In clinical trials, using surrogate outcomes can result in earlier detection of treatment effects, accelerated drug approvals, and reduced costs \citep{buyse2000validation}. Examples include vaccine-induced immune responses as endpoints for HIV protection \citep{gilbert2008evaluating} and predictive short-term metrics to facilitate digital experiments \citep{duan2021online,richardson2023pareto}. The use of surrogates when primary outcomes are delayed has also been discussed in the adaptive design literature \citep{robertson2023response}. Examples include using short-term responses to guide treatment allocation for outpatients with depressive disorder \citep{tamura1994case} and for participants in stroke trials \citep{nowacki2017surrogate}, Bayesian modeling of surrogate–primary relationships to support RAR designs \citep{huang2009using}, and contextual bandits that handle delayed rewards by running pre-trained Bayesian reward models \citep{mcdonald2023impatient} or imputing delayed outcomes using surrogates constructed from external data under surrogacy assumptions \citep{yang2024targeting}.
Relatedly, \cite{gao2024response} analyzes RAR designs driven by an informative surrogate and provides consistent, asymptotically normal statistical inference for the primary outcome under exponential-family models.
While these approaches implement RAR or CARA procedures under a specific surrogate-guided design, a unified framework is lacking to dynamically re-evaluate and select among multiple surrogates to guide adaptive experiments.

Another related research area is surrogate evaluation. Many statistical frameworks have been developed to evaluate surrogates \citep{prentice1989surrogate,freedman1992statistical,robins1992identifiability,daniels1997meta,lin1997estimating, gilbert2008evaluating} and reviewed comprehensively \citep{buyse2000validation,vanderweele2013surrogate,weir2022informed,elliott2023surrogate}. 
However, most do not account for the potential benefits and risks of using surrogate outcomes in adaptive experiments that aim to learn and respond to treatment effect heterogeneity. Navigating these trade-offs requires both new approaches for quantifying the benefit of candidate surrogates, and novel designs and estimators that make possible online evaluation of these benefits and dynamic surrogate selection in response during the course of a study, which is addressed by our general TMLE-OSLAD framework.

\textbf{Inference based on adaptive design data.}
The statistical inference procedure in TMLE-OSLAD is related to the problem of how to provide valid inference for adaptively collected data in the adaptive design literature. 
For example, asymptotic inference under specified parametric models has been established for data generated by CARA designs \citep{zhang2007asymptotic,zhu2015covariate}.
Within the Targeted Maximum Likelihood Estimation framework \citep{van2006targeted, van2011targeted, van2018targeted}, \citet{van2008construction} developed TMLE tailored to adaptive designs without relying on specific parametric assumptions; subsequent work employed TMLE to estimate risk difference and log-relative risk \citep{chambaz2014inference} and the mean outcome under optimal dynamic treatment rules \citep{chambaz2017targeted} in CARA designs. 
Another line of work develops inference methods for adaptively collected data based on variance-stabilizing weighting techniques introduced by \cite{luedtke2016statistical} in non-adaptive settings and extended to treatment effect estimation or policy evaluation in multi-armed and contextual bandits \citep{hadad2021confidence, bibaut2021post, zhan2021off}, as well as to M-estimation of parametric models using adaptively collected data \citep{zhang2021statistical}.
Our proposed target estimand, the mean outcome under counterfactual adaptive designs, is distinct from these methods and related to causal inference for time series \citep{van2018robust,malenica2021adaptive,malenica2024adaptive}. We develop a TMLE for this estimand, and establish asymptotic normality by martingale theories and equicontinuity results. We also leverage properties of the càdlàg function classes \citep{gill1995inefficient} and non-parametric higher-order spline Highly Adaptive Lasso \citep{van2023higher} to establish equicontinuity conditions. 

\looseness -1
The rest of the paper is organized as follows. Section~\ref{section:stat_setup} presents the statistical setup and the general TMLE-OSLAD framework. Section~\ref{section:our_CARA} specializes the framework for CARA designs involving multiple surrogates. Section~\ref{section:tmle_analysis} provides theoretical results. Section~\ref{section:simulations} presents simulation studies, followed by a real-data-based simulation in Section~\ref{section:real_sim}.  Section~\ref{section:discussions} concludes.

\section{TMLE-Based Online Super Learner Adaptive Design}
\label{section:stat_setup}
\looseness -1
This section outlines the data structure, notation, and methods for the proposed TMLE-OSLAD framework. 
Suppose we have baseline covariates $W \in \cW$ and a treatment $A \in \cA = \{0,1\}$. Let $Y_K$ denote the primary outcome, observed $K$ follow-up periods after treatment assignment, with larger values indicating better outcomes. We also consider $K-1$ intermediate surrogate outcomes $Y_k$, each observed at $k = 1, \ldots, K-1$.
We assume all outcomes $Y_k$ lie in $[0,1]$.
Let $T$ denote the total duration of enrollment in the experiment. At each time point, the experimenter uses historical data from previously enrolled participants to update the treatment assignment mechanism for new participants. These updates can aim to increase the chance that participants receive optimal personalized treatment by learning heterogeneous treatment effects, which are detailed in Section~\ref{section:our_CARA} and Appendix~\ref{appendix:details_of_CARA}.

\subsection{Statistical Setup}
\textbf{Sequential data structure.}
\looseness -1
We use $t = 1,\ldots,T$ to index the sequential enrollment times when participants enter the study and are randomized to one of the treatment arms. Let $E(t)$ denote the number of subjects enrolled at time $t$, and define the cumulative enrollment by $N(t) = \sum_{t'=1}^{t} E(t')$.
For each participant $i$, we sequentially observe 
$
O_i = (t_i, W_i, A_i, Y_{i,1}, \ldots, Y_{i,K})$,
where $t_i$ is the enrollment time, $W_i$ are baseline covariates observed at $t_i$, and $A_i$ is the treatment assigned at $t_i$. Surrogate outcomes $Y_{i,1}, \ldots, Y_{i,K-1}$ are observed at times $t_i + k$ for $k = 1,\ldots,K-1$. The final outcome of interest, $Y_{i,K}$, is observed at $t_i + K$. To emphasize the timeline when each variable becomes available, define
$
\tilde{W}_i(t_i) \equiv W_i, \, 
\tilde{A}_i(t_i) \equiv A_i, \, $ and $ 
\tilde{Y}_i(t_i + k) \equiv Y_{i,k}$ for $k = 1,\ldots,K$.
We use these notations interchangeably.
At each time $t$, the experimenter has access to all information accrued up to $t-1$, denoted by
$
\bar{O}(t-1) := \{ \tilde{W}_i(t'), \tilde{A}_i(t'), \tilde{Y}_i(t') : i \in [N(t-1)],\ t' \leq t-1 \}.
$
This represents all covariates, treatments, and outcomes of previous participants that have become observable before time $t$.
From $t-1$ to $t$, new data are collected from two sets:
$\cI_1(t) := \{ i : t - K \leq t_i \leq t-1 \}$ and 
$\cI_2(t) := \{ i : t_i = t \}.
$
$\cI_1(t)$ contains previously enrolled participants whose outcomes are still arriving, and $\cI_2(t)$ contains participants newly enrolled at time $t$. The experimenter first observes the outcomes revealed at time $t$ for $\cI_1(t)$, then updates the history $
\bcH(t) := \bar{O}(t-1) \cup \{ \tilde{Y}_i(t) : i \in \cI_1(t) \}
$,  
where $\tilde{Y}_{i}(t)$ is the outcome for participant $i$ that becomes observable at time $t$. Next, the experimenter enrolls $\cI_2(t)$, observes their baseline covariates, and assigns treatments.

\textbf{Adaptive treatment randomization.}  
When a new participant $i$ is enrolled at time \( t_i \), a treatment $A_i$ can be randomized based on the information available up to that point and the baseline covariates $W_i$.
We denote the participant-specific pre-treatment history as $\cH(i):=\left(\bcH(t_i), W_i\right)$, where $\bcH(t_i)$ is the observed data of previous participants enrolled before time $t_i$.
We assume that treatment randomization depends on a fixed-dimensional summary context of the pre-treatment history, defined as
$
C_i := f_C\big(\cH(i)\big) \in \mathcal{C}.
$
The treatment randomization probability is generated by a mapping $g_{i}: \mathcal{C} \to \cG$, where $\cG$ is a class of probability distributions of treatment assignment, and $g_i(a | C_i)$ is the probability of assigning $a \in \cA$ given context $C_i$.

\textbf{Likelihood.}
\looseness -1
Let $Q_W$ denote the distribution of baseline covariates $W$ and $Q_{Y_k}$ denote the distribution of $Y_k$ given $W,A,Y_1,\ldots,Y_{k-1}$ for outcome $Y_k$ ($k = 1,\ldots,K$).
We assume for every participant $i$, $t_i \ind (W_i,A_i,Y_{i,1},\ldots,Y_{i,K})$, $W_i \sim Q_W$, and $Y_{i,k} \sim Q_{Y_k}(\cdot|W_i,A_i,Y_{i,1},\ldots,Y_{i,k-1})$.
Specifically, let $q_W$ denote the marginal density of the baseline covariates $W$, and $q_{Y_k}$ denote the conditional density of the $k$-th outcome $Y_k$ given $W,A,Y_1,\ldots,Y_{k-1}$.
The likelihood of the observed data of $n$ subjects is given by 
$
p^n(o_1,\ldots,o_n) 
= \prod_{i=1}^{n}
q_W(w_i) \,
g_i(a_i | C_i)
\prod_{k=1}^{K}
q_{Y_k}\big(y_{i,k} \mid w_i, a_i, y_{i,1}, \ldots, y_{i,k-1} \big)$.

\textbf{Structural Causal Model.}
We formalize the data generating process through a Structural Causal Model (SCM) \citep{pearl2000causality} over the experiment, which also underpins estimation and inference in our setting. At each time $t$, we first observe the history data $\bcH(t)$, which includes, for previously enrolled participants with $t_i=t-k$ ($k = 1,\ldots,K$), the outcomes given by $Y_{i,k} = f_{Y_k}(W_i, A_i, Y_{i,1},\ldots, Y_{i,k-1}, U_{Y_{i,k}})$.
Second, for each new unit $i$ enrolled at time $t$, we observe 
$t_i = t$; $W_i = f_W(U_{W_i})$; 
$C_i = f_C\left(\cH(i)\right)$; and 
$A_i = f_A\left(C_i, U_{A_i}\right)$.
Here, $U_i = (U_{W_i}, U_{A_i}, U_{Y_{i,1}}, \ldots, U_{Y_{i,K}})$ are exogenous variables independent across participants. 
The actual adaptive design for $n$ subjects is represented by the sequence of applied treatment-randomization functions, denoted by $\bg_{1:n}:=\{g_{i}: i=1,\ldots,n\}$.

\subsection{Evaluation of Counterfactual CARA Designs}
\label{subsection:general_target_parameter}

Here we define a target estimand within a modified SCM under a counterfactual adaptive design. 
Specifically, we replace every $A_i = f_A(C_i,U_{A_i})$ in the original SCM with $A_i^* \sim g_{*,i}(\cdot|C_i)$ for every $i$, which is a treatment randomization function under a counterfactual adaptive design, based on the observed summary measure $C_i$ for each unit $i$ at their entry time. The outcomes then follow $Y^*_{i,k} = f_{Y_k}(W_i, A_i^*, Y_{i,1}^*,\ldots, Y_{i,k-1}^*, U_{Y_{i,k}})$ for $k = 1, \ldots, K$. Let $\bC(n) := (C_1,\cdots,C_{n})$. Under this modified SCM, the proposed estimand is:
\[
\Psi_{\bg_{*,1:n}\left(\bC(n)\right)}
= \frac{1}{n} \sum_{i=1}^{n} E_{g_{*,i}\left(\cdot|C_i\right)}\left[Y_{i,K}^*\right],
\]
the mean counterfactual primary outcome had all participants, at their entry, been treated according to the treatment randomization mechanism $g_{*,i}$ that could have been applied under an alternative adaptive design $\bg_{*,1:n} = (g_{*,1},\ldots,g_{*,n})$. 

\begingroup
\color{vz}
This target estimand naturally supports the counterfactual evaluation of multiple candidate adaptive designs of interest.
Suppose we have $J$ different candidate adaptive designs, where each design is defined by a series of counterfactual treatment randomization functions $\{g_{j,i}^* : i=1, \ldots, n\}$ for $j \in [J]$, and depends on the $j$-specific design configuration aiming to increase the probability that new participants receive personalized optimal treatment assignments to maximize their outcomes. The corresponding counterfactual estimand for the $j$-th candidate design is defined as:
$
\Psi_j := 
\frac{1}{n} \sum_{i=1}^{n} 
E_{g_{j,i}^*(\cdot | C_i)} \big[ Y_{i,K}^* \big]$, reflecting the expected final outcome had the population been treated under that adaptive design.
This enables a principled comparison of all candidate designs by comparing the values of $\{ \Psi_j \}_{j=1}^J$. 
\par
\endgroup

Let $Q_0 := (Q_{0,W},Q_{0,K}) \in \cQ_W \times \cQ_K$, where $Q_{0,W}$ is the true distribution of $W$, and $Q_{0,K}$ is the conditional distribution of $Y_K$ given $A$ and $W$, with the associated $\bQ_{0,K}$ denoting the true conditional mean function of the final outcome $Y_K$ given $A$ and $W$.
We also use $g_{0,i}$ to denote the actual treatment randomization function applied to each participant $i$ in the experiment, which is known to the experimenter.
One can identify the proposed causal estimand based on the observed data under the following assumptions.

\begin{assumption}[Sequential Exchangeability]\label{sequential_assumption_general} 
For any $i \in [n]$, $Y_{i, K}(a) \indep A_i \mathrel{|} C_i$.
\end{assumption}

\begin{assumption}[Positivity]\label{assumption:positivity_general}
For any $i \in [n]$ and $a \in \mathcal{A}$, $g_{0, i}(a|C_i) > 0$.
\end{assumption}

\begin{theorem}[Identification]
With assumptions \ref{sequential_assumption_general} and \ref{assumption:positivity_general},
one can identify $\Psi_{j}$ under $Q_0$ by
\[
\Psi_{j}(Q_0) 
= \Psi_{\bg^*_{j,1:n}}(Q_0)
= \frac{1}{n} \sum_{i=1}^{n} E_{Q_0,g^*_{j,i}}\left[Y_{i,K}^*\right]= \frac{1}{n} \sum_{i=1}^{n} \sum_{a=0}^{1} \bQ_{0,K}(a,W_i) g^*_{j,i}(a|C_i).
\]
\end{theorem}

We note that this estimand provides a fast and robust approximation to an ideal estimand: the expected outcome obtained by integrating the entire distribution of data that would sequentially arise from the candidate adaptive design of interest. However, the latter's estimation is considerably more complex and computationally intensive.
In contrast, our proposed context-specific estimand substantially simplifies the estimation procedure, enables doubly-robust estimation, and closely approximates the ideal estimand when the observed contexts are reflective of those that would be observed in the candidate design, especially when the adaptation procedure is based on learning from a stationary target (e.g. $Q_0$) that dominates over the experiments. This yields an efficient and effective approach to distinguishing high-performing from low-performing designs and facilitating online selection among candidate designs. We refer readers to Appendix~\ref{appendix:ideal_estimand} for further discussions and empirical results showing minimal difference between the two estimands in our simulations. 

\subsection{Targeted Maximum Likelihood Estimation}
\label{subsection:TMLE_general}
At each time point $t$, let $n_t := N(t-K)$ denote the total number of participants for whom the final outcome $Y_K$ has been observed at that time. We define the target estimand for the $j$-th counterfactual adaptive design at time $t$ as $\Psi_{t,j}:= \frac{1}{n_t} \sum_{i=1}^{n_t} E_{Q_0,g^*_{j,i}}\left[Y_{i,K}^*\right]= \frac{1}{n_t} \sum_{i=1}^{n_t} \sum_{a=0}^{1} \bQ_{0,K}(a,W_i) g^*_{j,i}(a|C_i)$. 
To estimate $\Psi_{t,j}$, we provide a Targeted Maximum Likelihood Estimation (TMLE) approach.
Let $\bQ_{n,K,t}$ be an initial estimate of $\bQ_{0,K}$ based on the historical data $\bcH(t)$ observed up to time $t$.
We apply a TMLE procedure to update $\bQ_{n,K,t}$ by fitting the logistic model \citep{van2006targeted,van2008construction,van2011targeted,malenica2024adaptive}:
$
\logit \bQ_{n,K,t, \epsilon} \left(A, W \right) = \logit \bQ_{n,K,t} \left(A,W \right) + \epsilon$, 
with weights $\frac{g^*_{j, i}(A_i|C_i)}{g_{0, i}(A_i|C_i)}$.
The estimate of $\epsilon$ is denoted by $\hat{\epsilon}_{t}$, and an updated estimator of the conditional mean function $\bQ_{n,K,t}^* := \bQ_{n,K,t,\hat{\epsilon}_{t}}$ is obtained, which solves:
$
\frac{1}{n_t} \sum_{i=1}^{n_t}  \frac{g^*_{j, i}(A_i|C_i)}{g_{0, i}(A_i|C_i)}\left(Y_{i, K}-\bQ_{n,K,t}^*(A_i, W_i)\right)=0$.
The TMLE estimate of $\Psi_{t,j}$ is:
\[
\hat{\Psi}_{t,j} := \frac{1}{n_t} \sum_{i=1}^{n_t} \sum_{a=0}^{1} \bQ_{n,K,t}^*(a, W_i) g^*_{j,i}(a|C_i).
\]
The estimated standard error is $\sqrt{\hsigma2_{t,j}/n_t}$, 
where 
$
\hsigma2_{t,j} := \frac{1}{n_t} \sum_{i=1}^{n_t} \left[\frac{g^*_{j, i}(A_i|C_i)}{g_{0, i}(A_i|C_i)}\left(Y_{i, K}-\bQ_{n,K,t}^*(A_i, W_i)\right)\right]^2$.
The $100(1-\alpha)\%$ Wald confidence interval for $\hat{\Psi}_{t,j}$ is 
$
\left[\hat{\Psi}_{t,j}-z_{1-\alpha/2}\sqrt{\frac{1}{n_t}\hsigma2_{t,j}},\hat{\Psi}_{t,j}+z_{1-\alpha/2}\sqrt{\frac{1}{n_t}\hsigma2_{t,j}} \, \right]$, where $z_{1-\alpha/2}$ is the $1-\alpha/2$ quantile of the standard normal distribution (e.g. $\alpha = 0.05$).

\looseness -1
Under the conditions in Section~\ref{section:tmle_analysis}, TMLE is consistent, asymptotically normal and doubly robust, meaning that consistency holds if the outcome regression or the treatment assignment mechanism is correctly specified. In adaptive designs, the treatment assignment mechanism is known, so TMLE remains consistent even if the outcome regression is misspecified. That said, more accurate estimation of $\bQ_{0,K}$ can reduce the variance of TMLE. When multiple estimators are available, a Super Learner~\citep{van2007super} can be used to select the best-performing estimator by cross-validation. We further recommend including the higher-order Highly Adaptive Lasso~\citep{van2023higher} to facilitate the equicontinuity conditions.

\subsection{Online Super Learner Adaptive Design}
\label{subsection:OSAD}
As the experiment progresses to a new time point $t$, there are up to $J$ candidate treatment randomization functions $g^*_{1,i},\ldots,g^*_{J,i}$ for assigning treatment randomization probabilities for a new participant $i$ enrolled at $t$.
For each $j \in [J]$, $g^*_{j,i}$ is generated based on historical data using the $j$-specific adaptive design configurations.
The proposed Online Super Learner Adaptive Design begins by estimating
$\Psi_{t,j}$ for each $j \in [J]$ and obtaining confidence intervals for these estimates.
Then, it selects the design indexed by
$
j^*_t := \argmax_{j \in [J]}\hat{\Psi}_{t,j}-z_{1-\alpha/2}\sqrt{\frac{1}{n_t}\hsigma2_{t,j}}
$, 
the candidate design with the highest lower confidence interval bound among all $\hat{\Psi}_{t,j}$ for $j \in [J]$, prioritizing designs with high performance with low uncertainty to mitigate the risk of selecting a suboptimal design due to unstable estimates. Finally, for all participants 
$i$ enrolled at time 
$t$, we assign treatment randomization probabilities using $g^*_{j_t^*,i}$.

\subsection{Statistical Inference at the End of Adaptive Experiments}
\label{section:inference_at_the_end_general}
\looseness -1
We denote by $\Tend := T + K$ the time when the experiment concludes, i.e., when all final outcomes $Y_K$ have been observed after the last enrollment at time $T$. 
The corresponding sample size is $\nend := N(T) = \sum_{t=1}^T E(t)$. 
At time $\Tend$, the TMLE procedure in Section~\ref{subsection:TMLE_general} applies to estimate $\Psi_{\Tend,j}$, the expected final outcome under the $j$th counterfactual adaptive design.
We also provide TMLEs for other causal estimands commonly studied in randomized trials, including the average treatment effect and the mean outcome under dynamic treatment rules.
Their statistical inference is based on Wald confidence intervals using the corresponding TMLE estimates and their estimated standard errors.

\looseness -1
\textbf{Average treatment effect.} We define the following average treatment effect (ATE) on the primary outcome $Y_K$ across all participants within the experiment:
$
\psi_{\Tend}^{\text{ATE}} := \frac{1}{\nend} \sum_{i=1}^{\nend} \left[\bQ_{0,K}(1,W_i)-\bQ_{0,K}(0,W_i)\right]$.
For estimation, we first obtain an initial estimate $\bQ_{n,K,\Tend}$ of $\bQ_{0,K}$ and then perform the TMLE update with weights $\frac{2A_i-1}{g_{0,i}(A_i|C_i)}$, yielding $\bQ_{n,K,\Tend}^*$ that solves
$
\frac{1}{\nend} \sum_{i=1}^{\nend}  \frac{2A_i-1}{g_{0, i}(A_i|C_i)}\left(Y_{i, K}-\bQ_{n,K,\Tend}^*(A_i, W_i)\right)=0.
$
The TMLE estimate of $\psi_{\Tend}^{\text{ATE}}$ is  $\hat{\psi}_{\Tend}^{\text{ATE}} := \frac{1}{\nend} \sum_{i=1}^{\nend} \left[\bQ_{n,K,\Tend}^*(1,W_i)-\bQ_{n,K,\Tend}^*(0,W_i)\right]$.
The standard error estimate is $\sqrt{\hat{\sigma}^{2,\text{ATE}}_{\Tend}/\nend}$, where
$\hat{\sigma}^{2,\text{ATE}}_{\Tend} = \frac{1}{\nend} \sum_{i=1}^{\nend} \left[\frac{2A_i - 1}{g_{0,i}(A_i \mid C_i)} \left(Y_{i,K} - \bQ_{n,K,\Tend}^*(A_i, W_i)\right)\right]^2$.  

\looseness -1
\textbf{Mean outcome under a dynamic treatment rule.}
Let $d: \cW \to \cA$ denote a fixed dynamic treatment rule mapping baseline covariates to a treatment arm. 
We define a target estimand 
$\psi^{d}_{\Tend} := \frac{1}{\nend} \sum_{i=1}^{\nend} \bQ_{0,K}\big(d(W_i),W_i\big)$
as the mean final outcome $Y_K$ under rule $d$. 
We apply TMLE with weights
$\frac{I(A_i=d(W_i))}{g_{0,i}(A_i|C_i)}$ to solve
$\frac{1}{\nend} \sum_{i=1}^{\nend}
\frac{I(A_i=d(W_i))}{g_{0,i}(A_i|C_i)}
(Y_{i,K}-\bQ^*_{n,K,\Tend}(A_i,W_i))=0$, and yield 
$\hat{\psi}^{d}_{\Tend}:=\frac{1}{\nend}\sum_{i=1}^{\nend}\bQ^*_{n,K,\Tend}(d(W_i),W_i)$.
The standard error estimate is $\sqrt{\hat{\sigma}^{2,d}_{\Tend}/\nend}$, where 
$\hat{\sigma}^{2,d}_{\Tend} := \frac{1}{\nend}\sum_{i=1}^{\nend}\!\left[\frac{I(A_i=d(W_i))}{g_{0,i}(A_i|C_i)}(Y_{i,K}-\bQ^*_{n,K,\Tend}(A_i,W_i))\right]^2$.

\looseness -1
We also note their marginal versions $\psi^{\text{ATE}} := E\left[\bQ_{0,K}(1,W)-\bQ_{0,K}(0,W)\right]$ and 
$\psi^{d} := E\left[\bQ_{0,K}\big(d(W),W\big)\right]$. The TMLE estimates are the same as above, while the standard error estimates differ, with 
$\sqrt{\hat{\sigma}^{2,\text{ATE}}/\nend}$
for $\psi^{\text{ATE}}$ and $\sqrt{\hat{\sigma}^{2,d}/\nend}$ for $\psi^{d}$, where
$\hat{\sigma}^{2,\text{ATE}}:
=
\frac{1}{\nend}\sum_{i=1}^{\nend}
\Big[
\frac{2A_i-1}{g_{0,i}(A_i\mid C_i)}
\left(Y_{i,K}-\bQ_{n,K,\Tend}^*(A_i,W_i)\right)
+
\bQ_{n,K,\Tend}^*(1,W_i)
-
\bQ_{n,K,\Tend}^*(0,W_i)
-
\hat{\psi}^{\text{ATE}} 
\Big]^2$ and  
$\hat{\sigma}^{2,d} := \frac{1}{\nend}\sum_{i=1}^{\nend}\!\left[\frac{I(A_i=d(W_i))}{g_{0,i}(A_i|C_i)}(Y_{i,K}-\bQ^*_{n,K,\Tend}(A_i,W_i))+\bQ^*_{n,K,\Tend}(d(W_i),W_i)- \hat{\psi}^{d} \right]^2.$ 

\section{Evaluating and Utilizing Surrogates in CARA Designs}
\label{section:our_CARA}
\looseness -1
In this section, we apply the TMLE-OSLAD framework to the motivating example of adaptive experiments involving multiple surrogate outcomes before the final outcome.
Recall that there are $K-1$ candidate surrogate outcomes, $Y_1,\ldots,Y_{K-1}$  measured before the primary outcome $Y_K$ at follow-up times $k=1,\ldots,K-1$.
Each surrogate $Y_k$ defines a distinct candidate adaptive design, which leverages accumulating information from short-term surrogate outcomes to increase the probability for assigning the estimated personalized treatment associated with that surrogate over time.
\begingroup
\color{vz}
TMLE-OSLAD is implemented to continuously evaluate and dynamically select among different surrogate-specific adaptive designs over the experiment.
We first introduce a CARA procedure that adapts treatment randomization probabilities based on the estimated conditional average treatment effect (CATE) to improve participant outcomes. Building on Section \ref{section:stat_setup}, we define our target estimand that evaluates each candidate surrogate within a counterfactual adaptive design that utilizes the surrogate to adapt treatment randomization probabilities.
We then demonstrate TMLE-OSLAD that uses TMLE to estimate these target estimands to evaluate surrogate-guided adaptive designs, and utilizes the most useful one to guide the ongoing adaptive experiment. 
Finally, we conduct post-experiment inference for other causal estimands such as ATE and the mean primary outcome under optimal dynamic treatment rules that optimize primary or surrogate outcomes.
\par
\endgroup

\subsection{CARA Designs Using Surrogate Outcomes}
\label{section:sampling_schemes}

In our CARA designs, at each time point, we update the model to estimate the CATE for the outcome of interest and its associated uncertainty, and use these estimates to assign randomization probabilities to newly enrolled participants. 
This design aims to maximize participant outcomes by increasing the chance to allocate the estimated optimal treatment, with a principled exploration–exploitation balance ensuring sufficient exploration when the superiority of a treatment arm is uncertain and maintaining the ability to continue learning heterogeneous treatment effect from data collected at later time points. Specifically, when the estimated CATE is small relative to its uncertainty, indicating that the superiority of one treatment arm is uncertain, the treatment randomization probabilities remain close to 0.5. This avoids overcommitting to a noisy and potentially suboptimal rule and ensures sufficient exploration across treatment arms. As the estimated CATE suggests a stronger signal for the superior treatment arm, the probability of assigning that treatment increases up to a maximum exploitation probability. Meanwhile, a positive lower bound on the probability of assigning the alternative arm is maintained to ensure sufficient support for valid estimation and inference of other estimands at the end of the experiment (see Sections~\ref{section:inference_at_the_end_general} and~\ref{section:inference_at_the_end}).

More details of the CARA design are provided in the Appendix. Specifically, CATE is estimated using a doubly robust approach based on pseudo-outcomes \citep{van2006statistical,luedtke2016super} (Appendix~\ref{appendix:CATE}), with point estimates and confidence intervals constructed based on higher-order Highly Adaptive Lasso (HAL) \citep{van2023higher} (Appendix~\ref{appendix:HAL}). Details of the treatment randomization function can be found in Appendix~\ref{appendix:details_of_CARA}.

\subsection{Surrogate-Guided Online Super Learner Adaptive Design}
\label{subsection:TMLE_OSLAD_surrogate}
\begingroup
\color{vz}
To enable data-adaptive selection of the most effective surrogate to guide treatment allocation, we leverage the general TMLE-OSLAD framework in Section~\ref{section:stat_setup} to both evaluate and dynamically utilize candidate surrogate outcomes within real-time CARA experiments. At each interim time $t$, we use all observed historical data to assess the utility of each surrogate-driven adaptive design and to select the most promising surrogate for adaptation.

First, at time $t$, we define the target estimand
$
\Psi_{t, \bg^*_{k,1:n_t}}
:= \frac{1}{n_t} \sum_{i=1}^{n_t} E_{g^*_{k,i}}\left[Y_{i,K}^*\right]$  for each candidate surrogate $Y_k$,
which represents the mean final outcome that would be achieved for the $n_t$ fully-followed participants if each had been assigned with the treatment randomization function $g^*_{k,i}$, at their enrollment, determined by the adaptive design using surrogate $Y_k$. 
We refer to this estimand as ``the utility of a surrogate'' in the CARA design (see Appendix \ref{appendix:remark_utility_of_surrogates} for further discussion of this estimand).

To construct $g^*_{k,i}$ from the CATE for $Y_k$, we impose the following identification assumption.
\begin{assumption}[Sequential Exchangeability for Surrogate Outcomes]\label{sequential_assumption} 
For any $k \in \{1,\ldots,K-1\}$ and any subject $i \in \{1, \ldots, n_t\}$, $Y_{i, k}(a) \indep A_i \mathrel{|} C_i$.
\end{assumption}

Under Assumptions \ref{sequential_assumption_general} and  \ref{assumption:positivity_general}, $\Psi_{t, \bg^*_{k,1:n_t}}$ can be identified as
$
\Psi_{t,k} := \frac{1}{n_t} \sum_{i=1}^{n_t} \sum_{a=0}^{1} \bQ_{0,K}(a, W_i) g^*_{k,i}(a|C_i)$.
Then, we estimate $\Psi_{t,k}$ for each surrogate $Y_k$ using TMLE, which is given by
$
\hat{\Psi}_{t,k} = \frac{1}{n_t} \sum_{i=1}^{n_t} \sum_{a=0}^{1} \bQ_{n,K,t}^*(a, W_i) g^*_{k,i}(a|C_i)$.
The estimated standard error is $\sqrt{\hsigma2_{t,k}/n
_t}$,
where 
$
\hsigma2_{t,k}= \frac{1}{n_t} \sum_{i=1}^{n_t} \left[\frac{g^*_{k, i}(A_i|C_i)}{g_{0, i}(A_i|C_i)}\left(Y_{i, K}-\bQ_{n,K,t}^*(A_i, W_i)\right)\right]^2$.
Next, we compare all candidate surrogates and the final outcome by constructing confidence intervals for $\Psi_{t,k}$ ($k = 1,\ldots, K$) and identify the outcome $Y_{k^*_t }$ with the highest lower confidence bound, where
$
k^*_t := \argmax_{k \in [K]}\hat{\Psi}_{t,k}-z_{1-\alpha/2}\sqrt{\hsigma2_{t,k}/n_t}$.
For example, if two surrogates have similar point estimates but one yields a much narrower confidence interval, the surrogate with the higher lower bound is prioritized, as it reflects more evidence in its utility relative to uncertainty.
Accordingly, we select $Y_{k^*_t}$ as the outcome to guide the next-stage treatment assignment for newly enrolled participants at time $t$, utilizing the treatment randomization function $g^*_{k^*_t,i}$. 
This procedure is repeated at each subsequent time point, allowing adaptive surrogate selection and treatment allocation as new data accrue.

\par
\endgroup
\subsection{Inference after Adaptive Experiments Involving Surrogates}
\label{section:inference_at_the_end}
\looseness -1
At the end of the adaptive experiment, one can follow Section~\ref{section:inference_at_the_end_general}
to conduct statistical inference for causal estimands such as the average treatment effect or the mean outcome under dynamic treatment rules. 
Here, we consider the mean primary outcome under optimal dynamic treatment rules (ODTR) that maximize the primary outcome \citep{murphy2003optimal,robins2004optimal,van2015targeted} or a surrogate outcome \citep{hsu2015surrogate,yang2024targeting}.

\looseness -1
For a primary or surrogate outcome $Y_k$ ($k\in\{1,\ldots,K\}$), define the optimal dynamic treatment rule $d^*_{0,k}:\cW\to\cA$ by $d^*_{0,k}(W)=I\!\left(E[Y_k(1)-Y_k(0)\mid W]>0\right)$. We estimate this rule using all observed data and obtain $d^*_{n,k,\Tend}$ (for brevity, we also write $d^*_{n,k}$). The target estimand is the mean primary outcome $Y_K$ under this rule
$\psi^{d^*_{n,k}}_{\Tend}:=\frac{1}{\nend}\sum_{i=1}^{\nend}\bQ_{0,K}(d^*_{n,k}(W_i),W_i)$,
or its marginal version
$\psi^{d^*_{n,k}}:=E[\bQ_{0,K}(d^*_{n,k}(W),W)]$.
We estimate them by TMLE solving 
$\frac{1}{\nend}\sum_{i=1}^{\nend}\frac{I(A_i=d^*_{n,k}(W_i))}{g_{0,i}(A_i\mid C_i)}(Y_{i,K}-\bQ^*_{n,K,\Tend}(A_i,W_i))=0$,
yielding a point estimate 
$\hat{\psi}^{d^*_{n,k}}_{\Tend}:=\frac{1}{\nend}\sum_{i=1}^{\nend}\bQ^*_{n,K,\Tend}(d^*_{n,k}(W_i),W_i)$, with standard error estimates given by Section \ref{section:inference_at_the_end_general}.

\looseness -1
We note that this approach provides an additional treatment-rule-based perspective on evaluating surrogate utility by the mean primary outcome $Y_K$ achieved under ODTRs maximizing a surrogate outcome $Y_k$ ($k \ne K$), which can be compared with the value achieved under the ODTR maximizing $Y_K$ itself.
Moreover, for $d^*_{n,K}$, the estimated ODTR maximizing the primary outcome $Y_K$, the standard TMLE for its evaluation may exhibit upward bias as the same data are used to both learn the rule and evaluate its value. We recommend applying cross-validated TMLE (CV-TMLE) \citep{zheng2011cross,van2015targeted, montoya2021optimal}, which uses sample splitting to avoid finite sample bias.

\section{Theoretical Analysis of TMLE}
\label{section:tmle_analysis}
This section derives consistency and asymptotic normality of the
proposed TMLE.
For notational continuity with the surrogate-guided adaptive design setup introduced above, we continue to index candidate designs by \( k \in [K] \) in the TMLE analysis, and the target estimand of interest is \( \Psi_{t,k} \).  
This indexing is notationally equivalent to the general index \( j \in [J] \) used in Section~\ref{section:stat_setup} for estimating \( \Psi_{t,j} \); the theoretical results and proofs apply identically under either notation.
We first study the difference between TMLE $\hat{\Psi}_{t,k}$ and the truth $\Psi_{t,k}$. 
For any $\bQ_K \in \cbQ_K$, define
$
D_k(\bQ_K)\left(O_i,C_i\right) := \frac{g^*_{k,i}\left(A_i|C_i\right)}{g_{0,i}\left(A_i|C_i\right)}\left( Y_{i, K}-\bQ_K\left(A_i, W_i\right) \right)$.
\begin{restatable}{theorem}{repThmDecomposition}
\label{Decomposition}
For any 
$\bQ_{n,K,t}^* \in \cbQ_K$ and $\bQ_{1,K} \in \cbQ_{K}$,  
the difference between $\hat{\Psi}_{t,k}$ and $\Psi_{t,k}$ can be decomposed as
$
\hat{\Psi}_{t,k}-\Psi_{t,k} = M_{1,n_t}(\bQ_{1,K}) +  M_{2,n_t}(\bQ_{n,K,t}^*, \bQ_{1,K}),
$
where
\begin{align*}
M_{1,n_t}\left(\bQ_{1,K} \right) &= \frac{1}{n_t}\sum_{i=1}^{n_t}\Big[ D_k\left(\bQ_{1,K}\right)\left(O_i, C_i\right)-E\left[ D_k\left(\bQ_{1,K}\right)\left(O_i, C_i\right)|\cH(i) \right] \Big],\\
M_{2,n_t}\left(\bQ_{n,K,t}^*, \bQ_{1,K}\right) &= \frac{1}{n_t}\sum_{i=1}^{n_t}\Biggl\{ \biggl[
D_k\left(\bQ_{n,K,t}^*\right)\left(O_i,C_i\right)-E\left[ D_k\left(\bQ_{n,K,t}^*\right)\left(O_i,C_i\right)|\cH(i) \right] \biggr]\\
&- \biggl[ D_k\left(\bQ_{1,K}\right)\left(O_i, C_i\right)-E\left[ D_k\left(\bQ_{1,K}\right)\left(O_i, C_i\right)|\cH(i) \right] \biggr] \Biggr\}.
\end{align*}
\end{restatable}
The difference of TMLE $\hat\Psi_{t,k}$ and the true $\Psi_{t,k}$ is decomposed into $M_{1,n_t}$, an average of a martingale difference sequence, and a martingale process $M_{2,n_t}$. 
Under the strong positivity and stabilized conditional variance conditions below, the asymptotic normality of $M_{1,n_t}$ follows directly from the martingale central limit theorem (e.g., Theorem 2 in \cite{brown1971martingale}).

\begin{assumption}[Strong positivity]\label{assumption:strong_positivity}
    For any $i$, there exists $\zeta >0$ such that $g_{0,i} > \zeta$ and $g^*_{k,i} > \zeta$.
\end{assumption}

\begin{assumption}[Stabilization of the average of conditional variances] 
\label{assumption:stabilized_variance}
    There exists $\sigma^2_{t,k} \in \R^+$ such that 
    $\frac{1}{n_t} \sum_{i=1}^{n_t}Var\left[D_k(\bQ_{1,K})(O_i, C_i) \mid \cH(i)\right] \pto \sigma^2_{t,k}.
    $
\end{assumption}

\begin{theorem}[Asymptotic normality of the first term] 
\label{asynormal_first}
    Suppose that assumptions \ref{assumption:strong_positivity} and \ref{assumption:stabilized_variance} hold, then $\sqrt{n_t}M_{1,n_t}\left(\bQ_{1,K} \right) \dto N(0,\sigma^2_{t,k})$.
\end{theorem}
\looseness -1
We note that $M_{1,n_t}$ depends only on a fixed function $\bQ_{1,K} \in \cbQ_K$. For asymptotic normality of the TMLE, $\bQ_{1,K}$ need not equal the true conditional mean $\bQ_{0,K}$ as long as $\bQ_{n,K,t}^*$ converges to $\bQ_{1,K}$ satisfying the equicontinuity condition for the second term, $M_{2,n_t}$, as formalized below.

\begin{assumption}[Entropy integral of $\cbQ_{K}$]
\label{assumption:reasonable_covering_integral}
Let $\supcovering{\epsilon}$ denote the $\epsilon$-covering number in supremum norm of $\cbQ_K$. Define $\supintegral{\epsilon}:= \int_0^\epsilon \sqrt{\log \left(1+\supcovering{u} \right)} d u$ as the $\epsilon$-entropy integral in supremum norm of $\cbQ_K$.
Assume that $\supintegralTwo{\epsilon}{\cbQ_K}$ is finite and $\supintegralTwo{\epsilon}{\cbQ_K} \to 0$ as $\epsilon \to 0$.
\end{assumption}

\begin{assumption}[Convergence of $\bQ_{n,K,t}^*$]
\label{assumption:sigma_N_convergence}
For any $\bQ_K \in \cbQ_K$, define
\[
s_{n_t}\left(\bQ_K, \bQ_{1,K}\right):=\sqrt{\frac{1}{n_t} \sum_{i=1}^{n_t} E\left[\phi\left(\frac{|  D_k\left(\bQ_K\right)(O_i,C_i)-D_k\left(\bQ_{1,K}\right)(O_i,C_i)| }{\upsilon}\right) \bigm\vert C_i\right]},
\]
where $\phi(x):=e^x-x-1$, $\upsilon := 8(1-\zeta)/\zeta$.
Assume that $s_{n_t}\left(\bQ_{n,K,t}^*,\bQ_{1,K}\right) = O_P(\delta_{n_t})$ for some $\delta_{n_t} \to 0$.
\end{assumption}

\begin{theorem}[Equicontinuity]
\label{theorem:equicontinuity}
Under Assumptions 
\ref{assumption:strong_positivity},
\ref{assumption:reasonable_covering_integral}, and 
\ref{assumption:sigma_N_convergence},
$
M_{2,n_t}(\bQ_{n,K,t}^*, \bQ_{1,K}) = o_P\left(n_t^{-\frac{1}{2}}\right).
$
\end{theorem} 

This equicontinuity result demonstrates that, in addition to the asymptotic normality of $M_{1,n_t}$, the remainder term $M_{2,n_t}$ is negligible. The proof is given in Appendix~\ref{appendix:equicontinuity}, which depends on the equicontinuity result of the martingale process in general setup with sequential dependence (Appendix \ref{appendix:equicontinuity_general}).
We further provide a sufficient way to satisfy the equicontinuity condition by leveraging the recent work of \cite{van2023higher} that introduces a class of $m$-th order smoothness class of functions with $m$-th-order Radon-Nikodym derivative w.r.t. Lebesgue measure is càdlàg and of bounded variation ($m = 1,2,\ldots$), along with a higher-order spline Highly Adaptive Lasso (HAL) estimator (see Appendix \ref{appendix:HAL}).
We show that if $\bQ_{n,K,t}$ is estimated by higher-order HAL and converges to a $\bQ_{1,K}$ that is in the same class of functions, then the equicontinuity condition holds
(see Appendix~\ref{appendix:equicontinuity}).

\begin{restatable}{assumption}{repAssmpHigherOrderSpline}
\label{assumption:HigherOrderSpline}
    $\bQ_{1,K}$ and $\bQ_{n,K,t}$ are both from a class of functions that are represented by an $m$-th order smooth function whose $m$-th order Radon-Nikodym derivatives are càdlàg functions with bounded sectional variation norm ($m = 1, 2, \ldots $), and $\bQ_{n,K,t}$ is obtained by the $m$-th order Highly Adaptive Lasso estimator.
\end{restatable}
\begin{restatable}{corollary}{repCoroHighOrderSplineEquicontinuity}
\label{corollary:HighOrderSplineEquicontinuity}
Suppose Assumptions 
\ref{assumption:strong_positivity} and
\ref{assumption:HigherOrderSpline} hold. 
Then $M_{2,n_t}(\bQ_{n,K,t}^*, \bQ_{1,K}) = o_P\left(n_t^{-\frac{1}{2}}\right)$.
\end{restatable}

\begin{theorem}[Asymptotic normality of TMLE]
\label{asynormal_tmle}
    Suppose Assumptions 
    \ref{assumption:strong_positivity},
    \ref{assumption:stabilized_variance},
    \ref{assumption:reasonable_covering_integral}, and 
    \ref{assumption:sigma_N_convergence} hold, or alternatively, Assumptions 
    \ref{assumption:strong_positivity},
    \ref{assumption:stabilized_variance}, and 
    \ref{assumption:HigherOrderSpline}
    hold, then
    $
    \sqrt{n_t}(\hat{\Psi}_{t,k}-\Psi_{t,k})\dto N(0,\sigma^2_{t,k}).
    $
\end{theorem}

The asymptotic normality of TMLE is obtained directly from Theorems \ref{asynormal_first} and \ref{theorem:equicontinuity}.
We estimate the standard error by $\hsigma2_{t,k} := \frac{1}{n_t} \sum_{i=1}^{n_t}\left\{\frac{g^*_{k, i}\left(A_i \mid C_i\right)}{g_{0, i}\left(A_i \mid C_i\right)}\left[Y_{i, K}-\bQ_{n, K, t}^*\left(A_i, W_i\right)\right]\right\}^2$ (which converges in probability to $\sigma^2_{t,k}$), with the confidence interval  
$\left[\hat{\Psi}_{t,k}-z_{1-\alpha/2}\sqrt{\hsigma2_{t,k}/n_t}, \, \hat{\Psi}_{t,k}+z_{1-\alpha/2}\sqrt{\hsigma2_{t,k}/n_t} \right]$.

\section{Simulation Studies}
\label{section:simulations}

\textbf{Simulation setup.} 
We present simulation studies to evaluate the performance of our proposed adaptive design and estimation methods.\footnote{The code can be found at https://github.com/WenxinZhang25/TMLE-OSLAD.}
We consider a binary treatment $A \in \{0,1\}$, a univariate baseline covariate $W$ following a uniform distribution $U(-4,4)$, four surrogate outcomes $Y_k$ measured at follow-up times $k = 1, 2, 3, 4$, respectively, and a final outcome $Y_5$. 
Each participant's outcome $Y_k$ is generated by its conditional mean $E(Y_k|A,W)$ plus a random error following the normal distribution ($\sim N(0,1)$).
For detailed information on the simulations, we refer readers to Appendix \ref{appendix:CATE} for CATE estimation, Appendix \ref{appendix:details_of_CARA} for treatment randomization mechanism and 
Appendix \ref{appendix:additional_results} for data generating distributions.
We evaluate the performance of various designs in two scenarios.
In Scenario 1 (Figure \ref{fig:scenario1}), the optimal dynamic treatment rule based on the final outcome is more aligned with that defined in terms of later versus earlier surrogates. 
This implies that, despite requiring some time to measure and thus delaying adaptation in randomization probabilities, later surrogates are more useful in guiding treatment assignment probabilities.
In Scenario 2 (Figure \ref{fig:scenario2}), all surrogate outcomes have the same ODTR as that of the final outcome.
Notably, early surrogates have larger absolute CATE compared to later surrogates, suggesting that early surrogates are more sensitive to different treatments and potentially serve as superior indicators of the most beneficial treatment compared to the final outcome. 
In other words, in this scenario, earlier surrogates have the joint advantage of both timeliness and improved sensitivity for detecting treatment effect heterogeneity.  
\begin{figure}[H]
    \centering
    \begin{subfigure}[b]{0.4\textwidth}
        \includegraphics[width=\textwidth]{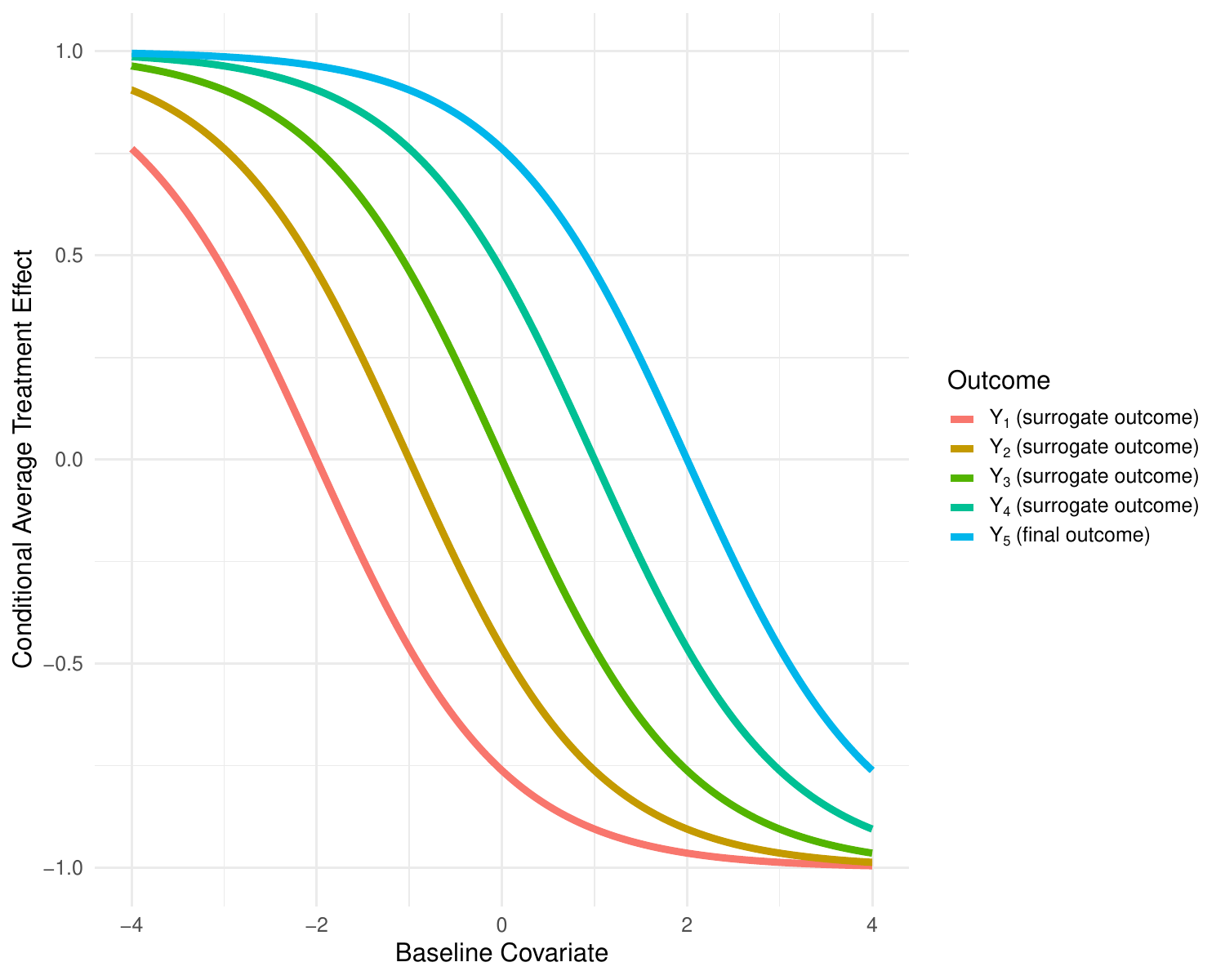}
        \caption{}
        \label{fig:scenario1}
    \end{subfigure}
    \hspace{0.05\textwidth}
    \begin{subfigure}[b]{0.4\textwidth}
        \includegraphics[width=\textwidth]{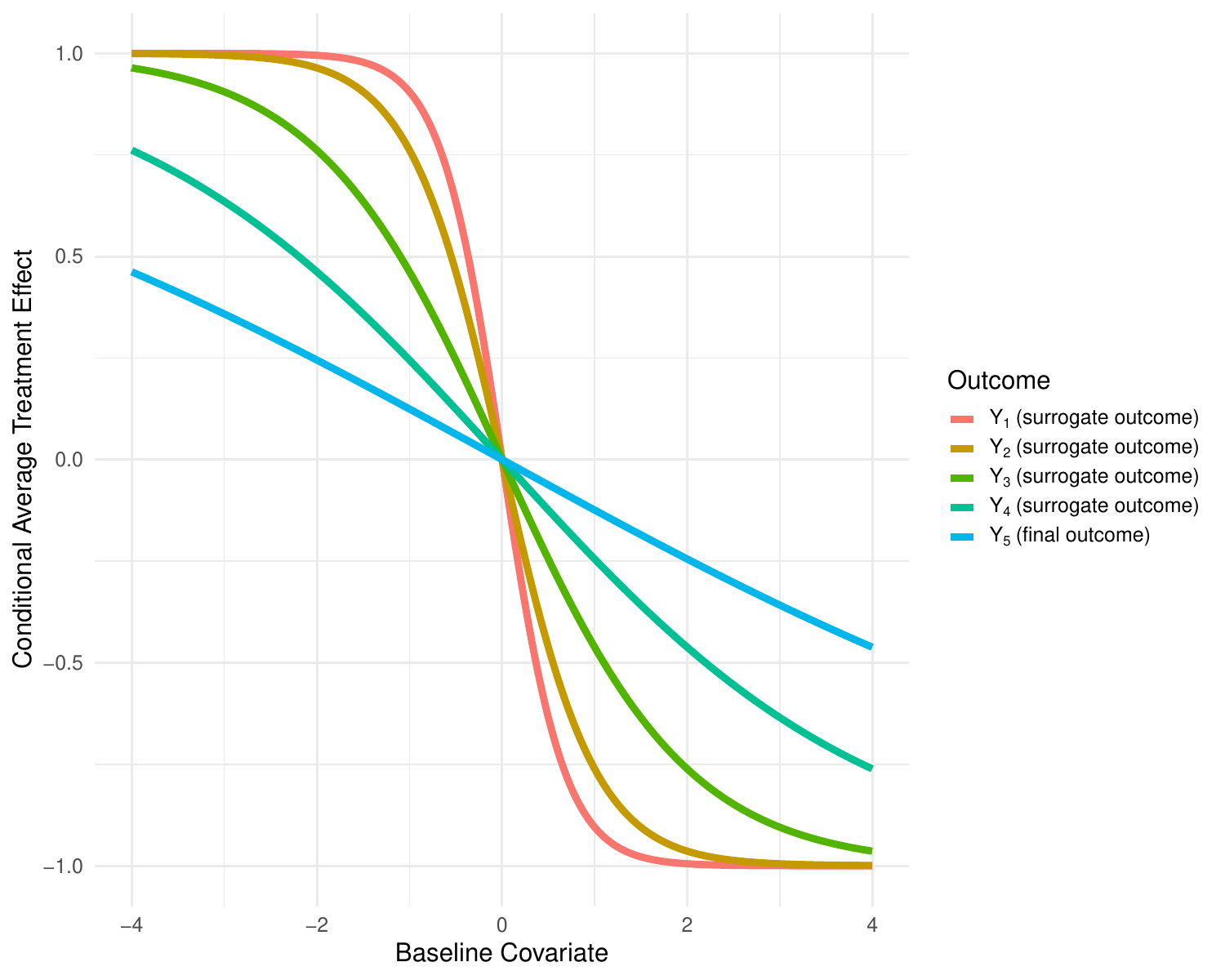}
        \caption{}
        \label{fig:scenario2}
    \end{subfigure}
    \caption{
    Conditional Average Treatment Effect (CATE) functions for the simulation scenarios: (a) Scenario 1, where later surrogates have higher utility; (b) Scenario 2, where earlier surrogates have higher utility. Each curve depicts the CATE function for a candidate outcome, $Y_1$ through $Y_5$.
    }
    \label{fig:scenarios}
\end{figure}

\looseness -1
We compare seven distinct designs with varying treatment randomization mechanisms over 500 Monte Carlo runs: (a) a simple non-adaptive randomized controlled trial (RCT) design, which allocates two treatments with equal probability at each time point; (b) five adaptive designs with randomization probabilities updated in response to an \textit{a priori} fixed choice of either a surrogate outcome $Y_k$ ($k$ = 1, ..., 4) or the final outcome $Y_K$ ($K=5$); (c) the proposed TMLE-OSLAD adaptive design that dynamically evaluates and selects which outcome (among $Y_1, \ldots, Y_5$) is used to update treatment randomization probabilities at each time point.
For each design, we run experiments for $T = 50$ time points with 50 subjects enrolled per time point. The maximal treatment assignment probability is capped at $90\%$.

\textbf{Surrogate evaluation and TMLE performance.}
We first evaluate if TMLE-OSLAD can consistently identify the oracle surrogate, which we define as the candidate with the highest true $\Psi_{t,k}$ at time $t$ among all $Y_k$ (in other words, the surrogate that if used would result in the largest expected primary outcomes for all enrolled participants in the experiment). In particular, if $\Psi_{t,K}$ is the largest, then the primary outcome $Y_K$ itself is the oracle.
Table \ref{tab:truePsi} presents the average of true $\Psi_{t,k}$ values at illustrative time points 11, 21, 31, 41, and 50 for both scenarios across Monte Carlo runs. The expected outcome of $Y_5$ under the non-adaptive randomization is also reported, denoted by $\Psi_{t,RCT}$.
In Scenario 1, the oracle surrogate gradually shifts towards later surrogates as time elapses and more data from these surrogates become available to guide the adaptation of treatment randomization probabilities. 
Initially, surrogates $Y_3$ and $Y_4$ have the highest value at times 11 and 21, reflecting the balance they achieve between earlier availability and degree to which they reflect the true optimal individualized treatment assignment. While these surrogates provide benefit in the short term, they do not consistently predict the correct optimal treatment for the primary outcome for all individuals, leading them to be eventually outperformed by designs that utilize the true primary outcome $Y_5$ as additional follow-up time is accrued.
This is expected, as the ultimate assessment of a surrogate’s effectiveness as follow-up time extends hinges on its ability to predict the optimal personalized treatment for the primary outcome. 
In Scenario 2, as expected, surrogate $Y_1$ remains the most useful (i.e. is the oracle surrogate) throughout, due to its early availability, and the fact that it correctly identifies the true optimal individualized treatment for the long-term primary outcome. In contrast, later surrogates $Y_4$ and $Y_5$ have lower scores. This disparity can be attributed to the fact that while $Y_4$ and $Y_5$ will eventually serve equally well to identify the optimal personalized treatment as time since trial initiation goes to infinity, their utility is limited under finite trial durations because they have been observed for fewer participants. 
Across both scenarios, the TMLE estimates for these estimands exhibit low bias and variance, and their confidence intervals obtain nominal coverage (see Tables \ref{tab:bias}, \ref{tab:variance}, \ref{tab:coverage} in  Appendix~\ref{appendix:additional_results}).

\begin{table}[h]\centering
    \caption{True target estimand $\Psi_{t,k}$ had participants enrolled from $1, \cdots, t-K$ were adaptively treated by a candidate surrogate $Y_k$ in Scenarios 1 and 2 ($t = 11,21,31,41,50$). Expected outcomes under oracle surrogate are shown in bold.}
    \label{tab:truePsi}
    \footnotesize
    \begin{minipage}[t]{0.48\linewidth}
        \centering\footnotesize (a) Scenario 1\par\vspace{1mm}
        \resizebox{0.96\linewidth}{5.2cm}{
        \begin{tabular}{*{7}{c}}
            \toprule
           $t$ & $\Psi_{t,RCT}$ & $\Psi_{t,1}$ &  $\Psi_{t,2}$ & $\Psi_{t,3}$ & $\Psi_{t,4}$ & $\Psi_{t,5}$ \\
            \midrule
            11 & 0.000 & -0.016 & 0.053 & \textbf{0.078} & 0.066 & 0.033 \\
            21 & 0.000 & -0.008 & 0.083 & 0.145 & \textbf{0.177} & 0.173 \\
            31 & 0.000 & -0.005 & 0.091 & 0.161 & 0.203 & \textbf{0.208} \\
            41 & 0.000 & -0.002 & 0.094 & 0.169 & 0.215 & \textbf{0.225} \\
            50 & 0.000 & -0.001 & 0.096 & 0.173 & 0.222 & \textbf{0.234} \\
            \bottomrule
        \end{tabular}
    }
    \vfill
    
\end{minipage}\hfill
    \begin{minipage}[t]{0.48\linewidth}
        \centering\footnotesize (b) Scenario 2\par\vspace{1mm}
        \resizebox{0.96\linewidth}{5.2cm}{
        \begin{tabular}{*{7}{c}}
            \toprule
           $t$ & $\Psi_{t,RCT}$ & $\Psi_{t,1}$ &  $\Psi_{t,2}$ & $\Psi_{t,3}$ & $\Psi_{t,4}$ & $\Psi_{t,5}$ \\
            \midrule
            11 & 0.000 & \textbf{0.073} & 0.057 & 0.039 & 0.019 & 0.004 \\
            21 & 0.000 & \textbf{0.087} & 0.081 & 0.073 & 0.061 & 0.038 \\
            31 & 0.000 & \textbf{0.090} & 0.086 & 0.081 & 0.073 & 0.052 \\
            41 & 0.000 & \textbf{0.091} & 0.089 & 0.085 & 0.078 & 0.060 \\
            50 & 0.000 & \textbf{0.092} & 0.090 & 0.087 & 0.081 & 0.065 \\
            \bottomrule
        \end{tabular}
        }
        \vfill
    \end{minipage}\hfill
\end{table}

\looseness -1
Figures \ref{fig:sim1_OSLAD_selection}--\ref{fig:sim2_OSLAD_selection} show the frequency with which each candidate design is selected by TMLE-OSLAD over time in Scenarios 1 and 2, respectively, while Figures \ref{fig:sim1_OSLAD_example}--\ref{fig:sim2_OSLAD_example} present a single-run example illustrating the selection mechanism. Without prior knowledge of relationships between surrogate and primary outcomes, 
the TMLE-OSLAD procedure is able to identify the oracle surrogate in each setting.
In Scenario 1, where later surrogate outcomes serve as better proxies for the final outcome
when determining optimal individualized treatments, TMLE-OSLAD progressively shifts its selections toward these later surrogates, ultimately favoring $Y_4$ and $Y_5$ most frequently. 
In contrast, in Scenario 2 where the earliest surrogate $Y_1$ remains the oracle surrogate throughout,
TMLE-OSLAD predominantly selects early surrogates, with $Y_1$ most frequently selected, followed by $Y_2$ and $Y_3$, while $Y_4$ and $Y_5$ are rarely chosen. 
\begin{figure}[H]
    \centering
    \begin{subfigure}[t]{0.4\textwidth}
        \centering
        \includegraphics[width=\textwidth]{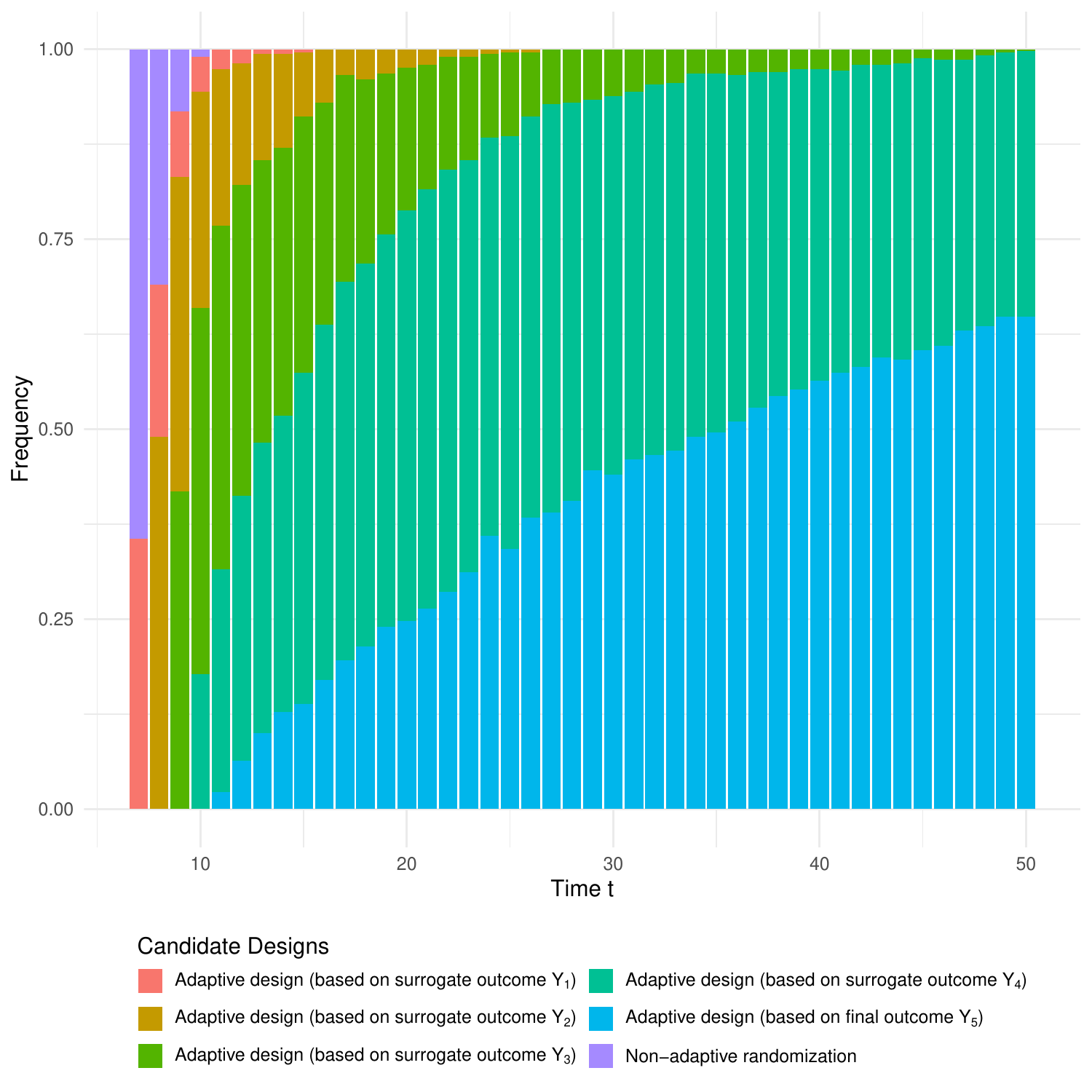}
        \caption{} 
        \label{fig:sim1_OSLAD_selection}
    \end{subfigure}
    \begin{subfigure}[t]{0.4\textwidth}
        \centering
        \includegraphics[width=\textwidth]{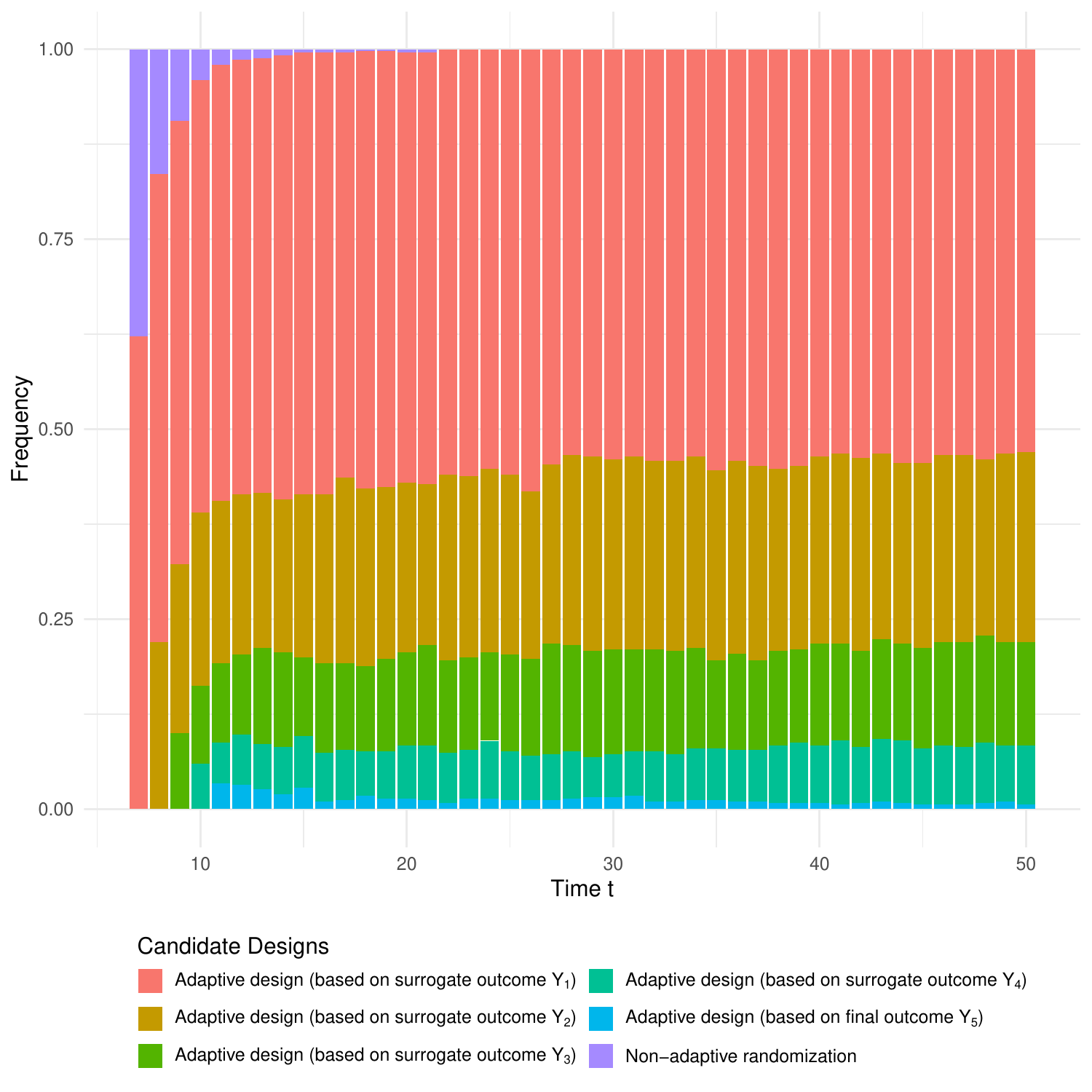}
        \caption{} 
        \label{fig:sim2_OSLAD_selection}
    \end{subfigure} \\
    \begin{subfigure}[t]{0.4\textwidth}
        \centering
        \includegraphics[width=\textwidth]{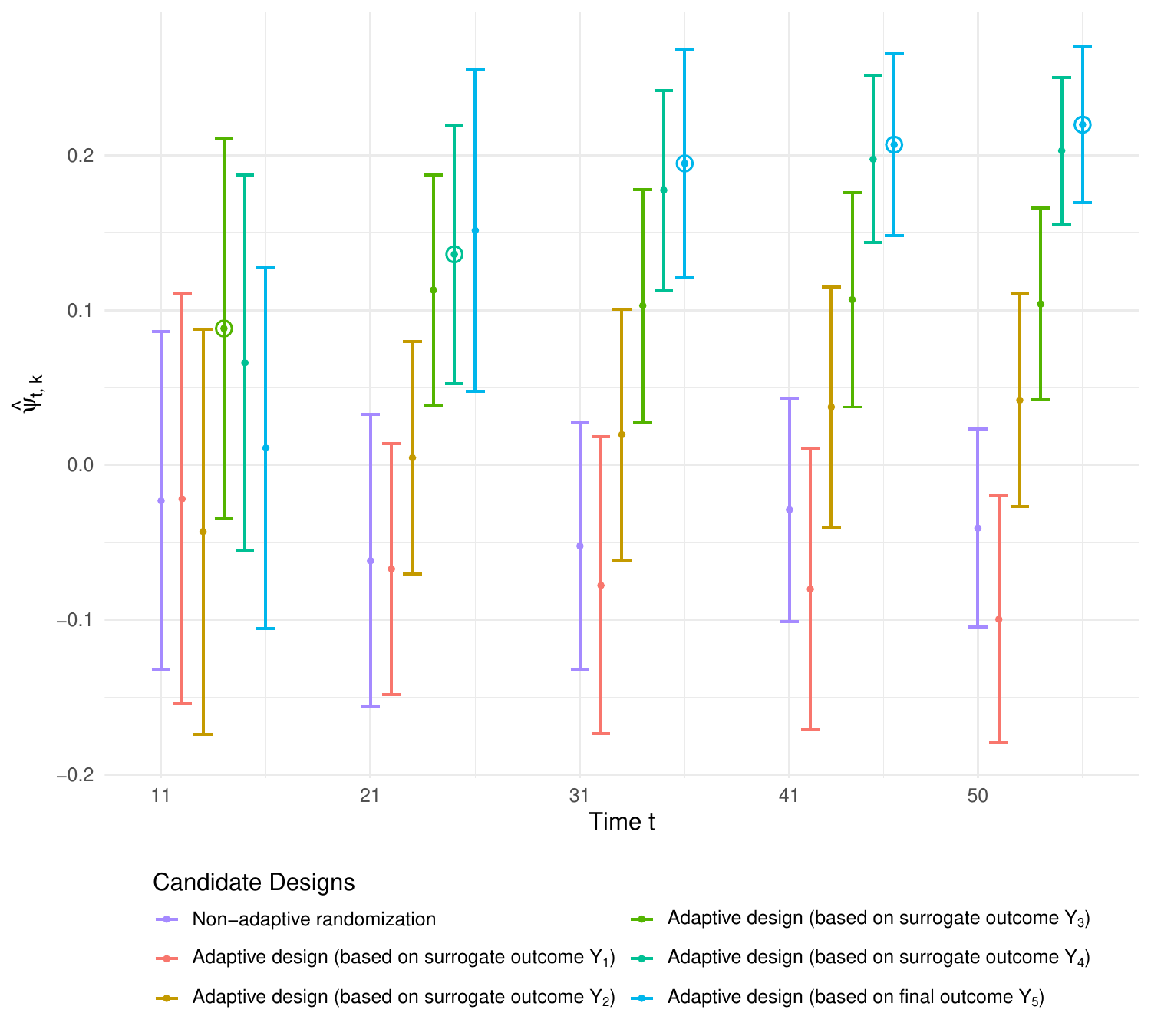}
        \caption{} 
        \label{fig:sim1_OSLAD_example}
    \end{subfigure}
    \begin{subfigure}[t]{0.4\textwidth}
        \centering
        \includegraphics[width=\textwidth]{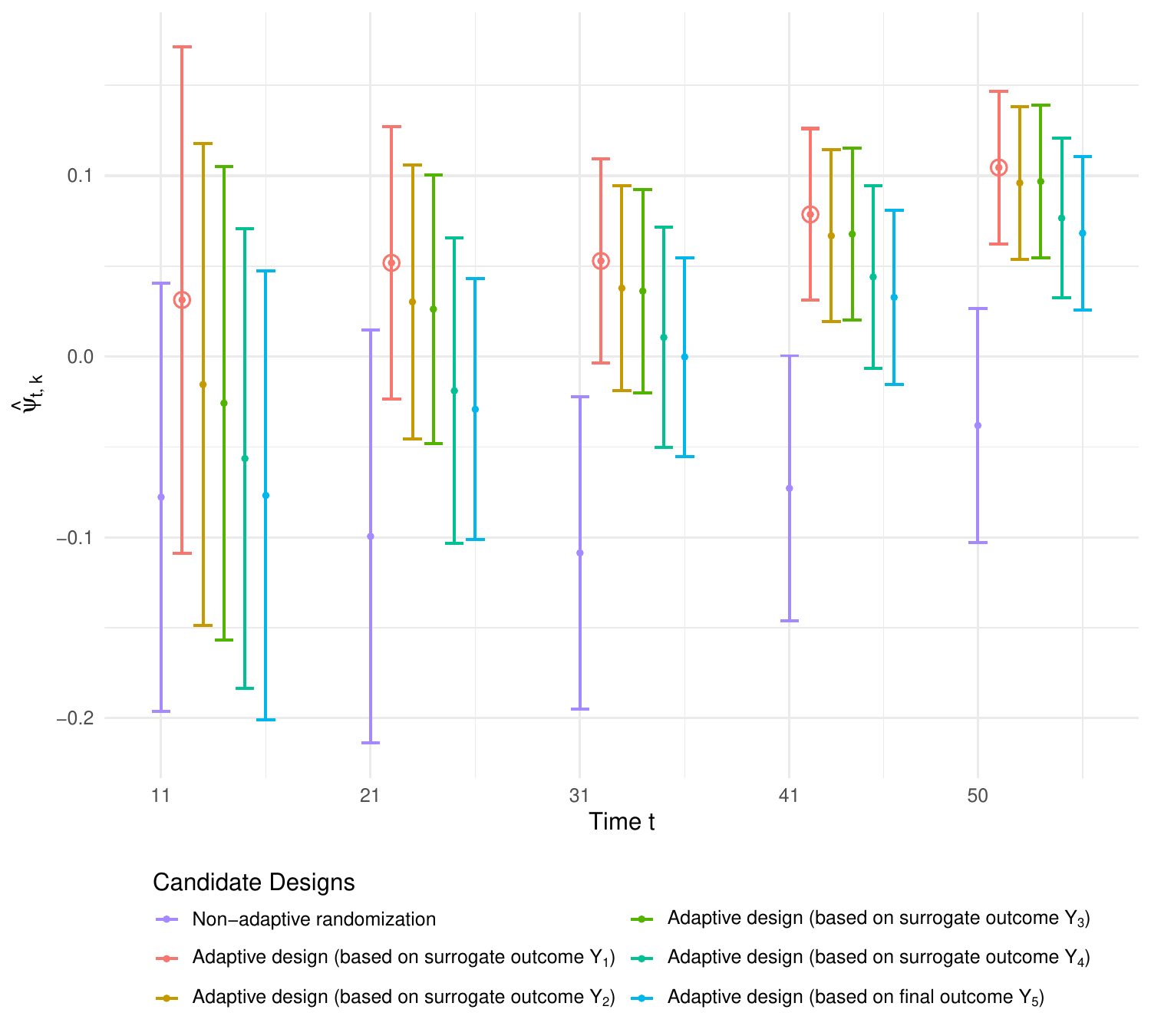}
        \caption{} 
        \label{fig:sim2_OSLAD_example}
    \end{subfigure}
    \caption{Frequency of design selection by TMLE-OSLAD at each time point. 
    Panels (a) and (b) correspond to Scenario~1 and Scenario~2, respectively. 
    Each color represents one candidate design, including adaptive designs based on outcomes $Y_k$ ($k \in \{1,\ldots,5\}$) and the non-adaptive design. The height of each bar indicates the proportion of simulation runs 
    in which that design attains the highest lower bound of the 95\% confidence interval for $\hat{\psi}_{t,k}$ among all candidate designs at time $t$, reflecting the selection mechanism of TMLE-OSLAD.  Panels (c) and (d) display the selection mechanism at illustrative time points  ($t = 11, 21, 31, 41, 50$) in a single simulation run under Scenarios~1 and~2, respectively. Colored points show the TMLE point estimates of $\psi_{t,k}$ (y-axis) for evaluating each candidate design at each time point (x-axis). Vertical bars indicate the corresponding 95\% Wald confidence intervals. At each time point, the selected design (defined as the one with the largest lower confidence bound among all candidate designs) is highlighted by a circular outline. }\label{fig:SLselection}
\end{figure}

\textbf{Performance of benefiting participants.}
We use $d_{0,K}: \cW \to \cA$ to denote the true optimal individualized treatment rule for $Y_K$, which is defined as $d_{0,K}(W) := I(\bQ_{0,K}(1,W)-\bQ_{0,K}(0,W)>0)$. This rule assigns to each individual the treatment that would maximize their expected final outcome, given their covariates.
To evaluate the extent to which each design assigns treatments that match the underlying oracle optimal rule, we define probability (\%) of not assigning optimal individualized treatments as $\Pnonopt_{t} := \frac{1}{E(t)}\sum_{i: t_i = t}I\left(A_i \ne d_{0,K}(W_i)\right)$,
where $E(t)$ is the number of participants enrolled at time $t$. 
Similarly, $\Pnonopt_{t,RCT}$ and $ \Pnonopt_{t,SL}$ denote this probability for the non-adaptive RCT and the proposed TMLE-OSLAD adaptive design, respectively. Table \ref{tab:incorrectTreatment} in Appendix~\ref{appendix:additional_results} compares these suboptimal treatment assignment probabilities across different designs.
We also define the average regret for participants enrolled at time $t$ within an adaptive design as
$
R_{t} := \frac{1}{E(t)}\sum_{i: t_i = t}I\left(A_i \ne d_{0,K}(W_i)\right) \times |\bQ_{0,K}(1,W_i)-\bQ_{0,K}(0,W_i)|.
$
Here, each participant \( i \) contributes the absolute difference in expected final outcomes between the optimal and actual treatments. A lower $R_{t}$ indicates a better design that brings benefit in $Y_K$ for participants enrolled at time $t$.
We use \( R_{t,k} \) to denote the average regret within an adaptive design responsive to \( Y_k \) at time \( t \), while \( R_{t,\textit{RCT}} \) and \( R_{t,\textit{SL}} \) represent the average regret for the RCT design and TMLE-OSLAD design, respectively.
Table \ref{tab:regret} reports average regret across simulations at \( t = 11, 21, 31, 41, \) and \( 50 \), with full trajectories shown in Figure~\ref{fig:regret}.

In Scenario 1, where later surrogates are superior at identifying the optimal individualized treatment, the adaptive design responsive to $Y_5$, as expected, consistently demonstrates the lowest regret across the experiment. 
In contrast, in Scenario 2, the regret within the adaptive design that responds to $Y_1$ remains consistently the lowest. 
There is a noticeable trend that the adaptive design in response to the final outcome $Y_5$ exhibits a slower convergence toward the regret of the oracle surrogate $Y_1$, with a substantial discrepancy in regret that persists even at $t=50$.
Of course, in practice, the oracle surrogate is not known \textit{a priori}, motivating the use of the TMLE-OSLAD procedure to evaluate and select among candidate surrogates in adaptive design.
Across both scenarios, the proposed adaptive design demonstrates stable efficacy in approximating the minimal achievable regret compared to other designs. 
Its regret ranks second lowest starting at $t = 23$ in Scenario 1, and matches the lowest regret trajectory of $Y_1$ following $t = 21$ in Scenario 2. 
This robust performance reflects the advantage of the real-time evaluation and selection mechanism in TMLE-OSLAD, which continuously assesses each candidate surrogate-guided adaptive design, selects the most advantageous one to guide treatment randomization in subsequent stages.

\textbf{Post-experiment inference for other causal estimands}.
Table \ref{tab:estimation_at_end} reports TMLE performance in estimating the additional causal estimands introduced in Section \ref{section:inference_at_the_end}, including the average treatment effect and the mean final outcome under the estimated optimal dynamic treatment rules that maximize surrogate outcomes and the primary outcome.
The mean final outcome $Y_5$ under the estimated optimal dynamic treatment rule for $Y_5$ is estimated using five-fold CV-TMLE.
Across all estimands, TMLE provides robust performance with low bias and variance, and nominal coverage of its confidence intervals.
\begin{table}[h]\centering
    \caption{Performance of TMLEs for other causal estimands at the end of experiments in Scenarios 1 and 2 across Monte Carlo simulations.}
    \label{tab:estimation_at_end}
    \setlength{\extrarowheight}{2pt}
    \makebox[\textwidth]{\footnotesize (a) Scenario 1} 
    \resizebox{0.96\textwidth}{!}{
    \footnotesize
    \renewcommand{\arraystretch}{1.16}
    \begin{tabular}{*{6}{c}}
        \toprule
        Estimand & Notation & Truth & Bias ($\times 10^{-3}$) & Var ($\times 10^{-3}$) & Coverage  (\%)\\
        \midrule
        Average treatment effect on the final outcome $Y_5$ & $\psi^{\text{ATE}}_{\Tend} $ & 0.468 & 3.94 & 3.92 & 95.8 \\
        Expected final outcome $Y_5$ under the estimated ODTR based on surrogate $Y_1$ & $\psi_{\Tend}^{d^*_{n,1}}$ & 0.011 & -1.01 & 2.85 & 94.4 \\
        Expected final outcome $Y_5$ under the estimated ODTR based on surrogate $Y_2$ & $\psi^{d^*_{n,2}}_{\Tend}$ & 0.130 & -0.87 & 2.02 & 93.8 \\
        Expected final outcome $Y_5$ under the estimated ODTR based on surrogate $Y_3$ & $\psi^{d^*_{n,3}}_{\Tend}$ & 0.235 & -0.08 & 1.62 & 93.6 \\
        Expected final outcome $Y_5$ under the estimated ODTR based on surrogate $Y_4$ & $\psi^{d^*_{n,4}}_{\Tend}$ & 0.312 & -1.01 & 0.89 & 94.6 \\
        Expected final outcome $Y_5$ under the estimated ODTR based on $Y_5$ & $\psi^{d^*_{n,5}}_{\Tend}$ & 0.340 & 0.42 & 0.72 & 95.6 \\ \bottomrule
    \end{tabular}
    }
    \vspace{0.3em} 
    
    \makebox[\textwidth]{\footnotesize (b) Scenario 2} 
    \resizebox{0.96\textwidth}{!}{
    \footnotesize
    \renewcommand{\arraystretch}{1.16}
    \begin{tabular}{*{6}{c}}
        \toprule
        Estimand & Notation & Truth & Bias ($\times 10^{-3}$) & Var ($\times 10^{-3}$) & Coverage (\%)\\
        \midrule
        Average treatment effect on the final outcome $Y_5$ & $\psi_{\Tend}^{\text{ATE}}$ & 0.000 & -1.17 & 3.99 & 94.8\\
        Expected final outcome $Y_5$ under the estimated ODTR based on surrogate $Y_1$ & $\psi^{d^*_{n,1}}_{\Tend}$ & 0.120 & -0.25 & 0.57 & 93.8 \\
        Expected final outcome $Y_5$ under the estimated ODTR based on surrogate $Y_2$ & $\psi^{d^*_{n,2}}_{\Tend}$ & 0.120 & -0.32 & 0.59 & 94.2 \\
        Expected final outcome $Y_5$ under the estimated ODTR based on surrogate $Y_3$ & $\psi^{d^*_{n,3}}_{\Tend}$ & 0.120 & -0.57 & 0.61 & 93.6 \\
        Expected final outcome $Y_5$ under the estimated ODTR based on surrogate $Y_4$ & $\psi^{d^*_{n,4}}_{\Tend}$ & 0.119 & -0.38 & 0.60 & 94.6 \\
        Expected final outcome $Y_5$ under the estimated ODTR based on $Y_5$ & $\psi^{d^*_{n,5}}_{\Tend}$ & 0.114 & 0.31 & 0.81 & 93.0 \\ \bottomrule
    \end{tabular}
    }
\end{table}
\section{Application}
\label{section:real_sim}

We conduct an extended simulation study inspired by the Adaptive Strategies for Preventing and Treating Lapses of Retention in HIV Care (ADAPT-R) trial \citep{geng2023adaptive}, a sequential multiple assignment randomized (SMART) trial designed to evaluate adaptive strategies for improving retention in care among individuals initiating HIV treatment. We demonstrate our proposed TMLE-OSLAD framework in this context, showing how it utilizes effective surrogates to benefit participants and evaluating TMLE performance in this setup.

\textbf{Simulation setup}.
\label{subsection:realSim_setup}
We focus on Phase 2 of the ADAPT-R trial, where participants who had a lapse in care in Phase 1 (defined as missing a clinic visit by 14 or more days) were re-randomized to a re-engagement strategy: either a peer navigation strategy (Navigator) or a combined Short Message Service reminder of appointments and conditional cash transfer (SMS+CCT). 
We discretize the study timeline into 50-day stages and define outcomes as the proportion of time in care during which a participant is engaged in across five windows, from the first day of receiving the re-engagement strategy to five specific time points: 50, 100, 150, 200, and 250 days. 
These outcomes are labeled as $Y_1$, $Y_2$, $Y_3$, $Y_4$, and $Y_5$, respectively, with $Y_5$ (proportion of time in care in 250 days) the primary outcome, and $Y_1$ through $Y_4$ the surrogate outcomes. 
We consider two covariates in the conditioning set of heterogeneous treatment effect: (a) initial Phase 1 treatment arm: SMS, CCT, or standard of care (SOC) and (b) time from Phase 1 enrollment to initial lapse in care. We used data from 215 participants who were assigned either Navigator or SMS+CCT in Phase 2 after a lapse in care in Phase 1 and fit a generalized linear model to obtain the counterfactual outcomes for simulation.
The CATE functions for the surrogate outcomes and the final outcome are presented in Figure~\ref{fig:CATE-ADAPT-R} (see Appendix \ref{appendix:real_sim} with more details). 
We implement a series of CARA designs as outlined in Section~\ref{section:simulations}, including a simple sequential RCT, five adaptive designs that use a fixed choice from 
$Y_1$ through $Y_5$, and the proposed TMLE-OSLAD adaptive design that evaluates and selects the best surrogate in real time. 
We also compared their performance with a naive policy that employs the latest available outcome at each time point to adapt treatment randomization probabilities.
In each simulation, participants arrive following the enrollment schedule in the original trial with their observed covariates, and treatment assignments are generated sequentially according to the type of the simulated design and the covariate, treatment and outcome data of previously enrolled participants, then the outcomes are generated in response to the assigned treatment. This construction ensures that each observation is conditionally independent given the past. We simulate each design 500 times to assess the performance of TMLEs and different adaptive designs.

\textbf{Results.}
\label{subsection:realSim_results}
\looseness -1
We compare regret over time obtained from implementing different design strategies (Figure \ref{fig:realSim_abs_regret} in Appendix \ref{appendix:real_sim}). The adaptive design using the earliest outcome $Y_1$ outperforms those based on later outcomes or the naive policy from $t=2$ through $t=11$.
After $t=12$, adaptive designs using later outcomes start achieving smaller regret for new subjects than the earliest outcome, as the underlying true CATE functions of later surrogates are closer to that for the final outcome and are estimated more precisely as data accumulate. However, this delayed improvement does not compensate the overarching benefit provided by the earliest surrogate, which delivers prompt benefits to participants by adopting surrogate-guided optimal personalized treatments learned from earlier and more abundant data.
For the naive policy, although it initially matches the best performing adaptive design based on $Y_1$, its aggressive shift to the most recent outcomes leads it to behave similarly to the suboptimal later-outcome designs, resulting in higher regret.
In contrast, TMLE-OSLAD predominantly selects the earliest surrogates ($Y_1$ and $Y_2$) over time (Figure \ref{fig:realSim_SLselectionA}). It yields a higher expected final outcome (Figure \ref{fig:realSim_SLselectionB}) than designs that adhere to later outcomes without ongoing evaluation and approaches the performance of the oracle design using the earliest surrogate.
TMLE results for the proposed estimands during the experiment and for additional estimands after the experiment are reported in Appendix~\ref{appendix:real_sim}, showing robust performance.
\begin{figure}[h]
    \centering
    \begin{subfigure}[b]{0.42\textwidth}
        \centering
        \includegraphics[width=\textwidth]{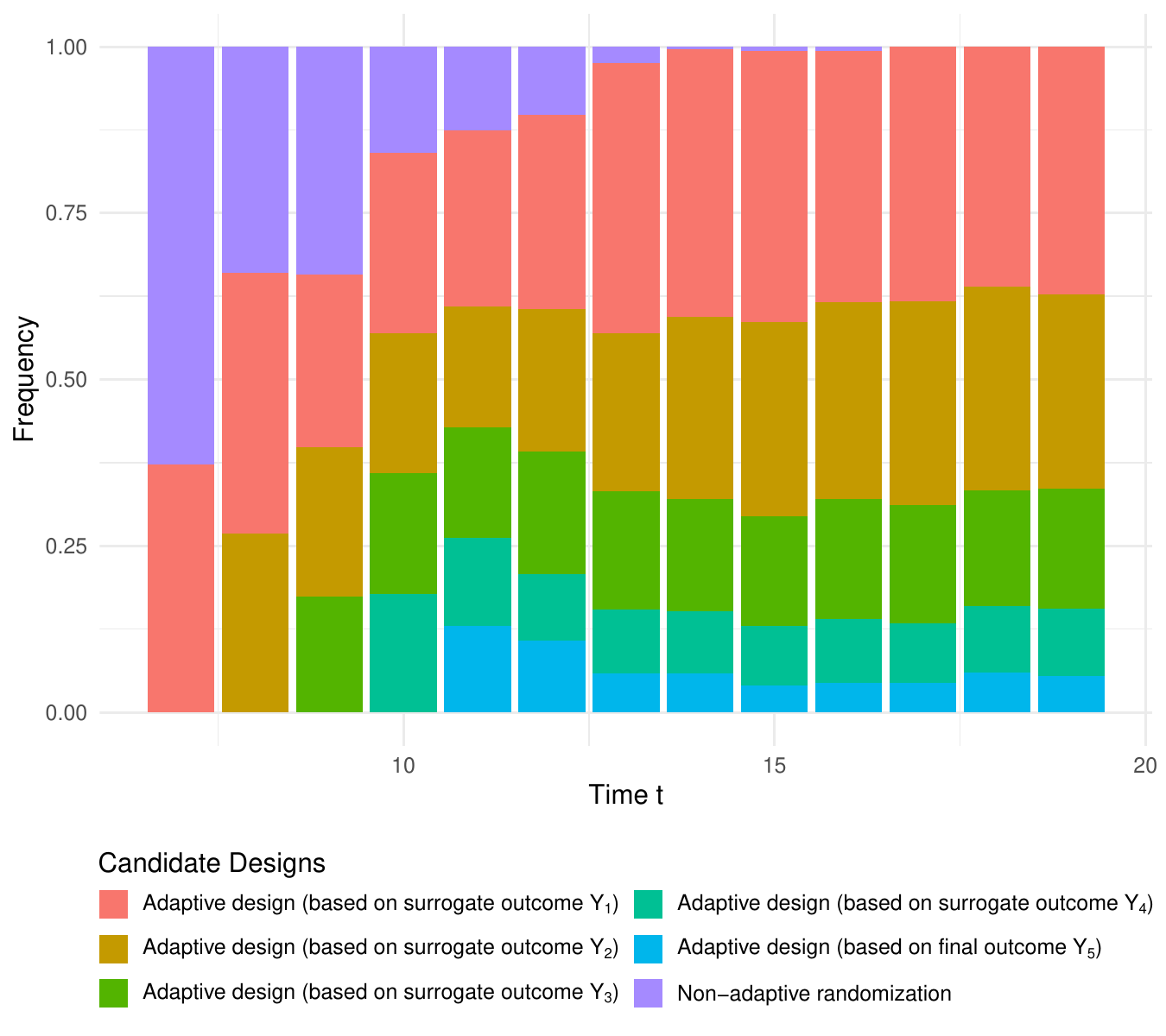}
        \caption{} 
        \label{fig:realSim_SLselectionA}
    \end{subfigure}
    \hspace{0.05\textwidth}
    \begin{subfigure}[b]{0.46\textwidth}
        \centering
        \includegraphics[width=\textwidth]{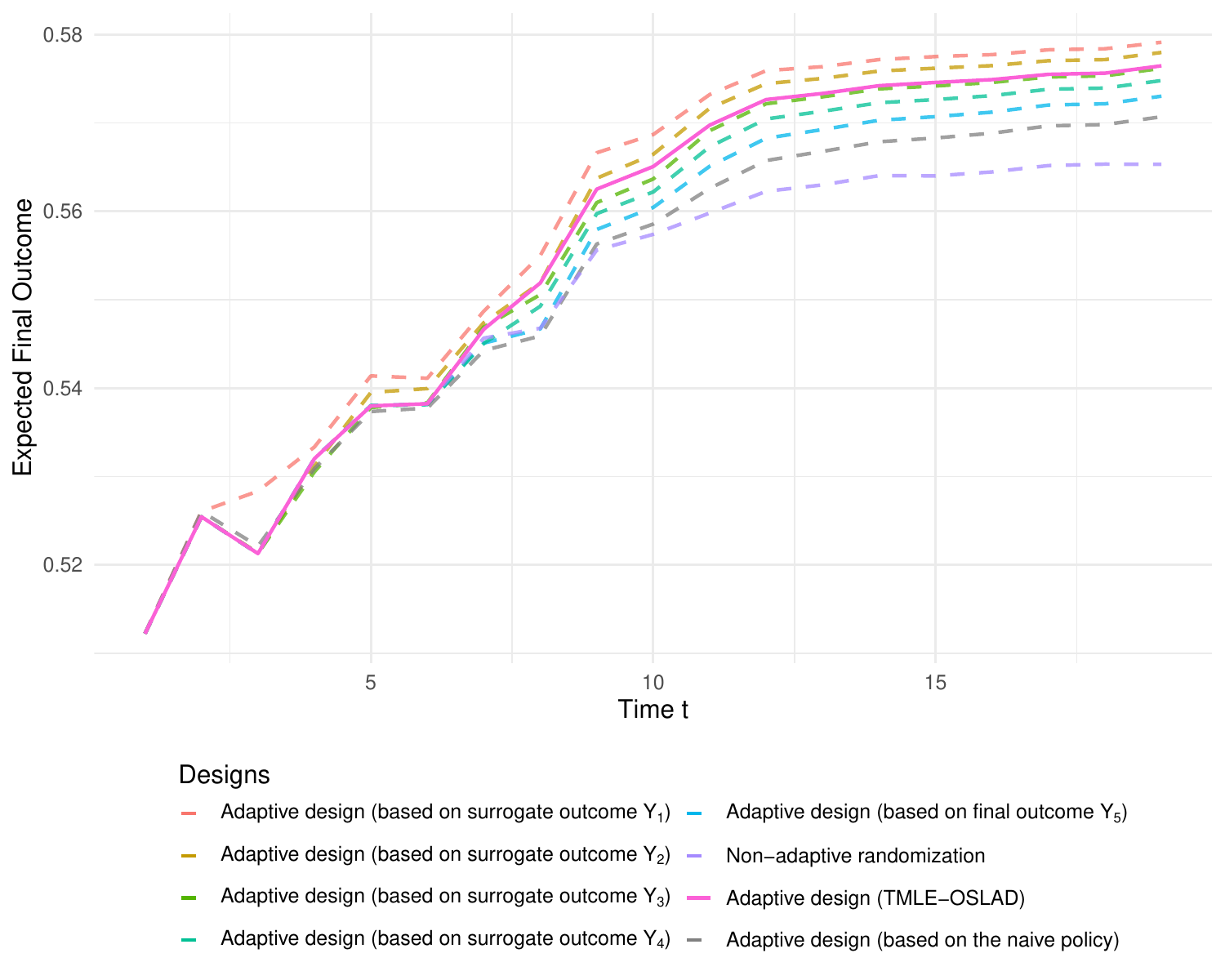}
        \caption{} 
        \label{fig:realSim_SLselectionB}
    \end{subfigure}
    \caption{Panel (a) shows the frequency of candidate adaptive design selections by TMLE-OSLAD at each time point in the ADAPT-R trial simulation. 
    Panel (b) shows the cumulative average expected final outcome over time obtained from implementing different designs.}
    \label{fig:realSim_SLselection}
\end{figure}

\section{Conclusions}
\label{section:discussions}
In summary, this work provides a general TMLE-based Online Super Learner Adaptive Design Framework (TMLE-OSLAD) for constructing a real-time, data-driven ensemble system for adaptive experimental designs and shows how this framework contributes to the statistical theory and application of surrogate outcomes in adaptive designs. 
We propose a novel causal estimand for surrogate utility in sequential adaptive designs: the mean primary outcome under a counterfactual CARA design where, at enrollment, each participant’s randomization is shifted toward the optimal individualized rule learned from a candidate surrogate. It captures both how well the surrogate guides optimal treatment assignment and the extent to which it enables earlier benefit for more participants, with the goal of improving participants' primary outcomes.
Building on this, we define an Online Super Learner adaptive design that incorporates the proposed estimand to evaluate a library of candidate surrogates, selects the most useful surrogate based on TMLE inference, and dynamically adapts randomization probabilities by the selected surrogate-guided design. 

This work on surrogate outcomes is a concrete application of our highly general TMLE-OSLAD framework, which provides statistical inference on the mean counterfactual outcome under each candidate adaptive design in a user-supplied set, and adaptively evaluates and updates the best-performing design over time to adapt treatment randomization for new participants. 
This unified framework enables principled evaluation of the opportunities and costs of diverse adaptive designs varying in their design choices, quantifies how accumulating information benefits participants over time, and supports dynamic navigation of trade-offs among multiple designs during an ongoing experiment. 

Methodologically, this framework advances non-parametric statistical inference for adaptive experiments: the proposed TMLE approach yields consistent, asymptotically normal estimators for the proposed estimand and other classical causal estimands, leveraging martingale theories, equicontinuity results, and Highly Adaptive Lasso. It broadly applies to estimation of various estimands in sequentially dependent settings like adaptive designs and time series.

Beyond the proposed context-specific estimand, further study of the computationally intensive ``ideal estimand'', the expected outcome under the full data distribution induced by a candidate design (as is discussed in Appendix~\ref{appendix:ideal_estimand}), could provide more insight into the strengths and limitations of different estimands, as well as trade-offs between statistical properties and computational efficiency. We leave this to future work.

\bibliography{TMLE-OSLAD}

\appendix
\renewcommand{\thefigure}{S\arabic{figure}}
\setcounter{figure}{0}

\renewcommand{\thetable}{S\arabic{table}}
\setcounter{table}{0}

\newpage

\renewcommand{\thepage}{S-\arabic{page}}
\setcounter{page}{1}

\renewcommand{\thetheorem}{E.\arabic{theorem}}
\renewcommand{\theassumption}{E.\arabic{assumption}}
\renewcommand{\thelemma}
{E.\arabic{lemma}}
\renewcommand{\thedefinition}{E.\arabic{definition}}
\renewcommand{\thecorollary}{E.\arabic{corollary}}
\renewcommand{\theproposition}{E.\arabic{proposition}}

\setcounter{theorem}{0}
\setcounter{assumption}{0}
\setcounter{lemma}{0}
\setcounter{definition}{0}
\setcounter{corollary}{0}
\setcounter{proposition}{0}

\renewcommand{\theHtheorem}{appendix.\thesection.\arabic{theorem}}
\renewcommand{\theHassumption}{appendix.\thesection.\arabic{assumption}}
\renewcommand{\theHlemma}{appendix.\thesection.\arabic{lemma}}
\renewcommand{\theHdefinition}{appendix.\thesection.\arabic{definition}}
\renewcommand{\theHcorollary}{appendix.\thesection.\arabic{corollary}}
\renewcommand{\theHproposition}{appendix.\thesection.\arabic{proposition}}



\section{Additional Remarks on the Target Estimand}
\label{appendix:diff_of_two_target_estimands}

\subsection{Surrogate Evaluation with Adaptive Design Perspective}
\label{appendix:remark_utility_of_surrogates}
Our approach adopts an adaptive design perspective to evaluate surrogates in sequential decision-making contexts.
The proposed target estimand evaluates the utility of surrogates by their early availability, predictability, and sensitivity in indicating effective treatments in the following sense. 
We define a trajectory of treatment rules $\gtilde^*_{k,1:t} := 
(\gtilde^*_{k,1},\ldots,\gtilde^*_{k,t})$, with each rule $\gtilde^*_{k,t'}$ generating a personalized treatment randomization function $g^*_{k,i}$ for each participant $i$ enrolled at time $t'$. Every $g^*_{k,i}$ is specifically tailored to maximize the surrogate outcome $Y_k$ based on their baseline covariates $W_i$ and context.
We note that at each time point $t'$, $\gtilde^*_{k,t'}$ depends only on historical data $\bcH(t')$.
Therefore, $\gtilde^*_{k,t'}$ can not deviate from a static design with equal randomization probabilities during the initial $k$ time points, as observation of $Y_k$ is not yet available.
Once $Y_k$ is available at $t' = k+1$, $\gtilde^*_{k,t'}$ starts moving towards its optimal dynamic treatment rule $d_{0,k}(W):= I\left(\bQ_{0,k}(1,W)-\bQ_{0,k}(0,W)>0\right)$ over time.
If $d_{0,k}$ closely aligns with the true optimal rule for the primary outcome $d_{0,K}$, then the surrogate $Y_k$ demonstrates the advantage of its early availability in informing which treatment is better in the context of sequential decision making, compared to later outcomes.
The progression of $\gtilde^*_{k,t'}$ from a static design with equal randomization probabilities towards the optimal dynamic treatment rule $d_{0,k}$ depends on the magnitude of treatment superiority (as measured by the point estimate of CATE) and its uncertainty (as reflected by estimated standard error of CATE).
Therefore, if a surrogate $Y_k$ exhibits good alignment with the primary outcome in terms of optimal dynamic treatment rule and presents a strong signal of individualized optimal treatment over other surrogates,
this surrogate will be preferred and given a higher $\Psi_{t,k}$. 
Of course, even if a surrogate demonstrates early availability and sensitivity, it will present a lower $\Psi_{t,k}$ value if it fails to predict optimal personalized treatments of the primary outcome. This is because the ultimate evaluation of a surrogate's performance is based on the expected final outcome, underscoring the critical role of predictability in evaluating surrogate outcomes in adaptive designs.

In summary, our proposed target estimand provides a comprehensive measure of a candidate surrogate’s utility to accelerate participant benefit in sequential adaptive designs. It captures not only the timeliness with which the surrogate becomes available for early adaptation, but also its sensitivity to heterogeneous treatment effects and its predictive accuracy in identifying optimal personalized treatments for the primary outcome. This approach contrasts with existing surrogate evaluation methods, which are generally developed for static settings rather than from a dynamic adaptive design perspective.

\subsection{Approximation of an Ideal Estimand} 
\label{appendix:ideal_estimand}
\begingroup
\color{vz}
Ideally, one could fully evaluate a candidate adaptive design by defining an ideal estimand as the expected final outcome when the entire experiment is conducted under that design, averaged across all participants. Formally, this requires integrating over the joint distribution of entry-time histories, contexts, treatments, and outcomes that would be sequentially induced if the candidate adaptive design were implemented from the start of the experiment. In practice, however, targeting this estimand is computationally intensive and entails a substantially more involved estimation procedure.

Instead, our context-specific estimand provides a fast and robust approximation to the ideal estimand by evaluating the expected outcome under the candidate design using contexts of the observed histories under the actual experiment. This estimand avoids modeling the full induced distribution under the candidate design and admits doubly robust estimation via TMLE (see Section~\ref{section:tmle_analysis}). The approximation is most accurate when the observed contexts from the actual experiment are representative of those the candidate design would induce - especially when the candidate design’s adaptation is based on learning from a stationary target such as $Q_0$ or CATE that dominates over the experiments. This provides an effective approach to distinguishing low-utility from high-utility adaptive designs, thereby facilitating online selection among candidate designs with robust statistical support.

In simulation studies, we also compute the expected counterfactual final outcome obtained by implementing each candidate adaptive design (see Table~\ref{tab:truePsi_marginal} in Appendix~\ref{appendix:additional_results}), which corresponds to the ideal estimand. The results show that our context-specific estimand (Table~\ref{tab:truePsi}) closely matches this ideal target with minimal differences.
\par
\endgroup

\def\Dzero{D^{(0)}([0,1]^{d})}
\def\Dmzero{D^{(m)}([0,1]^{d})}
\def\DMnzero{D^{(0)}_{M_n}([0,1]^{d})}
\def\Done{D^{(1)}([0,1]^{d})}

\newcommand{\Dwhatzero}[1]{D^{(#1)}([0,1]^{d})}

\def\Dkzero{D^{(m)}([0,1]^{d})}
\def\DCzero{D_C^{(0)}([0,1]^{d})}
\def\DkCzero{D_M^{(m)}([0,1]^{d})}

\def\DmCzero{D_C^{(m)}}
\def\DmMRn{D^{(m)}_{M}(\cR_n)}
\def\DmMnRn{D^{(m)}_{M_n}(\cR_n)}
\def\DmCRn{D^{(m)}_{C}(\cR_n)}

\def\onetod{\{1, \ldots, d\}}

\def\Qv{\|Q\|_v^*}
\def\bQv{\|\bar{Q}\|_v^*}

\def\Szerod{\mathcal{S}^0_{[d]}}
\def\Szerostard{\mathcal{S}^{0,*}_{[d]}}
\def\Rzeros{\mathcal{R}^0(s)}
\def\Rzerostars{\mathcal{R}^{0,*}(s)}
\def\Rzerod{\mathcal{R}^0_{[d]}}
\def\Rzerostard{\mathcal{R}^{0,*}_{[d]}}

\def\RNzerostars{\mathcal{R}_n^{0,*}(s)}
\def\RNzerod{\mathcal{R}^{0,n}_{[d]}}

\def\Rkd{\mathcal{R}^m_{[d]}}
\def\Rkstars{\mathcal{R}^{m,*}(\bar{s}(m+1))}
\def\Rkstard{\mathcal{R}^{m,*}_{[d]}}
\def\RNkd{\mathcal{R}^{m,n}_{[d]}}
\def\RNkstars{\mathcal{R}_n^{m,*}(\bar{s}(m+1))}

\newcommand\RNwhatstard[1]{\mathcal{R}^{#1,n}_{[d]}}

\def\RzerosJ{\mathcal{R}^0(s, \mathbf{J}(s))}
\def\RzerodJ{\mathcal{R}^{0}[d,\mathbf{J}]}
\def\RzerostarsJ{\mathcal{R}^{0,*}(s, \mathbf{J}(s))}
\def\RzerostardJ{\mathcal{R}^{0,*}[d,\mathbf{J}]}

\def\RksJ{\mathcal{R}^m(s, \mathbf{J}(s))}
\def\RkdJ{\mathcal{R}^{m}[d,\mathbf{J}]}
\def\RkstarsJ{\mathcal{R}^{m,*}(s, \mathbf{J}(s))}
\def\RkstardJ{\mathcal{R}^{m,*}[d,\mathbf{J}]}

\def\Qsbark{Q_{\bar{s}(m+1)}^{(m)}}
\def\sbar{\bar{s}}

\def\Dhatzero{D^{(0)}}

\def\ibeta{\beta^{\textit{init}}}
\def\tbeta{\beta^{\textit{target}}}
\def\halbeta{\beta^{\textit{hal}}}
\def\rbeta{\beta^{\textit{relax}}}
\def\betaLone{\|\beta\|_1}

\def\tpsi{\psi^{\textit{target}}}

\def\tV{V^{\textit{target}}}
\def\tphi{\tilde{\phi}}

\def\dto{\xrightarrow{d}}
\def\pto{\xrightarrow{p}}

\def\EICbetank{D_{P_n, \phi_k}}
\def\EICxbetank{D_{P_n, \phi_k, \tx}}

\def\textlu{\textit{local-u}}
\def\textgu{\textit{global-u}}

\def\kcv{k_{cv}}
\def\tx{\tilde{x}}
\def\sigmahat{\hat{\sigma}}
\def\tsigmahat{\hat{\sigma}^{\textit{target}}}
\def\halsigmahat{\hat{\sigma}^{\textit{hal}}}
\def\relaxsigmahat{\hat{\sigma}^{\textit{relax}}}
\def\isigmahat{\hat{\sigma}^{\text{init}}}

\def\Mcv{M_{cv}}
\def\Mgu{M_{global-u}}
\def\Mlu{M_{local-u}}
\def\QnM{Q_{n, M}}
\def\Dk{D^{(k)}}
\def\betanM{\beta_{n,M}}
\def\betanMcv{\beta_{n,M_{cv}}}
\def\betanMgu{\beta_{n,{M_{global-u}}}}
\def\betanMlu{\beta_{n,{M_{local-u}}}}

\def\phinM{\phi_{n,M}}
\def\phinMcv{\phi_{n,M_{cv}}}
\def\betanphinM{\beta_{n,\phinM}}
\def\betanphi{\beta_{n}(\phi)}
\def\Dbetaphi{D_{\beta}(\phi)}
\def\Dbetanphi{D_{\beta_n}(\phi)}
\def\Dbetaphixtilde{D_{\beta,\tilde{x}}(\phi)}
\def\Dbetanphixtilde{D_{\beta_n,\tilde{x}}(\phi)}
\def\DbetanMphinMxtilde{D_{\beta_{n,M},\tilde{x}}(\phi_{n,M})}
\def\DbetaphiMphixtilde{D_{\beta,\tilde{x}}(\phi)}
\def\DbetaphiMkphixtilde{D_{\beta,\tilde{x}}(\phi_{n, M_k},\phi_{M_{cv}})}

\def\DbetanMphiMphixtilde{D_{\beta_{n,M},\tilde{x}}(\phi_{n, M},\phi_{n,M})}
\def\DbetanMkphiMkphixtilde{D_{\beta_{n,M_k},\tilde{x}}(\phi_{n, M_k},\phi_{n,M_k})}
\def\DbetanMkphiMkphixtildej{D_{\beta_{n,M_k},\tilde{x}_j}(\phi_{n, M_k},\phi_{n,M_k})}
\def\DbetanMkphiMcvphixtilde{D_{\beta_{n,M_k},\tilde{x}}(\phi_{n, M_k},\phi_{n, M_{cv}})}
\def\DbetanMkphiMcvphixtildej{D_{\beta_{n,M_k},\tilde{x}_j}(\phi_{n, M_k},\phi_{n, M_{cv}})}

\def\P{\mathbb{P}}
\def\R{\mathbb{R}}
\def\C{\mathbb{C}}
\def\Z{\mathbb{Z}}
\def\N{\mathbb{N}}
\def\E{\mathbb{E}}
\def\P{\mathbb{P}}
\def\T{\mathsf{T}}
\def\G{\mathcal{G}} 

\def\Cov{\mathrm{Cov}}
\def\Var{\mathrm{Var}}
\def\indep{\perp\!\!\!\perp}
\def\th{^{\text{th}}}
\def\tr{\mathrm{tr}}
\def\df{\mathrm{df}}
\def\dim{\mathrm{dim}}
\def\col{\mathrm{col}}
\def\row{\mathrm{row}}
\def\nul{\mathrm{null}}
\def\rank{\mathrm{rank}}
\def\nuli{\mathrm{nullity}}
\def\spa{\mathrm{span}}
\def\sign{\mathrm{sign}}
\def\supp{\mathrm{supp}}
\def\diag{\mathrm{diag}}
\def\aff{\mathrm{aff}}
\def\conv{\mathrm{conv}}
\def\dom{\mathrm{dom}}
\def\hy{\hat{y}}
\def\hf{\hat{f}}
\def\hmu{\hat{\mu}}
\def\halpha{\hat{\alpha}}
\def\hbeta{\hat{\beta}}
\def\htheta{\hat{\theta}}
\def\cA{\mathcal{A}}
\def\cB{\mathcal{B}}
\def\cD{\mathcal{D}}
\def\cE{\mathcal{E}}
\def\cF{\mathcal{F}}
\def\cG{\mathcal{G}}
\def\cK{\mathcal{K}}
\def\cH{\mathcal{H}}
\def\cI{\mathcal{I}}
\def\cL{\mathcal{L}}
\def\cM{\mathcal{M}}
\def\cN{\mathcal{N}}
\def\cP{\mathcal{P}}
\def\cR{\mathcal{R}}
\def\cS{\mathcal{S}}
\def\cT{\mathcal{T}}
\def\cW{\mathcal{W}}
\def\cX{\mathcal{X}}
\def\cY{\mathcal{Y}}
\def\cZ{\mathcal{Z}}
\def\bQ{\bar{Q}}
\def\bO{\bar{O}}

\def\bgkt{\mathbf{g}_{k,1:t}}
\def\g0t{\mathbf{g_0^t}}
\def\g1{\mathbf{g_1}}
\def\gK{\mathbf{g_K}}
\def\tW{\tilde{W}}
\def\tA{\tilde{A}}
\def\tY{\tilde{Y}}
\def\tO{\tilde{O}}
\def\bO{\bar{O}}
\def\bo{\bar{o}}
\def\bbO{\overline{\overline{O}}}
\def\bbo{\overline{\overline{o}}}

\def\asto{\overset{\mathrm{as}}{\to}}
\def\pto{\overset{p}{\to}}
\def\dto{\overset{d}{\to}}

\section{The Highly Adaptive Lasso Estimator}
\label{appendix:HAL}
Highly Adaptive Lasso (HAL) \citep{van2015generally,benkeser2016highly} is a flexible non-parametric regression estimator that minimizes empirical risk over the class of càdlàg (right-continuous with left limits) functions with bounded sectional variation norms \citep{gill1995inefficient}. Instead of relying on local smoothness assumptions, HAL imposes a global smoothness constraint through the sectional variation norm and achieves a convergence rate of order $n^{-1/3}$ up to logarithmic factors when using zero-order splines \citep{bibaut2019fast}. Recent work by \cite{van2023higher} studies higher-order smoothness classes of càdlàg functions and develops higher-order spline HAL, showing that the $m$-th order HAL is asymptotically normal and attains a uniform convergence rate of order $n^{-\frac{m+1}{2m+3}}$ up to a logarithmic factor ($m = 1,2,\ldots$).
In this work, we use first-order HAL to estimate conditional average treatment effect (CATE) functions (Appendix \ref{appendix:CATE}) and provide a corollary for proving equicontinuity conditions of martingale processes based on theoretical results of higher-order spline HAL-MLEs (Appendices \ref{appendix:equicontinuity_general} and \ref{appendix:equicontinuity}).

We first introduce a class of $d$-variate càdlàg functions defined on a unit cube $[0,1]^d$ with bounded sectional variation norm \citep{gill1995inefficient}.

For any $x \in [0,1]^d$, denote $x(j)$ as its $j$-th coordinate.
For a given subset $s \subset \onetod$, denote $x_s = (x(j): j \in s)$.
Let $\bQ$ be a real-valued $d$-variate càdlàg function on a unit cube $[0,1]^d$.
Define $\bQ_s: (0_s,1_s] \to \R$ as $\bQ_s(u_s)=\bQ(u_s,0_{-s})$, i.e., the $s$-specific section of $\bQ$ that sets coordinates in the complement of subset $s$ equal to 0.
The sectional variation norm of $\bQ$ is defined by 
$\bQv := |\bQ(0)|+\sum_{s \subset \onetod, s \ne \emptyset} \int_{(0_s, 1_s]}|d\bQ(u_s, 0_{-s})|$.
Suppose $\bQv$ is finite, then $\bQ$ can be represented as \citep{gill1995inefficient}:
\[
\bQ(x) = \int_{[0, x]} d \bQ(u) = \bQ(0)+\sum_{s \subset\onetod, s \ne \emptyset} \int_{(0_s, 1_s]} \phi_{u_s}^{(0)}(x_s) d \bQ(u_s, 0_{-s}),
\]
where $\phi_{u_s}^{(0)}(x_s) := I(x_s \geq u_s)$ is a zero-order basis function. 

One can use discrete measure with observed support points to approximate the presentation of $\bQ$ \citep{van2015generally, benkeser2016highly}. 
Given $n$ i.i.d samples $(x_i,y_i)$ ($i = 1,\ldots,n$) with $x_i \in [0,1]^d$, $y_i \in \R$, one can use the $s$-specific support points $x_{s,i} := (x_{i}(j): j \in s)$ observed for each unit $i$ across different subsets $s$ to approximate $\bQ$. Specifically, the approximation is given by 
$
\bQ_{\beta}(x) = \beta_0 +  \sum_{s \subset \{1,\ldots,d\}} \sum_{i=1}^n \beta_{s,i}\phi_{x_{s,i}}^{(0)}(x(s))$. 

We define $Pf:= \int f(o) dP(o)$ for a function $f$ with a data generating distribution $P$ that generates data $O \sim P$. Let $P_n f:= \int f(o) dP_n(o)$ under empirical distribution $P_n$.
Define the class
$
\mathcal{\bQ}^{(0)}_{n,C} = \left\{ \bQ_{\beta}(x):  |\beta_0| +  \sum_{s \subset \{1,\ldots,d\}} \sum_{i=1}^n |\beta_{s,i}| < C \right\}$
for some constant $C$, and let 
$L(\bQ)$ be a loss function.
A zero-order HAL-MLE is constructed by solving the following loss-based minimization problem:
$
\bQ_n = \argmin_{\bQ_\beta \in \mathcal{\bQ}^{(0)}_{n,C_n^u}} P_n L(\bQ_\beta)$,
where $C_n^u$ is a data-adaptive cross-validation or undersmoothed selector of bounded sectional variation norm for the working model constructed on the basis functions $\phi_{x_{s,i}}^{(0)}$'s.
\cite{bibaut2019fast} establishes the $L_2$ convergence rate of $n^{-\frac{1}{3}}$ up to logarithmic factors for the zero-order HAL-MLE.

\cite{van2023higher} extends the HAL construction to higher-order smoothness classes. For integer $m \ge 1$, define the $m$-th order smoothness class $\Dkzero$ $(m = 1, 2, \ldots)$ on $[0,1]^d$ as the class with functions whose $m$-th-order Lebesgue-Radon-Nikodym derivatives are càdlàg functions with bounded sectional variation norm.
Like the zero-order HAL-MLE, one can construct a higher-order HAL-MLE that uses $m$-th order splines to estimate the target function in $\Dkzero$. For the first-order ($m=1$), the approximation of $\bQ^{(1)} \in \Done$ can be written by 
$
    \bQ^{(1)}_{\beta}(x) = \beta_0 +  \sum_{s \subset \{1,\ldots,d\}} \sum_{i=1}^n \beta_{s,i}\phi_{x_{s,i}}^{(1)}(x(s)), 
$
where $\phi^{(1)}_u(x) = (x - u)I(x \ge u)$.
Define 
$
\cbQ_{n,C}^{(1)}=\{ \bQ^{(1)}_{\beta} \in \Done: |\beta_0|+ \sum_{s \subset \{1,\ldots,d\}} \sum_{i=1}^n |\beta_{s,i}| < C\}$.
The first-order HAL-MLE is constructed by 
$
\bQ_n^{(1)} = \argmin_{\bQ_\beta^{(1)} \in \cbQ_{n,C_n^u}^{(1)}} P_n L(\bQ_\beta^{(1)}).$
One can also consider the second-order spline function $\phi^{(2)}_u(x) = \frac{1}{2}(x - u)^2 I(x \ge u)$ or higher-order spline basis functions to fit higher-order HAL-MLEs.
Moreover, \citet{van2023higher} establishes uniform covering number and entropy bounds for these higher-order HAL estimators, and shows that the 
$m$-th order spline HAL-MLE achieves asymptotic normality and a uniform convergence rate of $n^{-\frac{m+1}{2m+3}}$ up to logarithmic factors.

\section{Estimation of Conditional Average Treatment Effect} \label{appendix:CATE}

A key procedure in our adaptive design is estimating the conditional average treatment effect (CATE). By estimating the CATE function, the experimenter can identify the optimal dynamic treatment rule that maximizes the expected outcome conditional on baseline covariates, and update treatment randomization probabilities accordingly to increase the chance of allocating the estimated optimal personalized treatment in adaptive design.

In this work, we use the Highly Adaptive Lasso (HAL) estimator to estimate the CATE function, with its advantages introduced in Appendix \ref{appendix:HAL}.
Furthermore, we apply the delta-method to obtain a Wald-type confidence interval for the CATE. This confidence interval guides the exploration-exploitation trade-off in adaptive designs by adjusting treatment randomization probabilities based on the uncertainty in CATE estimates, which will be discussed in Appendix \ref{appendix:details_of_CARA}.

Recall that $\bQ_{0,k}(A, W) = E[Y_{k}|A, W]$ is the true conditional expectation of the $k$-th surrogate outcome, given the treatment $A$ and baseline covariates $W$. 
For every new participant $i$ enrolled at time $t$, we use $B_{n,k,t}(W_i)$ to denote an estimate of $B_{0,k}(W_i):=E[Y_k(1)-Y_k(0)|W_i]$, the true CATE of $Y_k$ given $W_i$, the baseline covariates of participant $i$.
We use $\tau_{n,k,t}(W_i)$ to denote the estimated standard error of CATE. 
Both estimates are based on $\cH(i)$, the historical data observed before participant $i$ is enrolled.

For each surrogate outcome $Y_k$, we estimate CATE functions by two steps. First, we transform $Y_k$ into the pseudo-outcome \citep{van2006statistical,luedtke2016super}:
$$
\eta_{k}(\bQ_k, g)(O)=\frac{2 A-1}{g(A \mid W)}[Y_{k}-\bQ_k(A, W)]+\bQ_k(1, W)-\bQ_k(0, W).
$$
We note that the pseudo-outcome has the doubly-robust property that $E[\eta_k(\bQ_k, g)(O) \mid W]=B_{0,k}(W)$ if $\bQ_k=\bQ_{0, k}$ or $g=g_0$. 
Since $g_0$ is known in adaptive design, we only need to estimate $\bQ_{0, k}$ with $\bQ_{n, k}$, then impute $\bQ_{n, k, t}$ and $g_0$ to generate the pseudo-outcome $\eta_k(\bQ_{n,k,t}, g_0)$. 

Second, we estimate the CATE function $B_{0,k}$ by regressing the pseudo-outcome $\eta_{k}\left(\bQ_{n, k, t}, g_0\right)$ on baseline covariates $W$ by first-order spline HAL-MLE, which minimizes the empirical risk over $\Dwhatzero{1}$, the first-order smoothness class with its first-order derivative being càdlàg function with bounded sectional variation norm.
That is, at time $t$, we estimate the CATE function by solving the following minimization problem:
\begin{align*}
B_{n,k,t} = \argmin_{B_k \in \cbQ_{C_{n,k,t}^u}^{(1)}} P_{n_{t,k}} L_k(B_k)
\end{align*}
where $L_k(B_k)(O) := (B_k(W)-\eta_k(O))^2$, and $P_{n_{t,k}}$ denotes the empirical distribution over the $n_{t,k} := N(t-k)$ observations with observed $Y_k$ at time $t$; $C_{n,k,t}^u$ is an upper bound for the first-order sectional variation norm that is data-adaptively selected by cross-validation. 

As introduced in Appendix \ref{appendix:HAL},  $B_{n,k,t}(\cdot)$ can be written as a linear combination of first-order spline basis functions with non-zero coefficients.
To quantify the uncertainty of the CATE, we use the delta method under this HAL working model to obtain pointwise inference for the conditional mean of the pseudo-outcome, which corresponds to the CATE. Let \(\tau_{n,k,t}(W_i)\) denote the estimate of the standard error of \(B_{n,k,t}(W_i)\); a Wald-type $100(1-\alpha)\%$ confidence interval for the CATE at $W_i$ is $B_{n,k,t}(W_i) \pm z_{1-\alpha/2}\tau_{n,k,t}(W_i)$ ($\alpha = 0.05$). 
Both \(B_{n,k,t}(W_i)\) and \(\tau_{n,k,t}(W_i)\) are then incorporated into the treatment randomization function introduced in Appendix \ref{appendix:details_of_CARA} to construct \(g_{k,i}^*\) for each participant $i$ in the adaptive design guided by surrogate $Y_k$.

\section{Details of CARA Procedures}
\label{appendix:details_of_CARA}

We define a treatment randomization function $\tilde{h}_{\nu}: [-1,1] \times \mathbb{R}^{+} \to [0,1]$ as follows:
\begin{align*}
\tilde{h}_{\nu}(x,b) = \nu I(x \leq -b) + (1-\nu) I(x \geq b) + \left(-\frac{1/2 - \nu}{2b^3}x^3 + \frac{1/2 - \nu}{2b/3}x + \frac{1}{2}\right) I(-b < x < b),
\end{align*}
where $\nu \in (0,0.5)$ is a user-specified constant ensuring that $\nu \leq \tilde{h}_\nu \leq 1-\nu$.

This mapping, adapted from \cite{chambaz2017targeted}, is extended here to explicitly incorporate both the magnitude and uncertainty of the estimated CATE. In our setup, $x$ is the point estimate of CATE for each subject, while $b$ is the estimated standard error of CATE, scaled by the $z_{1-\alpha/2}$ quantile of the standard normal distribution.

For interpretation, it is helpful to re-express this function in terms of the standardized CATE $z = x / b$, which is the CATE divided by half the width of its confidence interval. The function then becomes
\begin{eqnarray}
h_\nu(z) = 
\begin{cases}
    \nu, & z \leq -1 \\
    \nu + (1 - 2\nu)\left( -\dfrac{1}{4}z^3 + \dfrac{3}{4}z + \dfrac{1}{2} \right), & -1 < z < 1 \\ 
    1 - \nu, & z \geq 1
    \nonumber
\end{cases}
\end{eqnarray}

This function smoothly interpolates between balanced randomization and deterministic allocation of optimal personalized treatment based on the standardized CATE signal relative to uncertainty. 
When the signal is modest, a near-balanced allocation is maintained between treatment and control, avoiding over-commitment to a noisy or unstable estimated rule and ensuring continued learning of CATE by preserving exploration in both arms.
As evidence accumulates and the magnitude of the CATE signal increases, the treatment randomization probability shifts further toward the estimated optimal treatment. Meanwhile, a minimal exploration probability is retained, which is necessary for estimating CATE during the experiment and a range of other causal estimands at the end of the adaptive experiment.

In our surrogate-guided adaptive designs, the treatment randomization probabilities are generated as follows. For each candidate surrogate outcome $Y_k$ and each subject $i$ with baseline covariates $W_i$, we estimate CATE $B_{n,k}(W_i)$ and its standard error $\tau_{n,k}(W_i)$ as detailed in  Appendix~\ref{appendix:CATE}.
Consequently, the treatment randomization probability for $a=1$ is given by $h_\nu(\frac{B_{n,k}(W_i)} {z_{1-\alpha/2} \cdot \tau_{n,k}(W_i)})$, and for $a=0$ it is $1 - h_\nu(\frac{B_{n,k}(W_i)}{z_{1-\alpha/2} \cdot \tau_{n,k}(W_i))})$, where $\nu=0.1$. If the surrogate outcome $Y_k$ has not yet been observed, the assignment defaults to equal randomization probability of 0.5.

\section{Equicontinuity Condition for a General Martingale Process}
\label{appendix:equicontinuity_general}

In this section, we provide a theorem about the equicontinuity condition of martingale processes in a general setup that includes but not limited to adaptive design settings.
\subsection{Statistical Setup}
We consider a filtered probability space 
$(\Omega, \mathcal F, \{\mathcal F_i\}_{i \ge 1}, \mathbf P)$, 
where $\{O_i\}_{i \ge 1}$ is a sequence of random variables taking values in $\mathcal O$, and 
$\mathcal F_i := \sigma(O_1, \ldots, O_i)$. 
Let $\bar O(i) := (O_1, \ldots, O_i)$.
We assume that the data are generated sequentially with each $O_i$ depending on $\bar O(i-1)$ through a fixed-dimensional summary $C_i = f_C(\bar O(i-1)) \in \cC$.
Let $\Theta$ be a class of functions, and 
$\{f_\theta : \theta \in \Theta\}$
be measurable functions from 
$\mathcal{O} \times \mathcal{C} \to \mathbb{R}$. 
Define
\[
M_N(\theta) = \frac{1}{N} \sum_{i=1}^N \left\{f_{\theta}(O_i,C_i) - E\left[f_{\theta}(O_i,C_i) \mid \bO(i-1)\right]\right\}.
\]

Given a fixed $\theta^* \in \Theta$ and a $\theta_N \in \Theta$ constructed from data $\bar O(N)$.
We provide a general theorem for the equicontinuity condition with respect to the martingale process as follows:
\[M_{2,N}(\theta_N,\theta^*) := M_N(\theta_N) - M_N(\theta^*) = \smallopsqrtN.\]

\subsection{Proof of Equicontinuity Condition}
We first define a bracketing entropy following Definition A.1 in \cite{van2011minimal} (see also Definition 8.1 in \cite{geer2000empirical}).

\begin{restatable}{definition}{repDefSeqBracketing}{(Bracketing entropy)}
\label{def:seq_bracketing}
Consider a stochastic process of the form $\left\{\left(\xi_{\theta,i}\right)_{i=1}^N: \theta \in \Theta\right\}$ where $\Theta$ is an index set such that, for every $\theta \in \Theta, i \in[N]$, $\xi_{\theta,i}$ is an $\bar{O}(i)$-measurable random variable for all $i$ and $\theta$. 
We say that a collection of random variables of the form $\mathcal{B}:=\left\{\left(\Lambda_i^j, \Upsilon_i^j\right)_{i=1}^N: j \in[J]\right\}$ is an $(N,\Theta, Z, \epsilon)$-bracketing if

\begin{enumerate}
    \item for every $i \in[N]$, and $j \in[J],\left(\Lambda_i^j, \Upsilon_i^j\right)$ is $\bar{O}(i)$-measurable,
    \item for every $\theta \in \Theta$, there exists $j \in[J]$, such that, for every $i \in[N], \Lambda_i^j \leq \xi_{\theta,i} \leq \Upsilon_i^j$,
    \item for every $j \in[J]$,
    $$
    \frac{2Z^2}{N} \sum_{i=1}^N E\left[\phi\left(\frac{\Upsilon_i^j-\Lambda_i^j}{Z}\right) \mid \bO(i-1)\right] \leq \epsilon^2,
    $$
\end{enumerate}
where $\phi(x) = e^x - x - 1$.
We denote $\bracketingnumber{\epsilon}{Z}$ the minimal cardinality of an $\left(N, \Theta, Z, \epsilon \right)$-bracketing.
We define the bracketing entropy integral as
$\bracketingintegral{\epsilon}{Z} := \int_{0}^{\epsilon} \bracketingnumber{u}{Z} du$.
\end{restatable}

\begin{theorem}[Theorem A.4 in \cite{van2011minimal}]
\label{theorem:handel_A4}
Given $Z > 0$, for all $i \ge 0$, define 
\[
M_{\theta, i} = \frac{1}{N} \sum_{l=1}^{i}\left(\xi_{\theta,l} - E\left[\xi_{\theta,l} \mid \cF_{l-1} \right]\right),\]
\[
R_{\theta, i} = \frac{1}{N} 2Z^2 \sum_{l=1}^{i} E\left[\phi\left(\frac{|\xi_{\theta,l}|}{Z}\right) \bigg| \cF_{l-1} \right].
\]
Then,
for any $N \in \mathbb{N}$, $x > 0$, and $0 < R <\infty$, 
\begin{align*}
& P\left( \sup_{\theta \in \Theta} I\left(R_{\theta,N} \le R \right) \max_{i \in [N]}\sqrt{N} M_{\theta,i} \ge 16 \Gamma + 32\sqrt{Rx} + \frac{16}{\sqrt{N}} Zx \right) \le 2 e^{-x},
\end{align*}
where $\Gamma := \frac{Z}{\sqrt{N}} \log \bracketingnumber{\sqrt{R}}{Z} + 4 \bracketingintegral{\sqrt{R}}{Z}$.
\end{theorem}

We focus on the process  
\(\xi_{\theta,i} = f_{\theta}(O_i, C_i) - f_{\theta^*}(O_i, C_i)\) for \(i = 1, \ldots, N\), and \(\theta, \theta^* \in \Theta\).
Since the bracketing entropy condition may be difficult to verify, we relate it to the supremum norm covering entropy, following \citet{geer2000empirical} (Section 8.2) and \citet{chambaz2011targeting} (Lemma 7). In our setting, this is done by working with the covering entropy of \(\Theta\).

\begin{assumption}[Entropy Integral]
\label{assumption:entropy_integral}
    Let $\supcoveringTwo{\epsilon}{\Theta}$ denote the $\epsilon$-covering number in supremum norm of $\Theta$ and let $\supintegralTwo{\epsilon}{\Theta}:= \int_0^\epsilon \sqrt{\log \left(1+\supcoveringTwo{u}{\Theta} \right)} d u$ denote the $\epsilon$-entropy integral in supremum norm of $\Theta$. Assume that $\supintegralTwo{\epsilon}{\Theta}$ is finite and $\supintegralTwo{\epsilon}{\Theta} \to 0$ as $\epsilon \to 0$.
\end{assumption}

\begin{assumption}[Uniform Boundedness]
\label{assumption:bounded_function} 
$\sup_{\theta \in \Theta}\supnorm{\theta}< \infty$, $\sup_{\theta \in \Theta}\supnorm{f_{\theta}} < \infty$.
\end{assumption}

\begin{assumption}
\label{assumption:Lipschitz}
$\sup_{\theta_1 \in \Theta, \theta_2 \in \Theta, \theta_1 \ne \theta_2}\frac{\supnorm{f_{\theta_1} - f_{\theta_2}}}{\supnorm{\theta_1 - \theta_2}} < \infty $.
\end{assumption}

\begin{lemma}
\label{lemma: bracketing_dominated_by_supcovering_general}
 Suppose that Assumptions \ref{assumption:entropy_integral},
 \ref{assumption:bounded_function} and \ref{assumption:Lipschitz} hold, such that $\supnorm{\theta} < M$ for some finite constant $M > 0$ and $\sup_{\theta_1 \in \Theta, \theta_2 \in \Theta, \theta_1 \ne \theta_2}\frac{\supnorm{f_{\theta_1} - f_{\theta_2}}}{\supnorm{\theta_1 - \theta_2}} \le L $ for some finite constant $L > 0$. Then the following holds:
\[
\bracketingnumber{2\sqrt{2}L\epsilon}{8LM} \le 
 \cN_{\infty}(\epsilon, \Theta), 
\] 
and
\[
\bracketingintegral{2\sqrt{2}L\epsilon}{8LM} \le
\supintegralTwo{\epsilon}{\Theta}.
\]
\end{lemma}
\begin{proof}
Suppose $\{\theta_j:j \in [J]\} \subset \Theta$ is an $\epsilon$-covering of $\Theta$ in supremum norm such that $\forall \theta \in \Theta$, there exists $j \in [J]$ such that $\theta_j \in \Theta$ and $ \|\theta_j - \theta \|_{\infty} \le \epsilon$.
Denoting $\lambda_j = \theta_j - \epsilon$ and $v_j = \theta_j + \epsilon$. 
With Assumption \ref{assumption:bounded_function}, we can let $\epsilon \le 2M$ as $\lambda_j = \theta_j - M$ and $v_j = \theta_j + M$ always construct a valid pair of brackets.
Under Assumption \ref{assumption:Lipschitz}, $\supnorm{f_{\theta_j}-f_{\theta}} \le L \supnorm{\theta_j - \theta} \le L \epsilon$. 
Define $\Lambda_i^j = f_{\theta_j} - L\epsilon$ and $\Upsilon_i^j = f_{\theta_j} + L\epsilon$.
Then, we have that
    \[ \frac{1}{N}2Z^2 \sum_{i=1}^N E\left[\phi\left(\frac{|\Upsilon_i^j-\Lambda_i^j|}{Z}\right) \mid \bar{O}(i-1)\right]
    \le 2Z^2 \phi\left(\frac{2L\epsilon}{Z}\right)
    \le 2Z^2 \sum_{p \ge 2} \frac{1}{p!} \frac{p!}{2}\left(\frac{ 2L\epsilon}{Z}\right)^p
    =\frac{2L\epsilon}{\frac{1}{2L\epsilon}-\frac{1}{Z}}.\]
Let $Z = 8LM$. We have that
\[
\frac{1}{N}2Z^2 \sum_{i=1}^N E\left[\phi\left(\frac{|\Upsilon_i^j-\Lambda_i^j|}{Z}\right) \mid \bar{O}(i-1)\right] \le (2\sqrt{2}L\epsilon)^2.
\]
Therefore, 
\[
\bracketingnumber{2\sqrt{2}L\epsilon}{8LM} \le 
 \cN_{\infty}(\epsilon, \Theta), 
\] 
and
\[
\bracketingintegral{2\sqrt{2}L\epsilon}{8LM} \le
\supintegralTwo{\epsilon}{\Theta}.
\]
\end{proof}

\begin{assumption}[Convergence of $\theta_N$]
\label{assumption:convergence_of_theta_N}
    For any $\theta, \theta^* \in \Theta$, define 
    \[
    s_N\left(\theta, \theta^*\right):= \sqrt{\frac{1}{N} \sum_{i=1}^N E\left[\phi\left(\frac{f_{\theta}\left(O_i,C_i\right)-f_{\theta^*}\left(O_i,C_i\right)}{Z}\right) \mid \bO(i-1)\right]},
    \] where $Z = 8LM$. 
    Assume that $s_N\left(\theta_N, \theta^*\right) = O_P(\delta_N)$ for some $\delta_N \to 0$.
\end{assumption}

\begin{theorem}[Maximal Inequality of a Martingale Process]
\label{theorem:max_inequality_general}
$\forall x>0, N \in \mathbb{N}$, and $0 < r < \infty$. Denote $a_N(r,x) := \frac{16Z}{\sqrt{N}} \log \supcoveringTwo{\frac{r}{2\sqrt{2}L}}{\Theta} + 64 \supintegralTwo{\frac{r}{2\sqrt{2}L}}{\Theta} + 32r\sqrt{x} + \frac{16}{\sqrt{N}} Zx$.
Suppose that Assumptions \ref{assumption:entropy_integral},
 \ref{assumption:bounded_function}, and \ref{assumption:Lipschitz} 
hold. Then,
\[
P\left(\sup_{\theta \in \Theta} I\left(s_N(\theta, \theta^*) \le r \right) \sqrt{N} M_{2,N}(\theta,\theta^*) \ge a_N(r,x) \right) 
\le 2e^{-x}.
\]
\end{theorem}

\begin{proof}
Recall that  
\[
M_{2,N}(\theta,\theta^*) = \frac{1}{N} \sum_{i=1}^N \left\{f_{\theta}\left(O_i,C_i\right) - f_{\theta^*}\left(O_i,C_i\right) - E\left[f_{\theta}(O_i,C_i)-f_{\theta^*}(O_i,C_i) \mid \bO(i-1)\right]\right\}.
\]

Let $Z = 8LM$ and $b_N(r,x) = \frac{16Z}{\sqrt{N}} \log \bracketingnumber{r}{Z} + 64 \bracketingintegral{r}{Z} + 32r\sqrt{x} + \frac{16}{\sqrt{N}} Zx$. 
We have that
\begin{align*}
& P\left(\sup_{\theta \in \Theta} I\left(s_N\left(\theta, \theta^*\right) \le r \right) \sqrt{N} M_{2,N}(\theta,\theta^*) \ge a_N(r,x)
\right) \\
& \le P\left(\sup_{\theta \in \Theta} I\left(s_N\left(\theta, \theta^*\right) \le r \right) \sqrt{N} M_{2,N}(\theta,\theta^*) \ge b_N(r,x)
\right) \\
& \le 2e^{-x},
\end{align*}
where the first inequality is given by Lemma \ref{lemma: bracketing_dominated_by_supcovering_general} and the second inequality is based on the maximal inequality for the martingale process in Theorem \ref{theorem:handel_A4}.
\end{proof}

\begin{restatable}{corollary}{repCoroHighOrderSplineEquicontinuity_general}
\label{corollary:equicontinuity_general}
Suppose Assumptions \ref{assumption:entropy_integral},
\ref{assumption:bounded_function}, 
\ref{assumption:Lipschitz} and
\ref{assumption:convergence_of_theta_N} hold, then
$M_{2,N}(\theta_N, \theta^*) = \smallopsqrtN$. 
\end{restatable}
\begin{proof}
Since $s_N\left(\theta_N, \theta^*\right)=O_P(\delta_N)$, one can find a $r_N = \delta_N^{1-\rho}$ for some $\rho > 0$ such that $r_N \to 0$ and $\sqrt{N}r_N^2 \to \infty$. Then, for any $x>0$, we have
\begin{align*}
    a_N(r_N, x) &= \frac{16Z}{\sqrt{N}} \log \supcoveringTwo{\frac{r_N}{2\sqrt{2}L}}{\Theta} + 64 \supintegralTwo{\frac{r_N}{2\sqrt{2}L}}{\Theta} + 32r_N\sqrt{x} + \frac{16}{\sqrt{N}} Zx \nonumber \\
    &\le \frac{16Z}{\sqrt{N}} \left(\frac{2\sqrt{2}L}{r_N}\supintegralTwo{\frac{r_N}{2\sqrt{2}L}}{\Theta}\right)^2 + 64 \supintegralTwo{\frac{r_N}{2\sqrt{2}L}}{\Theta} + 32r_N\sqrt{x} + \frac{16}{\sqrt{N}} Zx \nonumber \\
    &\le \left(\frac{108L^2Z}{\sqrt{N} r_N^2}\supintegralTwo{\frac{r_N}{2\sqrt{2}L}}{\Theta} + 64\right) \supintegralTwo{\frac{r_N}{2\sqrt{2}L}}{\Theta} + 32r_N\sqrt{x} + \frac{16}{\sqrt{N}} Zx \label{equation:equi_tail} \\
    & \to 0 \nonumber.
\end{align*}
Therefore, for any $y>0$, there exists a $N_1 > 0$ such that for $N > N_1$, $a_N(r_N,x)<y$.
Hence,
\begin{align*}
&P(I\left(s_N\left(\theta_N, \theta^*\right) \le r_N\right) \sqrt{N}M_{2,N}(\theta_N,\theta^*) > y) \\ 
&\le
P\left(\sup_{\theta \in \Theta} I\left(s_N\left(\theta, \theta^*\right)  \le r_N \right) \sqrt{N} M_{2,N}(\theta,\theta^*)
\ge a_N(r_N,x)
\right) \\
& \le 2e^{-x},
\end{align*}
where the last inequality is implied by 
Theorem \ref{theorem:max_inequality_general}.
Similarly, for any $y>0$ and the chosen $x>0$, there exists a $N_2 > 0$, such that for every $N > N_2$,
$
P(-\sqrt{N}M_{2,N} (\theta_N,\theta^*)I\left(s_N\left(\theta_N, \theta^*\right) \le r_N \right) > y) \le 2e^{-x}$. 
Since the bound can be made arbitrarily close to zero by choosing $x$ sufficiently large, and $P(s_N\left(\theta_N, \theta^*\right) \le r_N) \to 1$,
we have that 
$
\sqrt{N}M_{2,N}(\theta_N,\theta^*) = o_P(1).
$
\end{proof}

As introduced in Appendix \ref{appendix:HAL},
\cite{van2023higher} proved the supremum-norm covering number for the class of functions $\Dkzero$ for $(m = 1, 2, \ldots)$, which implies a bound on its uniform entropy integral 
$J_{\infty}\left(r, \Dkzero \right) \asymp r^{(2m+1) /(2m+2)}(-\log r)^{d-\frac{3}2}$. 
The $m$-th order HAL-MLEs are also defined, with the uniform convergence rate of $N^{-\frac{m+1}{2m+3}}$, up to logarithmic factors.
Using these theoretical results, we provide another corollary which shows that if $\theta^*$ lies in $\Dkzero$ ($m = 1,2,3,\ldots$) and $\theta_N$ is an estimate of $\theta^*$ by higher-order spline HAL-MLE, then the equicontinuity result for the martingale process $M_{2,N}$ holds.

\begin{restatable}{assumption}{repAssmpHigherOrderSpline_general}
\label{assumption:HigherOrderSpline_general}
    Suppose that $\theta^* \in \Dkzero$ and $\theta_N \in  \Dkzero$ is estimated by the $m$-order spline HAL-MLE $(m=1,2,\ldots)$.
\end{restatable}

\begin{restatable}{corollary}{repCoroHighOrderSplineEquicontinuity_general}
\label{corollary:HighOrderSplineEquicontinuity_general}
Suppose Assumptions 
\ref{assumption:bounded_function}, 
\ref{assumption:Lipschitz} and
\ref{assumption:HigherOrderSpline_general} hold, then
$M_{2,N}(\theta_N, \theta^*) = \smallopsqrtN$. 
\end{restatable}

\begin{proof}
Let $N_{\infty}\left(r, \Dkzero \right)$ denote the number of spheres of size $r$ w.r.t. supremum norm that are needed to cover $\Dkzero$.
Define $J_{\infty}\left(r, \Dkzero \right)$ as the corresponding supremum norm entropy integral of $\Dkzero$.
In \cite{van2023higher}, Lemma 29 shows that
$
J_{\infty}\left(r, \Dkzero \right) \asymp r^{(2m+1) /(2m+2)}(-\log r)^{d-\frac{3}2}$, which implies $J_{\infty}\left(r_N, \Dkzero \right) \to 0$ for  any $r_N \to 0$ and satisfies Assumption \ref{assumption:entropy_integral};
Theorem 12 proves the uniform convergence rate $\supnorm{\theta_N - \theta^*} = O_P(N^{-\frac{m+1}{2m+3}}\log^{c_1} N)$ for some $c_1 < \infty$, which satisfies Assumption \ref{assumption:convergence_of_theta_N}.
Therefore, together with \ref{assumption:bounded_function} and
\ref{assumption:Lipschitz},
one can apply Corollary \ref{corollary:equicontinuity_general} to prove that $M_{2,N}(\theta_N, \theta^*) = \smallopsqrtN$. 
\end{proof}

\section{Proof of Equicontinuity Result in Adaptive Designs} 
\label{appendix:equicontinuity}
Recall that in Section \ref{section:tmle_analysis}
\begin{align*}
    M_{2,n_t}\left(\bQ_{n,K,t}^*, \bQ_{1,K}\right) &= \frac{1}{n_t}\sum_{i=1}^{n_t}\Biggl\{ \biggl[
    D_k\left(\bQ_{n,K,t}^*\right)\left(O_i,C_i\right)-E\left[ D_k\left(\bQ_{n,K,t}^*\right)\left(O_i,C_i\right)|\cH(i) \right] \biggr]\\
    &- \biggl[ D_k\left(\bQ_{1,K}\right)\left(O_i, C_i\right)-E\left[ D_k\left(\bQ_{1,K}\right)\left(O_i, C_i\right)|\cH(i) \right] \biggr] \Biggr\},
\end{align*}
where
$
D_k(\bQ_K)\left(O_i,C_i\right) := \frac{g^*_{k,i}\left(A_i|C_i\right)}{g_{0,i}\left(A_i|C_i\right)}\left( Y_{i, K}-\bQ_K\left(A_i, W_i\right) \right)$.
We apply the general results in Appendix \ref{appendix:equicontinuity_general} to prove the equicontinuity condition
$M_{2,n_t}(\bQ_{n,K,t}^*, \bQ_{1,K}) = o_P(n_t^{-\frac{1}{2}})$ in Theorem \ref{theorem:equicontinuity} and Corollary \ref{corollary:HighOrderSplineEquicontinuity} in the context of adaptive designs.

\begin{proof}[Proof of Theorem \ref{theorem:equicontinuity} and Corollary \ref{corollary:HighOrderSplineEquicontinuity}]
        Let $\Theta = \cbQ_K$, $f_{\theta} = D_k(\bQ_K)$, $\theta^* = \bQ_{1,K}$, $N=n_t$, $\theta_N = \bQ_{n,K,t}$ in the general setup in Appendix \ref{appendix:equicontinuity_general}. 
    Under Assumption \ref{assumption:strong_positivity}, 
    $\zeta < g_{0,i} < 1-\zeta$ and $\zeta < g_{k,i}^{*} < 1-\zeta$.
    Since $Y_K$ and $Q_K$ take values in the interval $[0,1]$, we have that $\sup_{\bQ_K \in \cbQ_K} \supnorm{D_k(\bQ_K)} \le  \frac{2(1-\zeta)}{\zeta} < \infty$, $\sup_{\bQ_{K,1} \in \cbQ_K, \bQ_{K,2} \in \cbQ_K, \bQ_{K,1} \ne \bQ_{K,2}}\frac{\supnorm{D_k(\bQ_{K,1}) - D_k(\bQ_{K,2})}}{\supnorm{\bQ_{K,1} - \bQ_{K,2}}} \le \frac{1-\zeta}{\zeta} < \infty $. 
    Therefore, Assumptions \ref{assumption:bounded_function} and \ref{assumption:Lipschitz} are satisfied.
    Since Assumption \ref{assumption:sigma_N_convergence} corresponds to Assumption \ref{assumption:convergence_of_theta_N} in the general framework, we can utilize Corollary \ref{corollary:equicontinuity_general} to demonstrate that \( M_{2,n_t}\left(\bQ_{n,K,t}^*, \bQ_{1,K}\right) = o_P\left(n_t^{-\frac{1}{2}}\right) \).
    Similarly, Assumption \ref{assumption:HigherOrderSpline} parallels Assumption \ref{assumption:HigherOrderSpline_general} in Corollary \ref{corollary:HighOrderSplineEquicontinuity_general}. Therefore, Corollary \ref{corollary:HighOrderSplineEquicontinuity} is a specification of Corollary \ref{corollary:HighOrderSplineEquicontinuity_general} within our adaptive design setup, which implies $s_{n_t}(\bQ_{n,K,t},\bQ_{1,K})=O_P(\delta_{n_t})$ where $\delta_{n_t}$ is the HAL convergence rate and preserved by $s_{n_t}(\bQ_{n,K,t}^*,\bQ_{1,K})$ with the TMLE update $\bQ_{n,K,t}^*$. This results in the equicontinuity condition $M_{2,n_t}\left(\bQ_{n,K,t}^*, \bQ_{1,K}\right) = o_P\left(n_t^{-\frac{1}{2}}\right)$.
\end{proof}

\section{Additional Results for Section \ref{section:simulations}}
\label{appendix:additional_results}

The data generating distributions of the two simulation scenarios in Section \ref{section:simulations} are as follows.
The baseline covariate $W$ follows a uniform distribution $U(-4,4)$. 
In Scenario 1, the conditional mean functions of $Y_k$ given $A$ and $W$ are
$E[Y_k|A,W] = (2A-1) \times \left[0.5 - (1 + \text{exp}(-(3-k) - W))^{-1}\right]$.
In Scenario 2, the conditional mean functions of $Y_k$ are
$E[Y_k|A,W] = (2A-1) \times \left[0.5 - (1 + \text{exp}(-\gamma_k W))^{-1}\right]$, 
where $\gamma_1 = 3$,$\gamma_2 = 2$,$\gamma_3 = 1$,$\gamma_4 = 0.5$ and $\gamma_5 = 0.25$.
The treatment randomization probability of binary treatment $A \in \{0,1\}$ depends on the configuration of each design.

In simulation studies, we estimate these conditional mean functions using Super Learner \citep{van2007super}, which includes ordinary least squares, Random Forest \citep{breiman2001random}, the first-order Highly Adaptive Lasso (HAL) estimator \citep{van2023higher} and an intercept model in the library. The CATE functions and their pointwise confidence intervals are estimated using the first-order HAL, as described in Appendix \ref{appendix:CATE}.

\begin{table}[H]\centering
    \caption{Bias ($\times 10^{-3}$) of TMLE estimators for $\Psi_{t,k}$ in Scenarios 1 and 2 ($t = 11,21,31,41,50$).}
    \label{tab:bias}
    \vspace{0.5em}
    \begin{minipage}[t]{0.48\linewidth}
        \makebox[\textwidth]{\small (a) Scenario 1} 
        \resizebox{0.96\textwidth}{!}{
        \large
        \begin{tabular}{*{10}{c}}
            \toprule
           $t$ & $\Psi_{t,RCT}$ & $\Psi_{t,1}$ &  $\Psi_{t,2}$ & $\Psi_{t,3}$ & $\Psi_{t,4}$ & $\Psi_{t,5}$ \\
            \midrule
            11 & 3.90 & 3.22 & 4.20 & 5.02 & 5.09 & 4.70 \\
            21 & -0.83 & -0.55 & -0.32 & -0.78 & 1.06 & 0.98 \\
            31 & 0.03 & -0.07 & -0.16 & 1.09 & 0.18 & 1.13 \\
            41 & 0.61 & -0.02 & 1.27 & 0.25 & 0.72 & -0.51 \\
            50 & -0.51 & -0.48 & -0.91 & 0.11 & -0.94 & -0.01 \\
            \bottomrule
        \end{tabular}
    }
\end{minipage}\hfill
    \vspace{0.5em}
    \begin{minipage}[t]{0.48\linewidth}
        \makebox[\textwidth]{\small (b) Scenario 2} 
        \resizebox{0.96\textwidth}{!}{
        \large
        \begin{tabular}{*{10}{c}}
            \toprule
           $t$ & $\Psi_{t,RCT}$ & $\Psi_{t,1}$ &  $\Psi_{t,2}$ & $\Psi_{t,3}$ & $\Psi_{t,4}$ & $\Psi_{t,5}$ \\
            \midrule
            11 & 1.96 & 2.94 & 1.42 & 2.57 & 2.91 & 2.96 \\
            21 & -2.73 & 0.23 & -0.35 & -0.59 & 0.43 & 0.97 \\
            31 & -1.34 & -0.61 & -1.17 & -1.24 & -0.92 & -0.72 \\
            41 & -1.18 & 0.14 & -0.17 & -0.36 & -0.22 & -0.58 \\
            50 & -1.84 & -0.23 & -0.42 & -0.74 & -0.42 & -0.86 \\
            \bottomrule
            \end{tabular}
        }
    \end{minipage}\hfill
\end{table}

\begin{table}[htb]\centering
    \caption{Variance ($\times 10^{-3}$) of TMLE estimators for $\Psi_{t,k}$ in Scenarios 1 and 2 ($t = 11,21,31,41,50$).}
    \label{tab:variance}
    \vspace{0.5em}
    \begin{minipage}[t]{0.48\linewidth}
        \makebox[\textwidth]{\small (a) Scenario 1} 
        \resizebox{0.96\textwidth}{!}{
        \large
        \begin{tabular}{*{10}{c}}
            \toprule
           $t$ & $\Psi_{t,RCT}$ & $\Psi_{t,1}$ &  $\Psi_{t,2}$ & $\Psi_{t,3}$ & $\Psi_{t,4}$ & $\Psi_{t,5}$ \\
            \midrule
            11 & 3.16 & 7.34 & 6.59 & 5.35 & 4.07 & 3.51 \\
            21 & 2.11 & 4.10 & 3.25 & 2.55 & 1.91 & 2.11 \\
            31 & 1.50 & 3.26 & 2.29 & 1.65 & 1.16 & 1.17 \\
            41 & 1.14 & 2.57 & 1.71 & 1.26 & 0.82 & 0.79 \\
            50 & 1.04 & 2.34 & 1.61 & 1.15 & 0.72 & 0.65 \\
            \bottomrule
        \end{tabular}
    }
\end{minipage}\hfill
    \vspace{0.5em}
    \begin{minipage}[t]{0.48\linewidth}
        \makebox[\textwidth]{\small (b) Scenario 2} 
        \resizebox{0.96\textwidth}{!}{
        \large
        \begin{tabular}{*{10}{c}}
            \toprule
           $t$ & $\Psi_{t,RCT}$ & $\Psi_{t,1}$ &  $\Psi_{t,2}$ & $\Psi_{t,3}$ & $\Psi_{t,4}$ & $\Psi_{t,5}$ \\
            \midrule
            11 & 3.40 & 4.74 & 4.64 & 4.24 & 3.81 & 3.76 \\
            21 & 2.46 & 1.63 & 1.62 & 1.60 & 1.64 & 1.82 \\
            31 & 1.59 & 0.96 & 0.98 & 0.97 & 1.00 & 1.13 \\
            41 & 1.30 & 0.66 & 0.67 & 0.67 & 0.71 & 0.82 \\
            50 & 1.07 & 0.54 & 0.55 & 0.54 & 0.58 & 0.68 \\
            \bottomrule
            \end{tabular}
        }
    \end{minipage}\hfill
\end{table}

\begin{table}[H]\centering
    \caption{Coverage probability (\%) of confidence intervals for $\Psi_{t,k}$ in Scenarios 1 and 2 ($t = 11,21,31,41,50$).}
    \label{tab:coverage}
    \vspace{0.5em}
    \begin{minipage}[t]{0.48\linewidth}
        \makebox[\textwidth]{\small (a) Scenario 1} 
        \resizebox{0.96\textwidth}{!}{
        \large
        \begin{tabular}{*{10}{c}}
            \toprule
           $t$ & $\Psi_{t,RCT}$ & $\Psi_{t,1}$ &  $\Psi_{t,2}$ & $\Psi_{t,3}$ & $\Psi_{t,4}$ & $\Psi_{t,5}$ \\
            \midrule
            11 & 96.4 & 96.2 & 95.8 & 95.6 & 95.6 & 96.2 \\
            21 & 95.0 & 93.8 & 93.4 & 94.6 & 94.2 & 95.2 \\
            31 & 96.2 & 95.2 & 94.8 & 94.2 & 95.2 & 95.0 \\
            41 & 95.6 & 95.6 & 97.2 & 94.4 & 94.6 & 95.8 \\
            50 & 94.4 & 94.2 & 95.0 & 93.6 & 94.0 & 94.8 \\
            \bottomrule
        \end{tabular}
    }
\end{minipage}\hfill
    \vspace{0.5em}
    \begin{minipage}[t]{0.48\linewidth}
        \makebox[\textwidth]{\small (b) Scenario 2} 
        \resizebox{0.96\textwidth}{!}{
        \large
        \begin{tabular}{*{10}{c}}
            \toprule
           $t$ & $\Psi_{t,RCT}$ & $\Psi_{t,1}$ &  $\Psi_{t,2}$ & $\Psi_{t,3}$ & $\Psi_{t,4}$ & $\Psi_{t,5}$ \\
            \midrule
            11 & 96.2 & 95.2 & 94.8 & 95.0 & 95.4 & 94.6 \\
            21 & 94.0 & 95.2 & 95.2 & 94.4 & 94.4 & 95.0 \\
            31 & 94.4 & 94.2 & 94.2 & 94.8 & 94.0 & 95.4 \\
            41 & 94.8 & 95.0 & 94.4 & 94.0 & 95.4 & 95.4 \\
            50 & 94.6 & 94.2 & 94.8 & 94.2 & 94.6 & 94.8 \\
            \bottomrule
            \end{tabular}
        }
    \end{minipage}\hfill
\end{table}

\begin{table}[H]\centering
    \caption{Probability (\%) of not assigning optimal personalized treatments in Scenarios 1 and 2 ($t = 11,21,31,41,50$). $\Pnonopt_{t} := \frac{1}{E(t)}\sum_{i: t_i = t}I\left(A_i \ne d_{0,K}(W_i)\right)$, where $d_{0,K}(W) := I(\bQ_{0,K}(1,W)-\bQ_{0,K}(0,W)>0)$}
    \label{tab:incorrectTreatment}
    \begin{minipage}[t]{0.48\linewidth}
        \makebox[\textwidth]{\small (a) Scenario 1} 
        \resizebox{0.96\textwidth}{!}{
        \large
        \begin{tabular}{*{10}{c}}
            \toprule
            $t$ & $\Pnonopt_{t,SL}$ & $\Pnonopt_{t,1}$ &  $\Pnonopt_{t,2}$ & $\Pnonopt_{t,3}$ & $\Pnonopt_{t,4}$ & $\Pnonopt_{t,5}$  & $\Pnonopt_{t,RCT}$ \\
            \midrule
            11 & 28.1 & 51.7 & 40.4 & 29.9 & 20.2 & 15.9 & 50.2 \\
            21 & 20.0 & 51.2 & 39.9 & 29.7 & 19.7 & 14.4 & 49.7 \\
            31 & 17.8 & 51.0 & 40.9 & 30.3 & 19.6 & 13.5 & 50.5 \\
            41 & 15.8 & 50.4 & 40.1 & 30.3 & 19.9 & 13.3 & 49.9 \\
            50 & 15.0 & 50.8 & 40.2 & 30.1 & 19.7 & 12.4 & 50.0 \\  
            \bottomrule
        \end{tabular}
        }
\end{minipage}\hfill
    \begin{minipage}[t]{0.48\linewidth}
        \makebox[\textwidth]{\small (b) Scenario 2} 
        \resizebox{0.96\textwidth}{!}{
        \large
        \begin{tabular}{*{10}{c}}
            \toprule
            $t$ & $\Pnonopt_{t,SL}$ & $\Pnonopt_{t,1}$ &  $\Pnonopt_{t,2}$ & $\Pnonopt_{t,3}$ & $\Pnonopt_{t,4}$ & $\Pnonopt_{t,5}$  & $\Pnonopt_{t,RCT}$\\
            \midrule
            11 & 15.5 & 13.9 & 14.3 & 15.7 & 18.7 & 29.8 & 49.8 \\
            21 & 13.4 & 13.1 & 13.2 & 13.9 & 16.3 & 22.7 & 49.7 \\
            31 & 13.2 & 12.3 & 12.4 & 12.7 & 14.8 & 19.7 & 49.5 \\
            41 & 12.5 & 11.8 & 11.8 & 12.7 & 13.9 & 17.9 & 49.5 \\
            50 & 12.4 & 12.0 & 12.1 & 12.5 & 14.2 & 16.9 & 50.1 \\
            \bottomrule
        \end{tabular}
        }
    \end{minipage}\hfill
\end{table}

\begin{table}[H]\centering
    \caption{Average regret of not selecting optimal personalized treatments in Scenarios 1 and 2 ($t = 11,21,31,41,50$).}
    \label{tab:regret}
    \begin{minipage}[t]{0.48\linewidth}
        \makebox[\textwidth]{\small (a) Scenario 1} 
        \resizebox{0.96\textwidth}{!}{
        \large
        \begin{tabular}{*{10}{c}}
            \toprule
            $t$ & $R_{t,SL}$ & $R_{t,1}$ &  $R_{t,2}$ & $R_{t,3}$ & $R_{t,4}$ & $R_{t,5}$ & $R_{t,RCT}$ \\
            \midrule
            11 & 0.153 & 0.350 & 0.245 & 0.159 & 0.100 & 0.083 & 0.343 \\
            21 & 0.102 & 0.344 & 0.239 & 0.154 & 0.095 & 0.076 & 0.342 \\
            31 & 0.089 & 0.341 & 0.247 & 0.158 & 0.094 & 0.074 & 0.346 \\
            41 & 0.083 & 0.337 & 0.240 & 0.158 & 0.094 & 0.074 & 0.341 \\
            50 & 0.080 & 0.341 & 0.241 & 0.157 & 0.093 & 0.071 & 0.345 \\
            \bottomrule
        \end{tabular}
        }
\end{minipage}\hfill
    \begin{minipage}[t]{0.48\linewidth}

        \makebox[\textwidth]{\small (b) Scenario 2} 
        \resizebox{0.96\textwidth}{!}{
        \large
        \begin{tabular}{*{10}{c}}
            \toprule
            $t$ & $R_{t,SL}$ & $R_{t,1}$ &  $R_{t,2}$ & $R_{t,3}$ & $R_{t,4}$ & $R_{t,5}$ & $R_{t,RCT}$ \\
            \midrule
            11 & 0.030 & 0.026 & 0.026 & 0.028 & 0.033 & 0.062 & 0.120 \\
            21 & 0.026 & 0.026 & 0.026 & 0.026 & 0.030 & 0.043 & 0.120 \\
            31 & 0.025 & 0.024 & 0.025 & 0.025 & 0.027 & 0.035 & 0.119 \\
            41 & 0.024 & 0.023 & 0.023 & 0.024 & 0.025 & 0.031 & 0.118 \\
            50 & 0.025 & 0.024 & 0.024 & 0.024 & 0.025 & 0.029 & 0.120 \\
            \bottomrule
        \end{tabular}
        }
    \end{minipage}\hfill
\end{table}

\begin{figure}[htb]
    \centering    \includegraphics[width=0.475\textwidth]{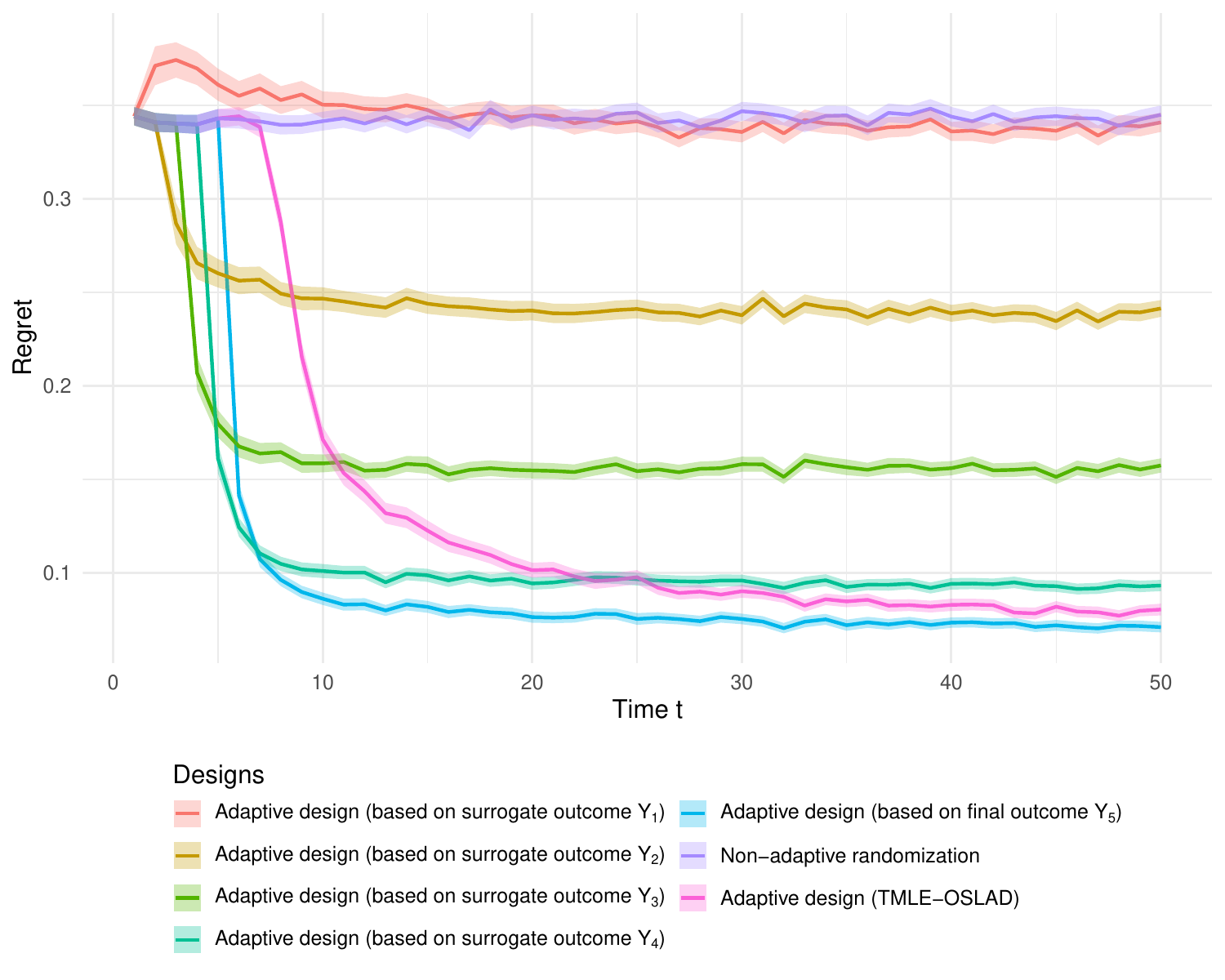}
    \includegraphics[width=0.475\textwidth]{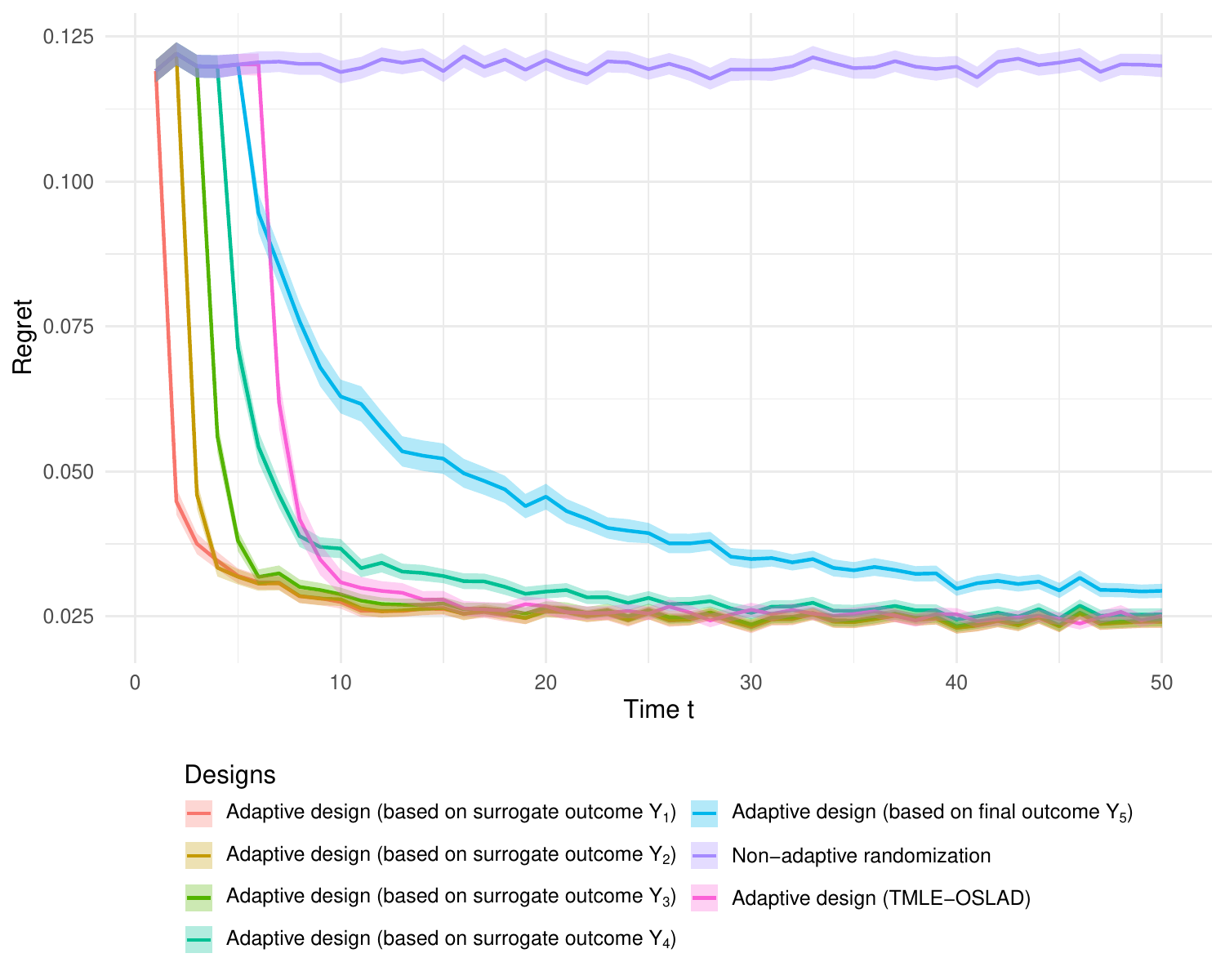}
    \caption{Regret at time $t$ for different design strategies in Scenario 1 (left) and Scenario 2 (right). This figure includes seven curves: five curves correspond to regrets from adaptive designs using outcomes from $Y_1$ to $Y_5$, respectively; one curve illustrates the regret for an adaptive design employing a TMLE-OSLAD to evaluate and utilize surrogates; and the last curve represents the regret of a fixed design, which assigns equal probabilities to each treatment group. The shaded regions represent twice the standard error of the regret across simulations.}
    \label{fig:regret}
\end{figure}

\begin{table}[thb]\centering
    \caption{Performance of TMLE estimators for the proposed adaptive design estimands at the end of experiments in Scenarios 1 and 2 across Monte Carlo simulations.}
    \label{tab:estimation_design_estimands_at_end}
    \makebox[\textwidth]{\small (a) Scenario 1} 
    \resizebox{0.96\textwidth}{!}{
    \large
    \begin{tabular}{*{10}{c}}
        \toprule
        Estimand & Notation & Truth & Bias ($\times 10^{-3}$) & Var ($\times 10^{-3}$) & Coverage  (\%)\\
        \midrule
        Expected final outcome $Y_5$ under the adaptive design using $Y_1$ & $\Psi_{\Tend,1}$ & 0.000 & -0.88 & 2.20 & 94.0 \\
        Expected final outcome $Y_5$ under the adaptive design using $Y_2$ & $\Psi_{\Tend,2}$ & 0.097 & -1.03 & 1.49 & 94.4 \\
        Expected final outcome $Y_5$ under the adaptive design using $Y_3$ & $\Psi_{\Tend,3}$ & 0.175 & 0.01 & 1.06 & 93.6 \\
        Expected final outcome $Y_5$ under the adaptive design using $Y_4$ & $\Psi_{\Tend,4}$ & 0.225 & 0.92 & 0.67 & 94.2 \\
        Expected final outcome $Y_5$ under the adaptive design using $Y_5$ & $\Psi_{\Tend,5}$ & 0.238 & 0.06 & 0.59 & 94.6 \\ \bottomrule
    \end{tabular}
    }
    \vspace{0.5em} 
    
    \makebox[\textwidth]{\small (b) Scenario 2} 
    \resizebox{0.96\textwidth}{!}{
    \large
    \begin{tabular}{*{10}{c}}
        \toprule
        Estimand & Notation & Truth & Bias ($\times 10^{-3}$) & Var ($\times 10^{-3}$) & Coverage (\%)\\
        \midrule
        Expected Final Outcome $Y_5$ under the adaptive design using $Y_1$ & $\Psi_{\Tend,1}$ & 0.093 & -0.46 & 0.48 & 94.6 \\
        Expected Final Outcome $Y_5$ under the adaptive design using $Y_2$ & $\Psi_{\Tend,2}$ & 0.091 & -0.70 & 0.50 & 94.4 \\
        Expected Final Outcome $Y_5$ under the adaptive design using $Y_3$ & $\Psi_{\Tend,3}$ & 0.088 & -1.01 & 0.49 & 94.2 \\
        Expected Final Outcome $Y_5$ under the adaptive design using $Y_4$ & $\Psi_{\Tend,4}$ & 0.083 & -0.62 & 0.52 & 94.6 \\
        Expected Final Outcome $Y_5$ under the adaptive design using $Y_5$ & $\Psi_{\Tend,5}$ & 0.068 & -0.88 & 0.60 & 94.4 \\ 
        \bottomrule
    \end{tabular}
    }
\end{table}

Table \ref{tab:truePsi_marginal} reports the true expected final outcome $Y_5$ obtained from implementing different adaptive designs at illustrative time points $t = 11, 21, 31, 41$, and $50$, respectively.
These values represent the ideal estimand as discussed in Appendix \ref{appendix:ideal_estimand}.
We use $\tPsi_{t,k}$ to denote the ideal estimand under the adaptive design guided by surrogate $Y_k$. We further use $\tPsi_{t,RCT}$ and $\tPsi_{t,SL}$ to denote the true expected final outcomes obtained from implementing the non-adaptive RCT and our proposed TMLE-OSLAD design, respectively.
Comparing with the proposed context-specific estimands $\Psi_{t,k}$ in Table \ref{tab:truePsi}, the two estimands show minimal discrepancy, and the relative superiority between different surrogate-guided designs remains unchanged.

\begin{table}[thb]\centering
    \caption{True expected final outcome $\tPsi_{t,k}$ of implementing different adaptive designs in Scenarios 1 and 2 ($t = 11,21,31,41,50,T_{\textit{end}}$).}
    \label{tab:truePsi_marginal}
    \vspace{0.5em}
    \begin{minipage}[t]{0.48\linewidth}
        \makebox[\textwidth]{\small (a) Scenario 1} 
        \resizebox{0.96\textwidth}{!}{
        \large
        \begin{tabular}{*{10}{c}}
            \toprule
           $t$ & $\tPsi_{t,RCT}$ & $\tPsi_{t,1}$ &  $\tPsi_{t,2}$ & $\tPsi_{t,3}$ & $\tPsi_{t,4}$ & $\tPsi_{t,5}$ & $\tPsi_{t,SL}$\\
            \midrule
            11 & 0.000 & -0.019 & 0.050 & 0.078 & 0.066 & 0.033 & 0.000 \\
            21 & 0.000 & -0.012 & 0.079 & 0.144 & 0.176 & 0.172 & 0.101 \\
            31 & 0.000 & -0.007 & 0.088 & 0.161 & 0.203 & 0.208 & 0.155 \\
            41 & 0.000 & -0.004 & 0.092 & 0.168 & 0.216 & 0.225 & 0.183 \\
            50 & 0.000 & -0.002 & 0.095 & 0.172 & 0.222 & 0.234 & 0.199 \\
            $T_{\textit{end}}$ & 0.000 & -0.001 & 0.096 & 0.174 & 0.225 & 0.238 & 0.205 \\
            \bottomrule
        \end{tabular}
    }
\end{minipage}\hfill
    \vspace{0.5em}
    \begin{minipage}[t]{0.48\linewidth}

        \makebox[\textwidth]{\small (b) Scenario 2} 
        \resizebox{0.96\textwidth}{!}{
        \large
        \begin{tabular}{*{10}{c}}
            \toprule
           $t$ & $\tPsi_{t,RCT}$ & $\tPsi_{t,1}$ &  $\tPsi_{t,2}$ & $\tPsi_{t,3}$ & $\tPsi_{t,4}$ & $\tPsi_{t,5}$ & $\tPsi_{t,SL}$\\
            \midrule
            11 & 0.000 & 0.070 & 0.057 & 0.039 & 0.019 & 0.004 & 0.000 \\
            21 & 0.000 & 0.085 & 0.079 & 0.072 & 0.060 & 0.038 & 0.054 \\
            31 & 0.000 & 0.089 & 0.085 & 0.081 & 0.072 & 0.053 & 0.069 \\
            41 & 0.000 & 0.090 & 0.088 & 0.085 & 0.078 & 0.062 & 0.076 \\
            50 & 0.000 & 0.091 & 0.090 & 0.087 & 0.081 & 0.067 & 0.080 \\
            $T_{\textit{end}}$ & 0.000 & 0.092 & 0.090 & 0.088 & 0.082 & 0.069 & 0.082 \\
            \bottomrule
        \end{tabular}
        }
    \end{minipage}\hfill
\end{table}

\begingroup
\color{vz}

\section{Sensitivity Analyses}
\label{app:sensitivity}
We further conduct a series of sensitivity analyses to compare the performance of TMLE-OSLAD across different design configurations.

First, we perform sensitivity analyses to assess how TMLE-OSLAD’s surrogate selection pattern changes with different $\alpha$'s used in (i) the $100(1-\alpha)\%$ TMLE-based Wald-type confidence interval for evaluating and choosing from candidate designs, and (ii) the $100(1-\alpha)\%$ CATE confidence intervals for generating treatment randomization probabilities.
Figures \ref{fig:sim1-sensi-tmle-alpha-freq} and \ref{fig:sim2-sensi-tmle-alpha-freq} report the selection frequency of each surrogate-guided design at each time point when the TMLE-based selection criterion, originally based on 95\% CIs, is re-evaluated using 90\% and 99\% confidence intervals. The selection patterns remain consistent across these choices.
\begin{figure}[H]
    \centering
    \includegraphics[width=0.95\linewidth]{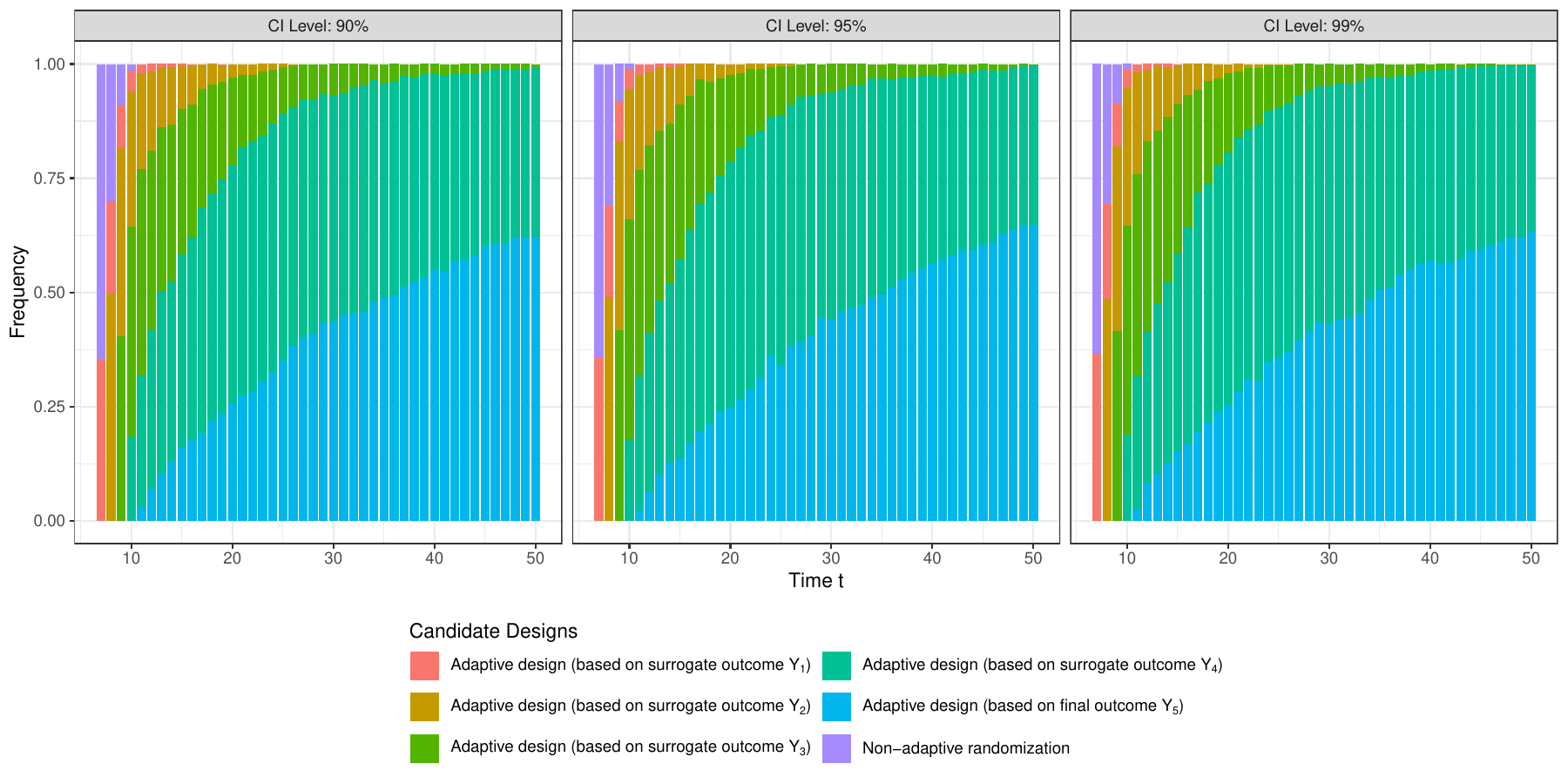}
    \caption{Frequency of candidate designs selected by TMLE-OSLAD using different TMLE confidence interval levels in Scenario 1.}
    \label{fig:sim1-sensi-tmle-alpha-freq}
\end{figure}
\begin{figure}[H]
    \centering
    \includegraphics[width=0.95\linewidth]{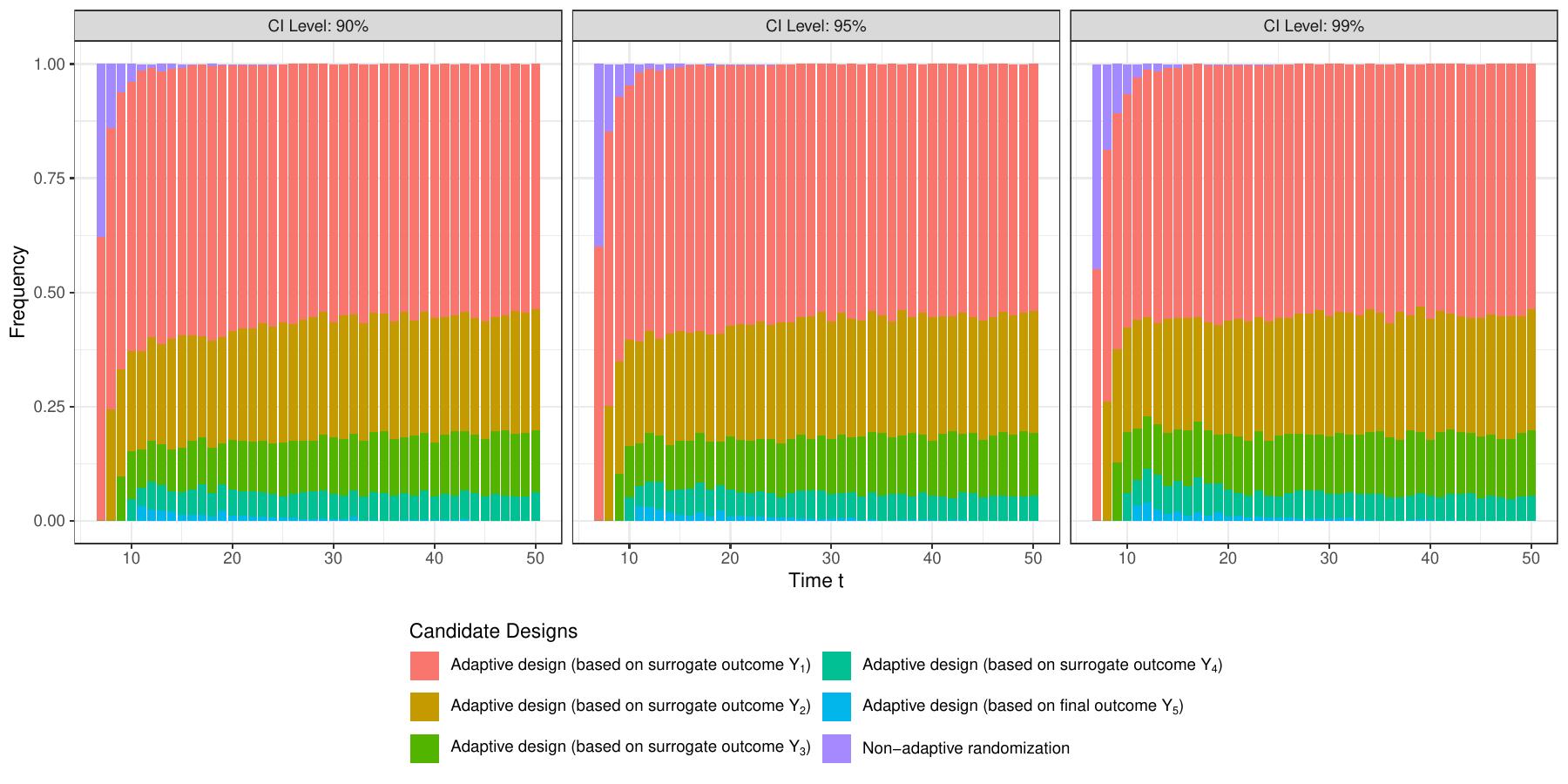}
    \caption{Frequency of candidate designs selected by TMLE-OSLAD using different TMLE confidence interval levels in Scenario 2.}
    \label{fig:sim2-sensi-tmle-alpha-freq}
\end{figure}

Figures \ref{fig:sim1-sensi-cate-alpha-freq} and \ref{fig:sim2-sensi-cate-alpha-freq} report the selection frequency of each candidate surrogate-driven design obtained by implementing adaptive designs where treatment-randomization probabilities are based upon $100(1-\alpha)\%$ CATE confidence intervals with $\alpha \in \{0.01, 0.05, 0.10\}$ for Scenarios 1 and 2, respectively. The selection patterns are consistent across different CI levels.
\begin{figure}[H]
    \centering
    \includegraphics[width=0.95\linewidth]{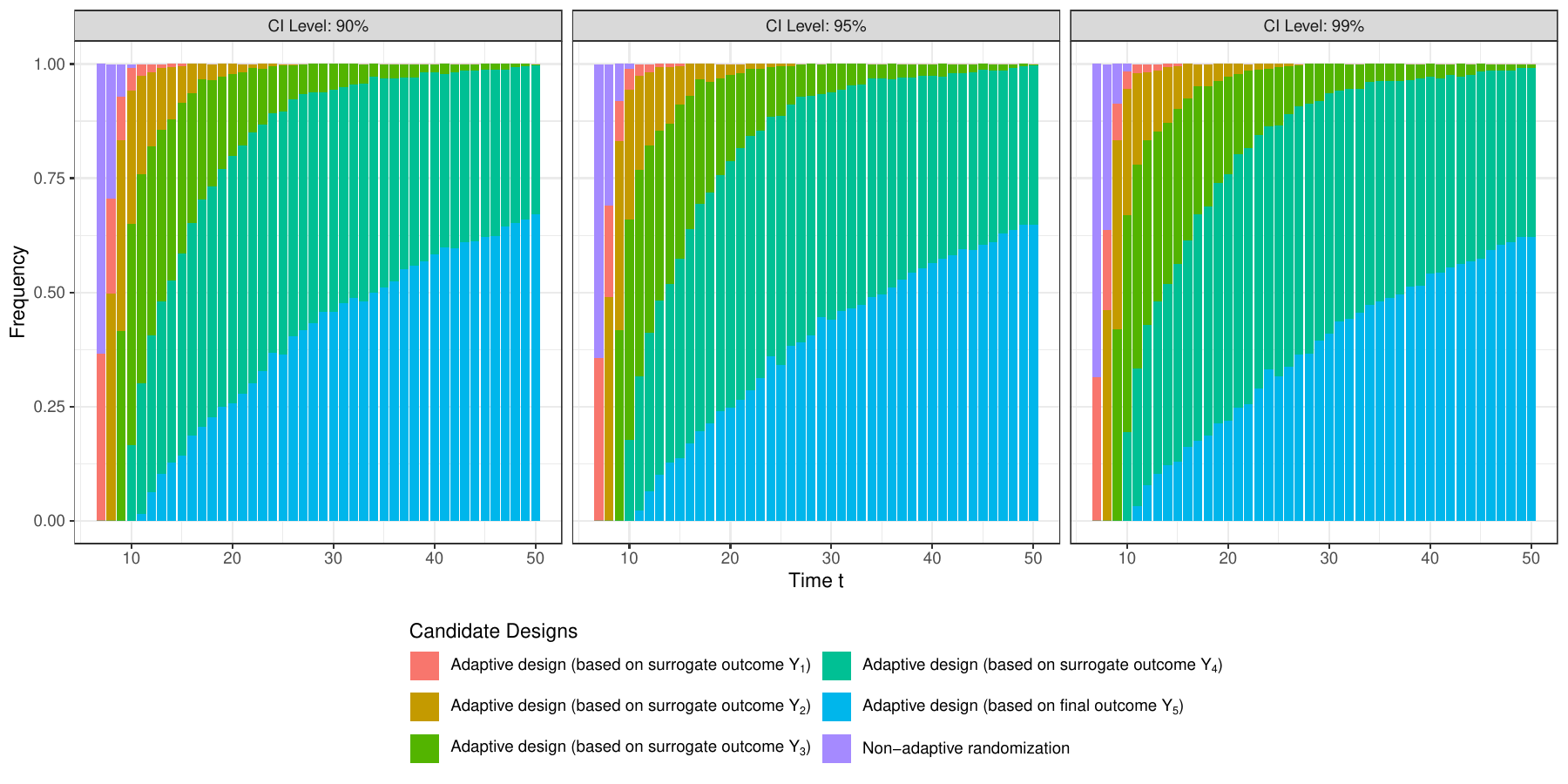}
    \caption{Frequency of candidate designs selected by TMLE-OSLAD using different levels of CATE confidence interval to adapt treatment randomization probabilities in Scenario 1.}
    \label{fig:sim1-sensi-cate-alpha-freq}
\end{figure}
\begin{figure}[H]
    \centering
    \includegraphics[width=0.95\linewidth]{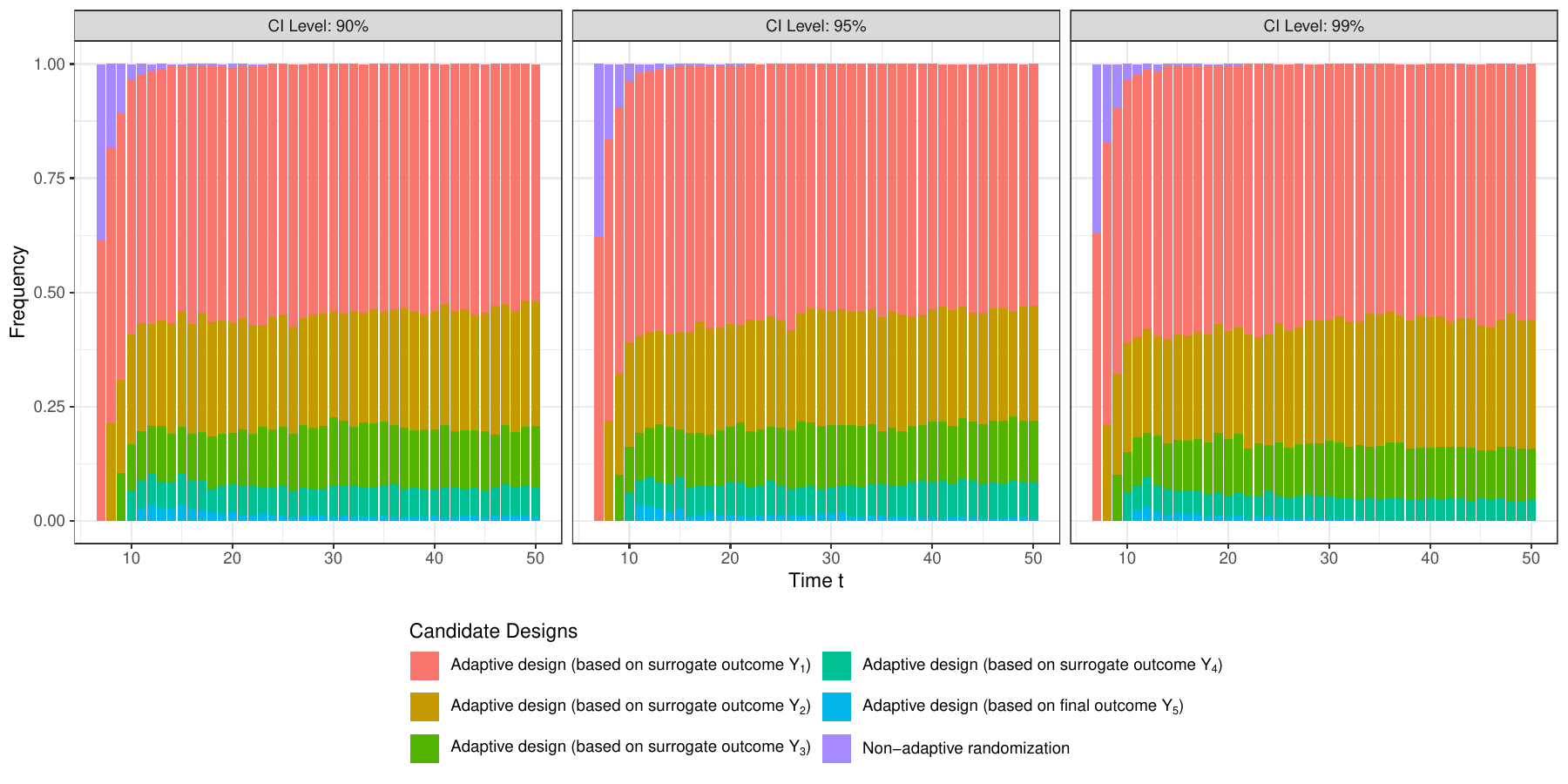}
    \caption{Frequency of candidate designs selected by TMLE-OSLAD using different levels of CATE confidence interval to adapt treatment randomization probabilities in Scenario 2.}
    \label{fig:sim2-sensi-cate-alpha-freq}
\end{figure}


Second, we conduct a sensitivity analysis on our proposed TMLE-OSLAD design with a different choice of $h_{\nu}$.
As noted in Appendix \ref{appendix:details_of_CARA}, the treatment randomization probabilities
 we use is
\begin{eqnarray}
h_\nu(z) = 
\begin{cases}
    \nu, & z \leq -1 \\
    \nu + (1 - 2\nu)\left( -\dfrac{1}{4}z^3 + \dfrac{3}{4}z + \dfrac{1}{2} \right), & -1 < z < 1 \\ 
    1 - \nu, & z \geq 1
\end{cases},
\end{eqnarray}
where the input \( z \) is the ratio of the CATE estimate to half the width of its \( 100(1-\alpha)\% \) confidence interval. 
Let
$
f^{(0)}
$
denote the function used in the second term of \( h_\nu \) for \( -1 < z < 1 \), defined as $f^{(0)}(z) = -\frac{1}{4}z^3 + \frac{3}{4}z + \frac{1}{2}$. 
To examine the sensitivity of our proposed methods to different treatment randomization schemes, we consider three additional functions that vary in how aggressively the estimated CATE signal is transformed into a skewed treatment randomization probability favoring the estimated optimal personalized treatment. These functions are defined as follows: (i) $f^{(1)}(z) = 1/(1 + \exp(-8z))$, (ii) $f^{(2)}(z) = \frac{1}{4}z^3 + \frac{1}{4}z + \frac{1}{2}$, and (iii) $f^{(3)}(z) = -\frac{1}{4}z^5 + \frac{3}{4}z^3 + \frac{1}{2}$.
To facilitate comparison, we use \( h_{\nu}^{(0)} \) to denote the baseline mapping function \( h_{\nu} \) in the main text. We then define \( h_{\nu}^{(1)}(z) \), \( h_{\nu}^{(2)}(z) \), and \( h_{\nu}^{(3)}(z) \) as alternative mapping functions that replace \( f^{(0)}(z) \) with \( f^{(1)}(z) \), \( f^{(2)}(z) \), and \( f^{(3)}(z) \) for $-1 < z < 1$, respectively.
Figure~\ref{fig:sensitivity_functions} compares the four functions. Relative to $h_{\nu}^{(0)}$, $h_{\nu}^{(1)}$ shifts the treatment randomization probability toward the estimated optimal treatment more rapidly as CATE increases,
while $h_{\nu}^{(2)}$ adjusts more gradually, and $h_{\nu}^{(3)}$ exhibits the slowest transition.
\begin{figure}
    \centering
    \includegraphics[width=0.75\linewidth]{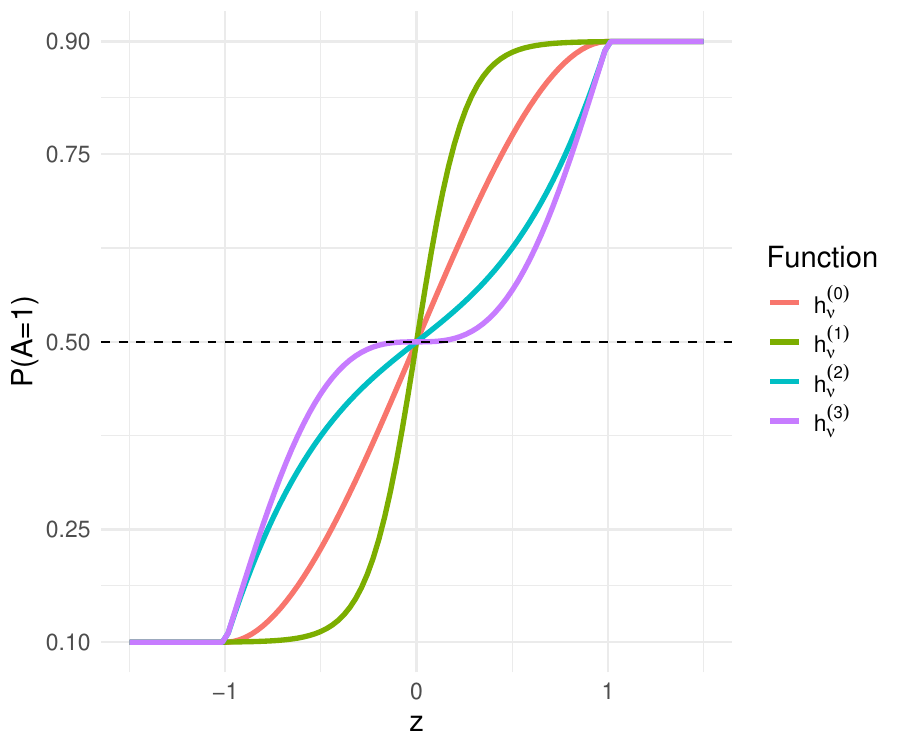}
    \caption{Comparison of the baseline mapping \( h_{\nu}^{(0)} \) and three alternatives \( h_{\nu}^{(1)} \), \( h_{\nu}^{(2)} \), and \( h_{\nu}^{(3)} \) (\( \nu = 0.1 \)) for sensitivity analysis. The x-axis shows the standardized CATE $(z)$, and the y-axis shows the probability of assigning treatment \( a=1 \) according to each mapping function (shown as curves).}    \label{fig:sensitivity_functions}
    \end{figure}

\begin{sidewaysfigure}
\renewcommand{\arraystretch}{1.2}

\noindent
\begin{minipage}[t]{\textheight}
\vspace{0pt}  
\centering
\captionsetup{width=\textheight}
\captionof{table}{True target estimand $\Psi_{t,k}$ had participants enrolled from $1, \cdots, t-K$ were adaptively treated by a candidate surrogate $Y_k$ in Scenarios 1 ($t = 11,21,31,41,50$) under different mapping functions \( h_{\nu}^{(0)} \) to \( h_{\nu}^{(3)} \). Expected outcomes under the adaptive design using the oracle surrogate are shown in bold.}
\label{tab:sensi-Psi-sim1}

\resizebox{\textwidth}{!}{   
\begin{tabular}{@{}c|cccccc|cccccc|cccccc|cccccc@{}}
\toprule
& \multicolumn{6}{c|}{\( h_{\nu}^{(0)} \)}
& \multicolumn{6}{c|}{\( h_{\nu}^{(1)} \)}
& \multicolumn{6}{c|}{\( h_{\nu}^{(2)} \)}
& \multicolumn{6}{c}{\( h_{\nu}^{(3)} \)} \\
\cmidrule{2-7} \cmidrule{8-13} \cmidrule{14-19} \cmidrule{20-25}
$t$ & $\Psi_{t,RCT}$ & $\Psi_{t,1}$ &  $\Psi_{t,2}$ & $\Psi_{t,3}$ & $\Psi_{t,4}$ & $\Psi_{t,5}$
   & $\Psi_{t,RCT}$ & $\Psi_{t,1}$ &  $\Psi_{t,2}$ & $\Psi_{t,3}$ & $\Psi_{t,4}$ & $\Psi_{t,5}$
   & $\Psi_{t,RCT}$ & $\Psi_{t,1}$ &  $\Psi_{t,2}$ & $\Psi_{t,3}$ & $\Psi_{t,4}$ & $\Psi_{t,5}$
   & $\Psi_{t,RCT}$ & $\Psi_{t,1}$ &  $\Psi_{t,2}$ & $\Psi_{t,3}$ & $\Psi_{t,4}$ & $\Psi_{t,5}$ \\
\midrule

11 & 0.000 & -0.016 & 0.053 & \textbf{0.078} & 0.066 & 0.033 
   & 0.000 & -0.013 & 0.058 & \textbf{0.083} & 0.071 & 0.036
   & 0.000 & -0.015 & 0.047 & \textbf{0.070} & 0.060 & 0.030 
   & 0.000 & -0.015 & 0.045 & \textbf{0.068} & 0.058 & 0.029 \\
21 & 0.000 & -0.008 & 0.083 & 0.145 & \textbf{0.177} & 0.173
   & 0.000 & -0.004 & 0.086 & 0.148 & \textbf{0.181} & 0.176
   & 0.000 & -0.009 & 0.079 & 0.140 & \textbf{0.171} & 0.166
   & 0.000 & -0.010 & 0.078 & 0.138 & \textbf{0.169} & 0.164 \\
31 & 0.000 & -0.005 & 0.091 & 0.161 & 0.203 & \textbf{0.208} 
   & 0.000 & -0.002 & 0.093 & 0.163 & 0.207 & \textbf{0.211}
   & 0.000 & -0.007 & 0.088 & 0.158 & 0.199 & \textbf{0.202}
   & 0.000 & -0.008 & 0.087 & 0.156 & 0.197 & \textbf{0.201} \\
41 & 0.000 & -0.002 & 0.094 & 0.169 & 0.215 & \textbf{0.225} 
   & 0.000 & 0.000 & 0.096 & 0.170 & 0.218 & \textbf{0.228}
   & 0.000 & -0.004 & 0.092 & 0.166 & 0.212 & \textbf{0.220} 
   & 0.000 & -0.006 & 0.091 & 0.165 & 0.211 & \textbf{0.218} \\
50 & 0.000 & -0.001 & 0.096 & 0.173 & 0.222 & \textbf{0.234}
   & 0.000 & 0.001 & 0.098 & 0.174 & 0.225 & \textbf{0.236}
   & 0.000 & -0.003 & 0.094 & 0.170 & 0.219 & \textbf{0.229}
   & 0.000 & -0.004 & 0.094 & 0.170 & 0.218 & \textbf{0.228} \\
\bottomrule
\end{tabular}
}
\end{minipage}

\hfill


\noindent
\begin{minipage}[t]{\textheight}
\vspace{0pt}  
\centering
\captionsetup{width=\textheight}
\captionof{table}{True target estimand $\Psi_{t,k}$ had participants enrolled from $1, \cdots, t-K$ were adaptively treated by a candidate surrogate $Y_k$ in Scenarios 2 ($t = 11,21,31,41,50$) under different mapping functions \( h_{\nu}^{(0)} \) to \( h_{\nu}^{(3)} \). Expected outcomes under the adaptive design using the oracle surrogate are shown in bold.}
\label{tab:sensi-Psi-sim2}

\resizebox{\textwidth}{!}{ 
\begin{tabular}{@{}c|cccccc|cccccc|cccccc|cccccc@{}}
\toprule
& \multicolumn{6}{c|}{\( h_{\nu}^{(0)} \)}
& \multicolumn{6}{c|}{\( h_{\nu}^{(1)} \)}
& \multicolumn{6}{c|}{\( h_{\nu}^{(2)} \)}
& \multicolumn{6}{c}{\( h_{\nu}^{(3)} \)} \\
\cmidrule{2-7} \cmidrule{8-13} \cmidrule{14-19} \cmidrule{20-25}
$t$ & $\Psi_{t,RCT}$ & $\Psi_{t,1}$ &  $\Psi_{t,2}$ & $\Psi_{t,3}$ & $\Psi_{t,4}$ & $\Psi_{t,5}$
   & $\Psi_{t,RCT}$ & $\Psi_{t,1}$ &  $\Psi_{t,2}$ & $\Psi_{t,3}$ & $\Psi_{t,4}$ & $\Psi_{t,5}$
   & $\Psi_{t,RCT}$ & $\Psi_{t,1}$ &  $\Psi_{t,2}$ & $\Psi_{t,3}$ & $\Psi_{t,4}$ & $\Psi_{t,5}$
  & $\Psi_{t,RCT}$ & $\Psi_{t,1}$ &  $\Psi_{t,2}$ & $\Psi_{t,3}$ & $\Psi_{t,4}$ & $\Psi_{t,5}$ \\
\midrule

11 & 0.000 & \textbf{0.073} & 0.057 & 0.039 & 0.019 & 0.004
   & 0.000 & \textbf{0.075} & 0.060 & 0.042 & 0.021 & 0.005
   & 0.000 & \textbf{0.069} & 0.054 & 0.035 & 0.015 & 0.003
   & 0.000 & \textbf{0.068} & 0.053 & 0.034 & 0.014 & 0.003 \\
21 & 0.000 & \textbf{0.087} & 0.080 & 0.073 & 0.061 & 0.038 
   & 0.000 & \textbf{0.088} & 0.082 & 0.075 & 0.064 & 0.042 
   & 0.000 & \textbf{0.084} & 0.078 & 0.070 & 0.056 & 0.031
   & 0.000 & \textbf{0.084} & 0.078 & 0.070 & 0.054 & 0.028 \\
31 & 0.000 & \textbf{0.090} & 0.086 & 0.081 & 0.073 & 0.052
   & 0.000 & \textbf{0.091} & 0.087 & 0.082 & 0.075 & 0.057
   & 0.000 & \textbf{0.088} & 0.085 & 0.079 & 0.068 & 0.043 
   & 0.000 & \textbf{0.088} & 0.084 & 0.079 & 0.067 & 0.041 \\
41 & 0.000 & \textbf{0.091} & 0.089 & 0.085 & 0.078 & 0.060
   & 0.000 & \textbf{0.092} & 0.089 & 0.086 & 0.081 & 0.065
   & 0.000 & \textbf{0.090} & 0.087 & 0.084 & 0.074 & 0.051
   & 0.000 & \textbf{0.090} & 0.087 & 0.083 & 0.073 & 0.049 \\
50 & 0.000 & \textbf{0.092} & 0.090 & 0.087 & 0.081 & 0.065
   & 0.000 & \textbf{0.093} & 0.091 & 0.088 & 0.083 & 0.070 
   & 0.000 & \textbf{0.091} & 0.089 & 0.086 & 0.078 & 0.057
   & 0.000 & \textbf{0.091} & 0.089 & 0.085 & 0.077 & 0.054 \\
\bottomrule
\end{tabular}
}
\end{minipage}

\end{sidewaysfigure}

Tables~\ref{tab:sensi-Psi-sim1} and \ref{tab:sensi-Psi-sim2} report the true target estimand $\Psi_{t,k}$ for each surrogate-specific adaptive design under the different mapping functions $h_{\nu}^{(0)}$ through $h_{\nu}^{(3)}$.  
The oracle surrogate that yields the highest target estimand is consistent across different mapping functions. 
Figures \ref{fig:sim1-sensi-freq} and \ref{fig:sim2-sensi-freq} further show the selection frequency for each candidate design by TMLE-OSLAD across different mapping functions used for treatment randomization. 
The surrogate selection pattern is robust to the choice of the function, with TMLE-OSLAD consistently identifying and increasingly favoring the oracle surrogate as the experiment progresses.
Moreover, TMLE provides robust inference for these target estimands, with confidence intervals achieving nominal coverage across simulations using different mapping functions (Tables~\ref{tab:sensi-cov-sim1} and \ref{tab:sensi-cov-sim2}). 

\begin{figure}[H]
    \centering
    \includegraphics[width=0.99\linewidth]{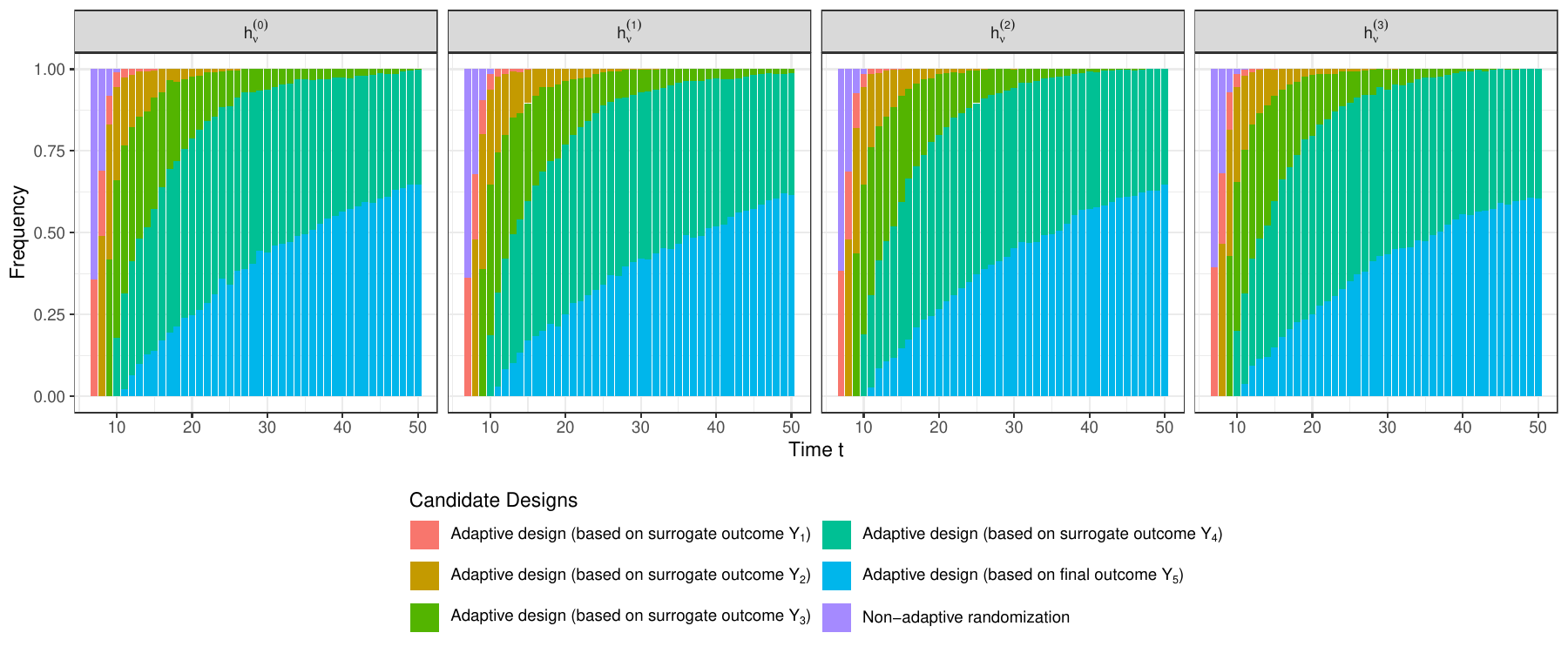}
    \caption{Frequency of candidate designs selected by TMLE-OSLAD under different treatment randomization functions applied in Scenario 1.}
    \label{fig:sim1-sensi-freq}
\end{figure}
\begin{figure}[H]
    \centering
    \includegraphics[width=0.99\linewidth]{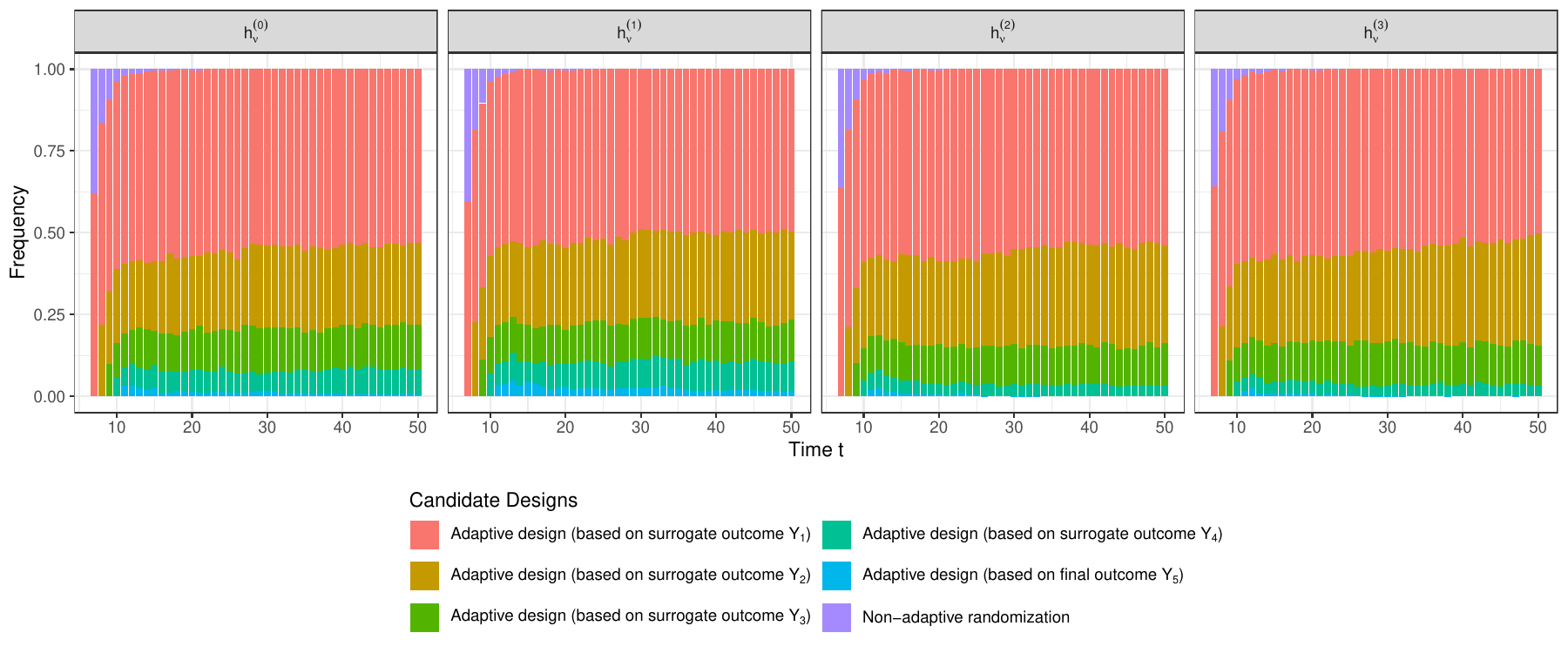}
    \caption{Frequency of candidate designs selected by TMLE-OSLAD under different treatment randomization functions applied in Scenario 2.}
    \label{fig:sim2-sensi-freq}
\end{figure}

\begin{sidewaysfigure}
\renewcommand{\arraystretch}{1.2}

\noindent
\begin{minipage}[t]{\textheight}
\vspace{0pt}  
\centering
\captionsetup{width=\textheight}
\captionof{table}{Coverage probability (\%) of confidence interval for target estimands $\Psi_{t,k}$ in Scenarios 1 ($t = 11,21,31,41,50$) under different mapping functions \( h_{\nu}^{(0)} \) to \( h_{\nu}^{(3)} \).}
\label{tab:sensi-cov-sim1}


\resizebox{\textwidth}{!}{   
\begin{tabular}{@{}c|cccccc|cccccc|cccccc|cccccc@{}}
\toprule
& \multicolumn{6}{c|}{\( h_{\nu}^{(0)} \)}
& \multicolumn{6}{c|}{\( h_{\nu}^{(1)} \)}
& \multicolumn{6}{c|}{\( h_{\nu}^{(2)} \)}
& \multicolumn{6}{c}{\( h_{\nu}^{(3)} \)} \\
\cmidrule{2-7} \cmidrule{8-13} \cmidrule{14-19} \cmidrule{20-25}
$t$ & $\Psi_{t,RCT}$ & $\Psi_{t,1}$ &  $\Psi_{t,2}$ & $\Psi_{t,3}$ & $\Psi_{t,4}$ & $\Psi_{t,5}$
   & $\Psi_{t,RCT}$ & $\Psi_{t,1}$ &  $\Psi_{t,2}$ & $\Psi_{t,3}$ & $\Psi_{t,4}$ & $\Psi_{t,5}$
   & $\Psi_{t,RCT}$ & $\Psi_{t,1}$ &  $\Psi_{t,2}$ & $\Psi_{t,3}$ & $\Psi_{t,4}$ & $\Psi_{t,5}$
   & $\Psi_{t,RCT}$ & $\Psi_{t,1}$ &  $\Psi_{t,2}$ & $\Psi_{t,3}$ & $\Psi_{t,4}$ & $\Psi_{t,5}$ \\
\midrule
11 & 96.4 & 96.2 & 95.8 & 95.6 & 95.6 & 96.2
   & 96.8 & 96.2 & 95.0 & 95.4 & 95.8 & 96.4
   & 96.8 & 95.8 & 95.8 & 95.6 & 95.8 & 96.6
   & 96.8 & 95.6 & 95.8 & 95.4 & 96.0 & 96.6 \\
21 & 95.0 & 93.8 & 93.4 & 94.6 & 94.2 & 95.2
   & 94.2 & 95.0 & 93.6 & 94.4 & 92.4 & 94.8
   & 93.2 & 93.6 & 92.0 & 94.2 & 92.2 & 94.8
   & 93.2 & 94.0 & 92.4 & 93.8 & 93.2 & 94.4 \\
31 & 96.2 & 95.2 & 94.8 & 94.2 & 95.2 & 95.0
   & 94.4 & 94.6 & 94.2 & 93.0 & 94.0 & 95.6
   & 95.4 & 94.0 & 93.2 & 94.0 & 94.4 & 94.0
   & 96.6 & 94.8 & 93.4 & 94.0 & 94.6 & 94.8 \\
41 & 95.6 & 95.6 & 97.2 & 94.4 & 94.6 & 95.8
   & 96.2 & 96.8 & 96.6 & 95.4 & 95.6 & 95.0
   & 95.4 & 95.4 & 95.4 & 95.4 & 95.2 & 95.6
   & 95.8 & 96.2 & 96.8 & 95.4 & 94.4 & 95.8 \\
50 & 94.4 & 94.2 & 95.0 & 93.6 & 94.0 & 94.8
   & 95.8 & 95.2 & 94.2 & 93.4 & 94.4 & 94.6
   & 95.8 & 93.6 & 93.8 & 95.0 & 95.0 & 95.0
   & 95.8 & 95.2 & 95.8 & 95.0 & 94.4 & 95.2 \\
\bottomrule
\end{tabular}
}
\end{minipage}

\hfill

\noindent
\begin{minipage}[t]{\textheight}
\vspace{0pt}  
\centering
\captionsetup{width=\textheight}
\captionof{table}{Coverage probability (\%) of confidence interval for target estimands $\Psi_{t,k}$ in Scenarios 2 ($t = 11,21,31,41,50$) under different mapping functions \( h_{\nu}^{(0)} \) to \( h_{\nu}^{(3)} \).}
\label{tab:sensi-cov-sim2}

\resizebox{\textwidth}{!}{ 
\begin{tabular}{@{}c|cccccc|cccccc|cccccc|cccccc@{}}
\toprule
& \multicolumn{6}{c|}{\( h_{\nu}^{(0)} \)}
& \multicolumn{6}{c|}{\( h_{\nu}^{(1)} \)}
& \multicolumn{6}{c|}{\( h_{\nu}^{(2)} \)}
& \multicolumn{6}{c}{\( h_{\nu}^{(3)} \)} \\
\cmidrule{2-7} \cmidrule{8-13} \cmidrule{14-19} \cmidrule{20-25}
$t$ & $\Psi_{t,RCT}$ & $\Psi_{t,1}$ &  $\Psi_{t,2}$ & $\Psi_{t,3}$ & $\Psi_{t,4}$ & $\Psi_{t,5}$
   & $\Psi_{t,RCT}$ & $\Psi_{t,1}$ &  $\Psi_{t,2}$ & $\Psi_{t,3}$ & $\Psi_{t,4}$ & $\Psi_{t,5}$
   & $\Psi_{t,RCT}$ & $\Psi_{t,1}$ &  $\Psi_{t,2}$ & $\Psi_{t,3}$ & $\Psi_{t,4}$ & $\Psi_{t,5}$
  & $\Psi_{t,RCT}$ & $\Psi_{t,1}$ &  $\Psi_{t,2}$ & $\Psi_{t,3}$ & $\Psi_{t,4}$ & $\Psi_{t,5}$ \\
\midrule

11 & 96.2 & 95.2 & 94.8 & 95.0 & 95.4 & 94.6
   & 95.8 & 96.0 & 95.2 & 95.0 & 95.6 & 94.0
   & 95.8 & 94.8 & 95.0 & 95.6 & 95.4 & 95.4
   & 95.8 & 94.8 & 94.8 & 95.6 & 95.4 & 95.4 \\
21 & 94.0 & 95.2 & 95.2 & 94.4 & 94.4 & 95.0
   & 93.6 & 94.4 & 95.0 & 95.2 & 94.0 & 95.4
   & 93.6 & 94.8 & 95.8 & 95.6 & 94.2 & 95.2
   & 94.0 & 96.2 & 96.4 & 95.2 & 95.2 & 95.6 \\
31 & 94.4 & 94.2 & 94.2 & 94.8 & 94.0 & 95.4
   & 94.4 & 94.6 & 93.4 & 93.2 & 94.2 & 94.8 
   & 94.0 & 94.0 & 93.8 & 94.4 & 94.0 & 95.0
   & 95.6 & 94.0 & 93.6 & 95.0 & 93.4 & 94.6 \\
41 & 94.8 & 95.0 & 94.4 & 94.0 & 95.4 & 95.4
   & 93.6 & 94.0 & 94.4 & 93.2 & 93.8 & 96.0
   & 93.6 & 93.8 & 93.2 & 94.2 & 94.2 & 95.4
   & 94.2 & 94.4 & 94.0 & 94.0 & 94.4 & 94.6 \\
50 & 94.6 & 94.2 & 94.8 & 94.2 & 94.6 & 94.8
   & 94.4 & 93.6 & 95.2 & 94.0 & 94.6 & 95.6
   & 94.4 & 94.4 & 93.4 & 93.6 & 95.0 & 94.8
   & 94.8 & 93.8 & 93.8 & 93.0 & 94.2 & 94.8 \\
\bottomrule
\end{tabular}
}
\end{minipage}

\end{sidewaysfigure}


Lastly, we extend the simulation scenarios to include three covariates to further evaluate the performance of our proposed methods. 
In Scenario 1, the conditional mean functions of $Y_k$ given $A$ and $W$ are $E[Y_k|A,W_1,W_2,W_3] = (2A-1) \times \left[0.5 - (1 + \text{exp}(-(3-k) - 0.4W_1-0.4W_2-0.2W_3))^{-1}\right]$. In Scenario 2, the conditional mean functions of $Y_k$ are $E[Y_k|A,W_1,W_2,W_3] = (2A-1) \times \left[0.5 - (1 + \text{exp}(-\gamma_k \times (0.4W_1+0.4W_2+0.2W_3)))^{-1}\right]$, where $\gamma_1 = 3$,$\gamma_2 = 2$,$\gamma_3 = 1$,$\gamma_4 = 0.5$ and $\gamma_5 = 0.25$. $W_1$, $W_2$, and $W_3$ are independent and all follow a uniform distribution $U(-4,4)$. Let $V:= 0.4W_1+0.4W_2+0.2W_3$. The CATE functions for the candidate outcomes $Y_1,\ldots,Y_5$ as functions of $V$ follow the same patterns shown in Figure \ref{fig:scenarios}.

Table~\ref{tab:dim3-truePsi} reports the true target estimand $\Psi_{t,k}$. Table~\ref{tab:dim3-coverage} shows the coverage probabilities of TMLE-based confidence intervals for these estimands. Figure~\ref{fig:dim3-SLselection} plots the probability of each surrogate-driven design being selected over time by TMLE-OSLAD in the two scenarios. Similar to the simulation results in the main paper, TMLE-OSLAD consistently provides valid statistical inference with confidence intervals that achieve nominal coverage and selects the oracle surrogate with high frequency in both scenarios. 
\begin{table}[thb]\centering
    \caption{True target estimand $\Psi_{t,k}$ had participants enrolled from $1, \cdots, t-K$ were adaptively treated by a candidate surrogate $Y_k$ in Scenarios 1 and 2 ($t = 11,21,31,41,50$) with three covariates. Expected outcomes under the adaptive design using the oracle surrogate are shown in bold.}
    \label{tab:dim3-truePsi}
    \begin{minipage}[t]{0.48\linewidth}
        \makebox[\textwidth]{\small (a) Scenario 1} 
        \resizebox{0.96\textwidth}{!}{
        \large
        \begin{tabular}{*{10}{c}}
            \toprule
           $t$ & $\Psi_{t,RCT}$ & $\Psi_{t,1}$ &  $\Psi_{t,2}$ & $\Psi_{t,3}$ & $\Psi_{t,4}$ & $\Psi_{t,5}$ \\
            \midrule
            11 & 0.000 & -0.174 & -0.066 & 0.035 & \textbf{0.060} & 0.036 \\
            21 & 0.000 & -0.207 & -0.078 & 0.088 & 0.162 & \textbf{0.170} \\
            31 & 0.000 & -0.214 & -0.080 & 0.103 & 0.189 & \textbf{0.204} \\
            41 & 0.000 & -0.217 & -0.081 & 0.109 & 0.202 & \textbf{0.219} \\
            50 & 0.000 & -0.218 & -0.080 & 0.112 & 0.209 & \textbf{0.228} \\
            \bottomrule
        \end{tabular}
    }
\end{minipage}\hfill
    \begin{minipage}[t]{0.48\linewidth}
        \makebox[\textwidth]{\small (b) Scenario 2} 
        \resizebox{0.96\textwidth}{!}{
        \large
        \begin{tabular}{*{10}{c}}
            \toprule
           $t$ & $\Psi_{t,RCT}$ & $\Psi_{t,1}$ &  $\Psi_{t,2}$ & $\Psi_{t,3}$ & $\Psi_{t,4}$ & $\Psi_{t,5}$ \\
            \midrule
            11 & 0.000 & \textbf{0.034} & 0.025 & 0.014 & 0.005 & 0.001 \\
            21 & 0.000 & \textbf{0.044} & 0.040 & 0.033 & 0.021 & 0.008 \\
            31 & 0.000 & \textbf{0.047} & 0.044 & 0.039 & 0.028 & 0.013 \\
            41 & 0.000 & \textbf{0.048} & 0.046 & 0.042 & 0.033 & 0.016 \\
            50 & 0.000 & \textbf{0.049} & 0.047 & 0.044 & 0.035 & 0.018 \\
            \bottomrule
        \end{tabular}
        }
    \end{minipage}\hfill
\end{table}

\begin{table}[H]\centering
    \caption{Coverage probability (\%) of confidence intervals for $\Psi_{t,k}$ in Scenarios 1 and 2 ($t = 11,21,31,41,50$) with three covariates.}
    \label{tab:dim3-coverage}
    \vspace{0.5em}
    \begin{minipage}[t]{0.48\linewidth}
        \makebox[\textwidth]{\small (a) Scenario 1} 
        \resizebox{0.96\textwidth}{!}{
        \large
        \begin{tabular}{*{10}{c}}
            \toprule
           $t$ & $\Psi_{t,RCT}$ & $\Psi_{t,1}$ &  $\Psi_{t,2}$ & $\Psi_{t,3}$ & $\Psi_{t,4}$ & $\Psi_{t,5}$ \\
            \midrule
            11 & 93.4 & 95.0 & 93.6 & 93.6 & 94.0 & 94.0 \\
            21 & 92.4 & 93.4 & 92.8 & 94.0 & 93.0 & 93.6 \\
            31 & 94.6 & 93.2 & 94.0 & 95.0 & 95.4 & 94.0 \\
            41 & 94.2 & 95.2 & 94.8 & 94.4 & 94.2 & 94.6 \\
            50 & 94.2 & 93.6 & 94.0 & 93.0 & 94.0 & 94.4 \\
            \bottomrule
        \end{tabular}
    }
\end{minipage}\hfill
    \vspace{0.5em}
    \begin{minipage}[t]{0.48\linewidth}
        \makebox[\textwidth]{\small (b) Scenario 2} 
        \resizebox{0.96\textwidth}{!}{
        \large
        \begin{tabular}{*{10}{c}}
            \toprule
           $t$ & $\Psi_{t,RCT}$ & $\Psi_{t,1}$ &  $\Psi_{t,2}$ & $\Psi_{t,3}$ & $\Psi_{t,4}$ & $\Psi_{t,5}$ \\
            \midrule
            11 & 93.2 & 93.4 & 93.2 & 93.4 & 93.4 & 94.0 \\
            21 & 93.6 & 91.8 & 94.2 & 94.4 & 93.6 & 93.6 \\
            31 & 94.4 & 91.8 & 92.4 & 93.0 & 93.8 & 94.0 \\
            41 & 95.0 & 94.4 & 94.2 & 93.4 & 95.0 & 95.0 \\
            50 & 93.8 & 93.6 & 93.0 & 93.0 & 94.0 & 95.2 \\
            \bottomrule
        \end{tabular}
        }
    \end{minipage}\hfill
\end{table}

\begin{figure}[htb]
    \centering
    \begin{subfigure}[b]{0.45\textwidth}
        \centering
        \includegraphics[width=\textwidth]{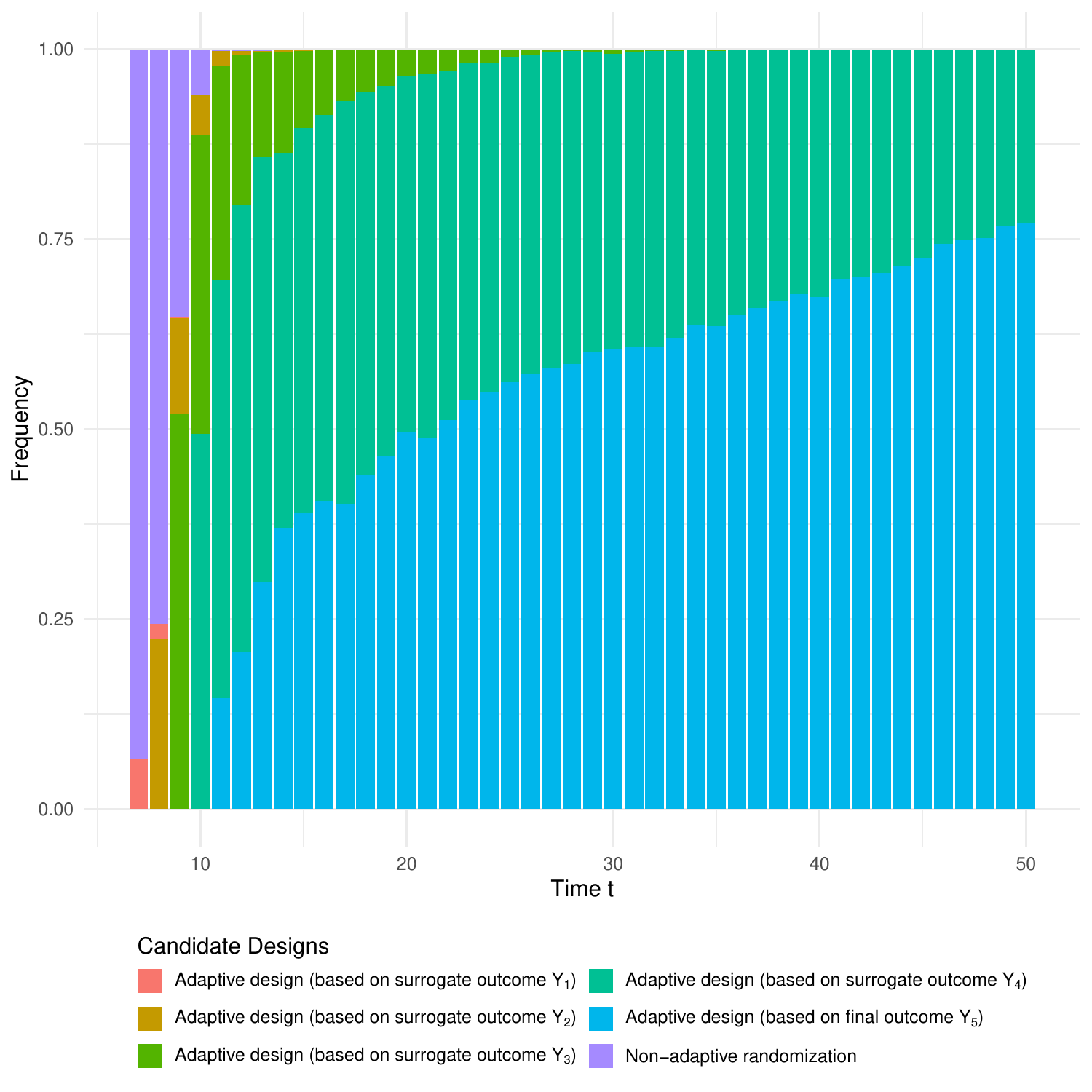}
        \caption{} 
        \label{fig:SLselectionA}
    \end{subfigure}
    \hfill
    \begin{subfigure}[b]{0.45\textwidth}
        \centering
        \includegraphics[width=\textwidth]{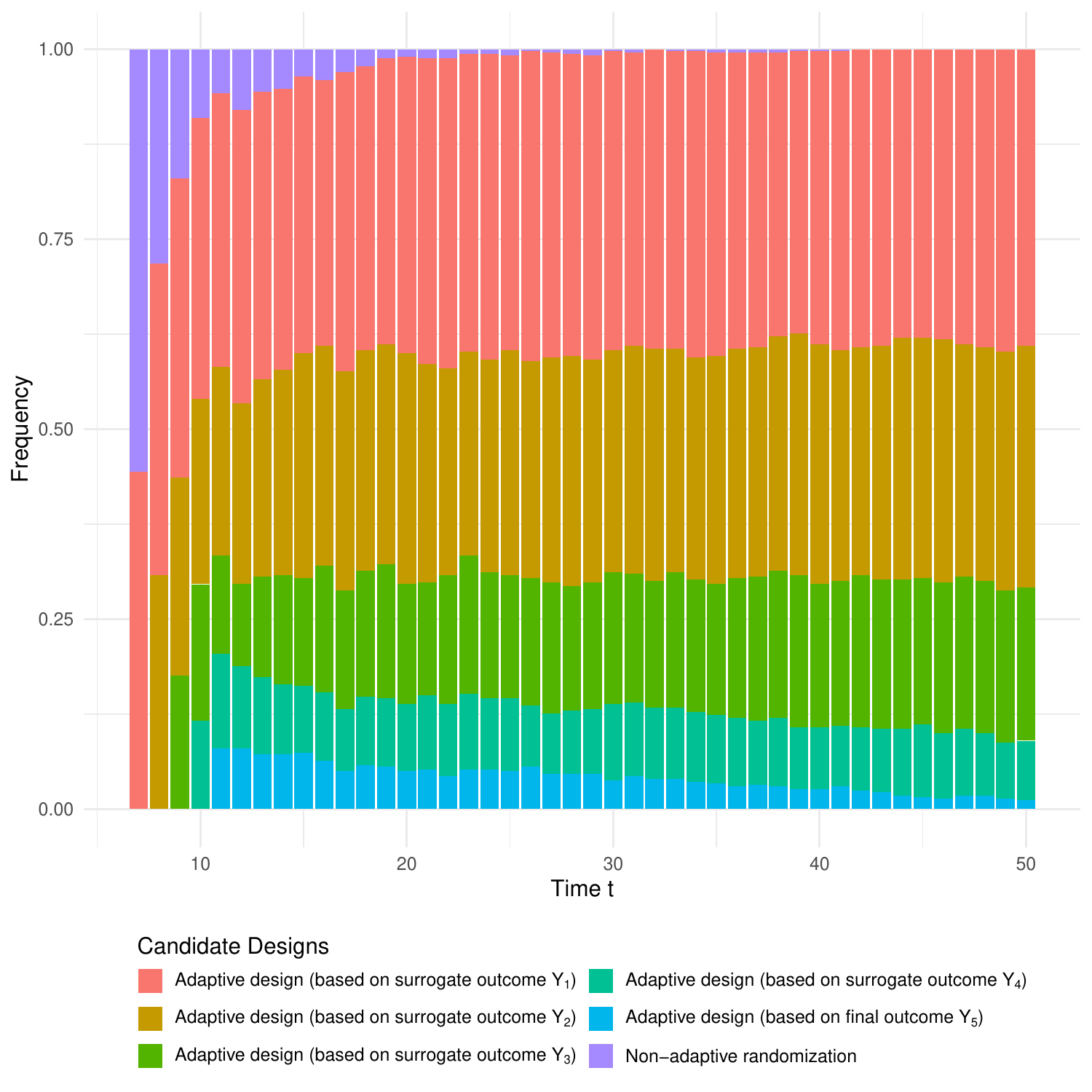}
        \caption{} 
        \label{fig:SLselectionB}
    \end{subfigure}
    \caption{Frequency of candidate adaptive design selected by TMLE-OSLAD at each time point. Panels (a) and (b) correspond to Scenario 1 and Scenario 2 with three covariates.}\label{fig:dim3-SLselection}
\end{figure}

\par
\endgroup

\section{Additional Results for Section 
\ref{section:real_sim}}
\label{appendix:real_sim}
In this section, we present additional results in ADAPT-R simulation.
Figure \ref{fig:CATE-ADAPT-R} shows the true CATE functions of Navigator vs. SMS+CCT in ADAPT-R simulation. 
For participants initially assigned to SMS, the optimal personalized treatment varies with time-to-lapse after initial treatment and shows different patterns across outcomes. For earlier outcomes ($Y_1$ through $Y_3$), the CATE suggests that Navigator is preferable with shorter time-to-lapse, but the benefit wanes over time, shifting the optimal treatment to SMS+CCT. In contrast, for later outcomes ($Y_4$ and $Y_5$), SMS+CCT is preferred for short time-to-lapse, with their effectiveness decreasing over time, making Navigator the optimal choice for longer time-to-lapse.
For participants whose initial treatment was SOC, the treatment effects are minimal: CATE remains close to zero over the entire range of time-to-lapse periods, indicating small differences between Navigator and SMS+CCT.
For participants with CCT as their initial strategy, Navigator is preferred for participants with short time-to-lapse periods, showing a positive CATE. However, as time-to-lapse increases, SMS+CCT becomes increasingly favored, with the CATE decreasing and turning negative.
\begin{figure}[tbh]
    \centering \includegraphics[width=0.7\textwidth]{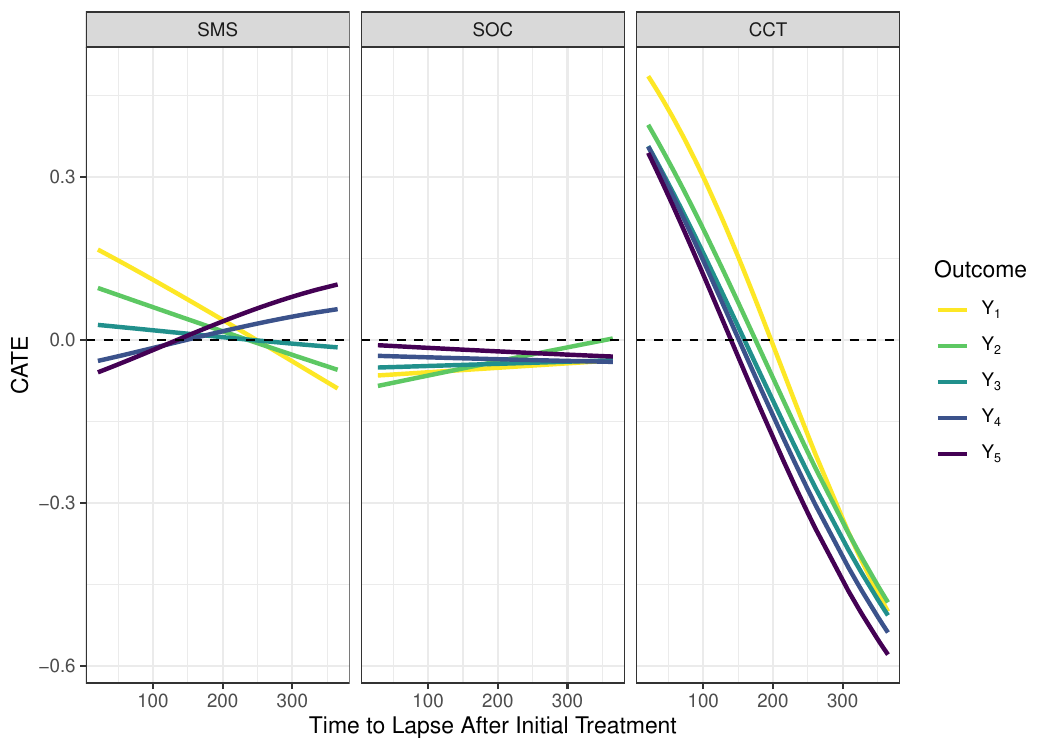}
    \caption{Estimated Conditional Average Treatment Effect (CATE) functions of Navigator v.s. SMS+CCT in Phase 2, plotted against initial retention strategy (grids) and time to lapse after initial strategy in Phase 1 (x-axis).  Each curve corresponds to the CATE function of a candidate surrogate, ranging from $Y_1$ to $Y_5$.}
    \label{fig:CATE-ADAPT-R}
\end{figure}

Figure \ref{fig:realSim_abs_regret} illustrates regret over time obtained from implementing different design strategies.
Table \ref{tab:realSim_truePsi_and_coverage} shows the true values of the proposed adaptive design target estimand $\Psi_{t,k}$ at the illustrative time points 11, 13, 15, 17, and 19, along with the coverage probabilities of the confidence intervals based on TMLE.
Table \ref{tab:realSim_table_bias_variance} reports the bias and variance of TMLE at these time points.
Table \ref{tab:realSim_estimation_adaptive_design_at_end} presents TMLE's performance in estimating the proposed adaptive design estimands at the end of the experiment.
Table \ref{tab:realSim_estimation_at_end} presents TMLE's performance of estimating ATE on final outcome and mean final outcome of optimal dynamic treatment rules based on surrogate outcomes and final outcomes after the experiment ends.
Table \ref{tab:realSim_truePsi_marginal} shows the true expected final outcome $\tPsi_{t,k}$ of implementing different designs in ADAPT-R simulation.

\begin{figure}[H]
    \centering
    \includegraphics[width=0.9\linewidth]{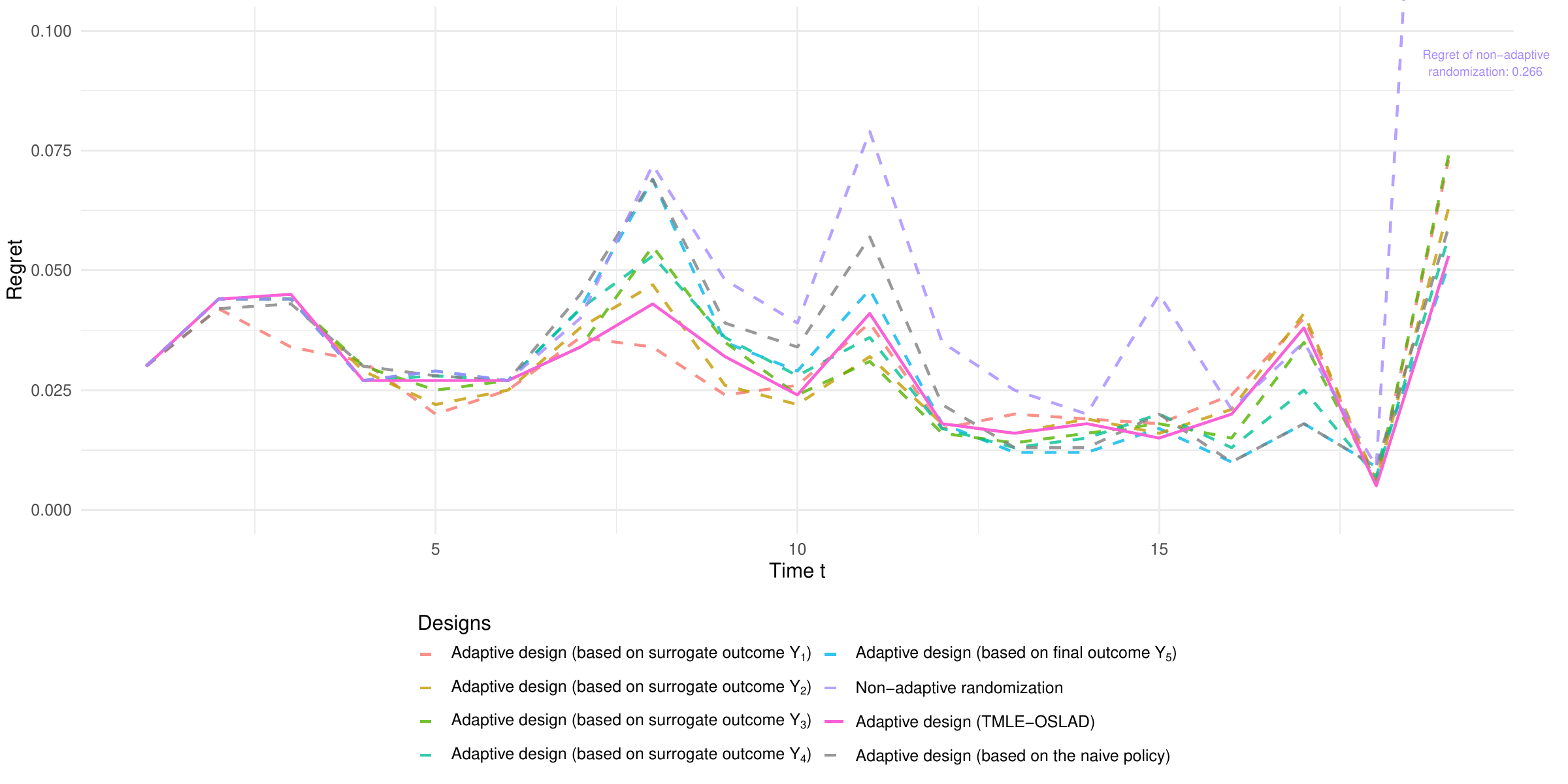}
    \caption{Regret at each time point $t$ obtained from implementing different design strategies in the ADAPT-R trial simulation. 
    }    \label{fig:realSim_abs_regret}
    \end{figure}
    
\begin{table}[H]\centering
    \caption{(a) True target estimand $\Psi_{t,k}$ had participants enrolled from $1, \cdots, t-K$ been adaptively treated by a candidate adaptive design using surrogate $Y_k$ in ADAPT-R trial simulations ($t = 11,13,15,17,19$). (b) Coverage probability (\%) of confidence intervals for $\Psi_{t,k}$ in ADAPT-R trial simulations ($t = 11,13,15,17,19$)}
    \label{tab:realSim_truePsi_and_coverage}
    \vspace{0.5em}
    \begin{minipage}[t]{0.48\linewidth}
        \makebox[\textwidth]{\small (a) $\Psi_{t,k}$} 
        \resizebox{0.96\textwidth}{!}{
        \large
        \begin{tabular}{*{7}{p{1.2cm}}}
            \toprule
           $t$ & $\Psi_{t,RCT}$ & $\Psi_{t,1}$ &  $\Psi_{t,2}$ & $\Psi_{t,3}$ & $\Psi_{t,4}$ & $\Psi_{t,5}$ \\
            \midrule
            11 & 0.538 & 0.542 & 0.540 & 0.538 & 0.538 & 0.538 \\
            13 & 0.547 & 0.559 & 0.555 & 0.552 & 0.549 & 0.547 \\
            15 & 0.557 & 0.572 & 0.569 & 0.566 & 0.563 & 0.560 \\
            17 & 0.562 & 0.579 & 0.577 & 0.574 & 0.571 & 0.568 \\
            19 & 0.564 & 0.580 & 0.578 & 0.576 & 0.573 & 0.570 \\
            \bottomrule
        \end{tabular}
    }
\end{minipage}\hfill
    \vspace{0.5em}
    \begin{minipage}[t]{0.48\linewidth}
        \makebox[\textwidth]{\small (b) Coverage probability (\%)} 
        \resizebox{0.96\textwidth}{!}{
        \large
        \begin{tabular}{*{7}{p{1.2cm}}}
            \toprule
           $t$ & $\Psi_{t,RCT}$ & $\Psi_{t,1}$ &  $\Psi_{t,2}$ & $\Psi_{t,3}$ & $\Psi_{t,4}$ & $\Psi_{t,5}$ \\
            \midrule
            11 & 94.8 & 95.6 & 95.8 & 95.0 & 94.6 & 94.8 \\
            13 & 94.2 & 95.8 & 95.8 & 93.0 & 93.6 & 94.2 \\
            15 & 92.2 & 96.2 & 95.4 & 95.4 & 92.4 & 93.2 \\
            17 & 93.8 & 96.8 & 95.6 & 95.2 & 94.2 & 95.6 \\
            19 & 93.2 & 96.8 & 95.2 & 95.4 & 95.2 & 95.6 \\
           \bottomrule
        \end{tabular}
        }
    \end{minipage}\hfill
\end{table}

\begin{table}[htb]\centering
    \caption{Bias and variance of TMLE estimators for $\Psi_{t,k}$ in ADAPT-R trial simulations ($t = 11,13,15,17,19$).}
    \label{tab:realSim_table_bias_variance}
    \vspace{0.5em}
    \begin{minipage}[t]{0.48\linewidth}
        \makebox[\textwidth]{\small (a) Bias ($\times 10^{-3}$)} 
        \resizebox{0.96\textwidth}{!}{
        \large
        \begin{tabular}{*{7}{p{1.2cm}}}
            \toprule
           $t$ & $\Psi_{t,RCT}$ & $\Psi_{t,1}$ &  $\Psi_{t,2}$ & $\Psi_{t,3}$ & $\Psi_{t,4}$ & $\Psi_{t,5}$ \\
            \midrule    
            11 & -0.50 & -0.40 & -0.30 &  0.22 & -0.54 & -0.31 \\
            13 &  0.45 & -0.16 &  0.18 &  0.88 &  0.36 &  0.33 \\
            15 &  0.56 & -0.11 &  0.29 &  0.52 &  0.15 &  0.10 \\
            17 &  0.29 & -0.02 &  0.09 &  0.37 & -0.16 & -0.12 \\
            19 &  0.16 & -0.08 & -0.04 &  0.23 & -0.22 & -0.12 \\
            \bottomrule
            \end{tabular}
    }
\end{minipage}\hfill
    \vspace{0.5em}
    \begin{minipage}[t]{0.48\linewidth}
        \makebox[\textwidth]{\small (b) Variance ($\times 10^{-3}$)} 
        \resizebox{0.96\textwidth}{!}{
        \large
        \begin{tabular}{*{7}{p{1.2cm}}}
            \toprule
           $t$ & $\Psi_{t,RCT}$ & $\Psi_{t,1}$ & $\Psi_{t,2}$ & $\Psi_{t,3}$ & $\Psi_{t,4}$ & $\Psi_{t,5}$ \\
            \midrule
            11 & 0.16 & 0.22 & 0.22 & 0.20 & 0.20 & 0.17 \\
            13 & 0.16 & 0.16 & 0.19 & 0.19 & 0.20 & 0.22 \\
            15 & 0.14 & 0.13 & 0.15 & 0.16 & 0.18 & 0.19 \\
            17 & 0.12 & 0.10 & 0.12 & 0.13 & 0.14 & 0.17 \\
            19 & 0.11 & 0.10 & 0.11 & 0.12 & 0.13 & 0.16 \\
            \bottomrule
        \end{tabular}
        }
    \end{minipage}\hfill
\end{table}
\begin{table}[H]\centering
    \caption{Performance of TMLE estimators for adaptive-design estimands at the end of experiments in ADAPT-R trial simulation.}
    \label{tab:realSim_estimation_adaptive_design_at_end}
    \vspace{0.5em}
    \makebox[\textwidth]{\small} 
    \resizebox{0.96\textwidth}{!}{
    \large
    \begin{tabular}{*{10}{c}}
        \toprule
        Estimand & Notation & Truth & Bias ($\times 10^{-3}$) & Var ($\times 10^{-3}$) & Coverage  (\%)\\
        \midrule
        Expected final outcome $Y_5$ under the adaptive design using $Y_1$ & $\Psi_{\Tend,1}$ & 0.581 & -0.11 & 0.09 & 96.2 \\
        Expected final outcome $Y_5$ under the adaptive design using $Y_2$ & $\Psi_{\Tend,2}$ & 0.580 & -0.05 & 0.10 & 95.8 \\
        Expected final outcome $Y_5$ under the adaptive design using $Y_3$ & $\Psi_{\Tend,3}$ & 0.578 & 0.29 & 0.11 & 95.0 \\
        Expected final outcome $Y_5$ under the adaptive design using $Y_4$ & $\Psi_{\Tend,4}$ & 0.576 & -0.17 & 0.13 & 93.8 \\
        Expected final outcome $Y_5$ under the adaptive design using $Y_5$ & $\Psi_{\Tend,5}$ & 0.573 & -0.15 & 0.15 & 95.0 \\ \bottomrule
    \end{tabular}
    }
\end{table}

\begin{table}[htb]\centering
    \caption{Performance of TMLE estimators for other causal estimands at the end of experiments in ADAPT-R trial simulation.}
    \label{tab:realSim_estimation_at_end}
    \vspace{0.5em}

    \makebox[\textwidth]{\small} 
    \resizebox{0.96\textwidth}{!}{
    \renewcommand{\arraystretch}{1.2}
    \large
    \begin{tabular}{*{10}{c}}
        \toprule
        Estimand & Notation & Truth & Bias ($\times 10^{-3}$) & Var ($\times 10^{-3}$) & Coverage  (\%)\\
        \midrule
        Average treatment effect on the final outcome $Y_5$ & $\psi_{\Tend}^{\text{ATE}}$ & -0.015 & 0.00 & 0.46 & 93.6 \\
        Expected final outcome $Y_5$ under the estimated ODTR of $Y_1$ & $\psi^{d^*_{n,1}}_{\Tend}$ & 0.590 & 0.98 & 0.12 & 95.0 \\
        Expected final outcome $Y_5$ under the estimated ODTR of $Y_2$ & $\psi^{d^*_{n,2}}_{\Tend}$ & 0.594 & 0.07 & 0.13 & 95.4 \\
        Expected final outcome $Y_5$ under the estimated ODTR of $Y_3$ & $\psi^{d^*_{n,3}}_{\Tend}$ & 0.595 & 1.87 & 0.15 & 91.2 \\
        Expected final outcome $Y_5$ under the estimated ODTR of $Y_4$ & $\psi^{d^*_{n,4}}_{\Tend}$ & 0.598 & 1.73 & 0.16 & 93.6 \\
        Expected final outcome $Y_5$ under the estimated ODTR of $Y_5$ & $\psi^{d^*_{n,5}}_{\Tend}$ & 0.597 & 1.13 & 0.20 & 93.2 \\ \bottomrule
    \end{tabular}
    }
\end{table}

\begin{table}[H]\centering
    \caption{True expected final outcome $\tPsi_{t,k}$ of implementing different designs in ADAPT-R trial simulation ($t = 11,13,15,17,19,T_{end}$).}
    \label{tab:realSim_truePsi_marginal}
    \scriptsize
    \vspace{0.5em}

    \resizebox{0.6\textwidth}{!}{
    \large
    \begin{tabular}{*{10}{c}}
    \toprule
    t & $\tPsi_{t,RCT}$ & $\tPsi_{t,1}$ & $\tPsi_{t,2}$ & $\tPsi_{t,3}$ & $\tPsi_{t,4}$ & $\tPsi_{t,5}$ & $\tPsi_{t,naive}$ & $\tPsi_{t,SL}$ \\
    \midrule
    11 & 0.538 & 0.541 & 0.540 & 0.538 & 0.538 & 0.538 & 0.538 & 0.538 \\
    13 & 0.547 & 0.555 & 0.552 & 0.551 & 0.549 & 0.547 & 0.546 & 0.552 \\
    15 & 0.557 & 0.569 & 0.566 & 0.564 & 0.562 & 0.560 & 0.559 & 0.565 \\
    17 & 0.562 & 0.576 & 0.574 & 0.572 & 0.570 & 0.568 & 0.566 & 0.573 \\
    19 & 0.564 & 0.577 & 0.576 & 0.574 & 0.572 & 0.570 & 0.568 & 0.574 \\
    $T_{end}$ & 0.565 & 0.579 & 0.578 & 0.576 & 0.575 & 0.573 & 0.570 & 0.576 \\
    \bottomrule
    \end{tabular}
    }
\end{table}

\section{An Extension of the Proposed CARA Design}
\label{appendix:extended_CARA}
In some applications, when a plausible surrogate is available, experimenters may wish to adopt surrogate-guided adaptations before formally validating their utility for optimizing the final outcome.
In such cases, a variant of TMLE-OSLAD can be used by designating a plausible surrogate (e.g., $Y_k$) as a primary outcome.
The design then treats $Y_k$ as the working primary outcome and updates treatment randomization probabilities as early as $t=k+1$.
As the experiment progresses, later surrogates can be promoted in sequence as primary outcomes until the true final outcome $Y_K$ is observed.
If prior knowledge or external data may indicate the relationship between surrogates and the primary
outcome, one may also initiate an earlier adaptation based on that information. However, these approaches typically rely on generalizability assumptions or pre-specified models of surrogate and primary outcomes.
Regardless of the strategy chosen, the TMLE-OSLAD framework can be invoked throughout the experiment to evaluate the opportunities and costs of each candidate adaptive design and dynamically select the best-performing one to support real-time decision making.


\end{document}